\documentclass[aps,pra,10pt,twocolumn,superscriptaddress,showkeys,letterpaper]{revtex4-2}

\pdfoutput=1

\usepackage[T1]{fontenc}
\usepackage[utf8]{inputenc}
\usepackage{color}
\usepackage{xcolor}
\usepackage{babel}
\usepackage{verbatim}
\usepackage{mathtools}
\usepackage{amsmath}
\usepackage{amsthm}
\usepackage{amssymb}
\usepackage[normalem]{ulem}
\usepackage[pdfusetitle,
 bookmarks=true,bookmarksnumbered=false,bookmarksopen=false,
 breaklinks=false,pdfborder={0 0 0},pdfborderstyle={},backref=false,colorlinks=true]
 {hyperref}
\hypersetup{
 linkcolor=magenta, urlcolor=blue, citecolor=red}

\makeatletter
\theoremstyle{plain}
\newtheorem{thm}{\protect\theoremname}
\theoremstyle{plain}
\newtheorem{lyxalgorithm}[thm]{\protect\algorithmname}
\theoremstyle{definition}
\newtheorem{defn}[thm]{\protect\definitionname}
\theoremstyle{definition}
\newtheorem{example}[thm]{\protect\examplename}
\theoremstyle{plain}
\newtheorem{prop}[thm]{\protect\propositionname}
\theoremstyle{plain}
\newtheorem{cor}[thm]{\protect\corollaryname}
\theoremstyle{remark}
\newtheorem{rem}[thm]{\protect\remarkname}

\DeclareMathOperator{\Tr}{Tr}
\DeclareMathOperator{\tot}{tot}
\DeclareMathOperator{\SPAN}{span}

\DeclareMathOperator{\arctanh}{arctanh}

\allowdisplaybreaks

\makeatother

\providecommand{\algorithmname}{Algorithm}
\providecommand{\corollaryname}{Corollary}
\providecommand{\definitionname}{Definition}
\providecommand{\examplename}{Example}
\providecommand{\propositionname}{Proposition}
\providecommand{\remarkname}{Remark}
\providecommand{\theoremname}{Theorem}

\begin{document}

\title{Constrained free energy minimization for the design of thermal states
and\\stabilizer thermodynamic systems}

\author{Michele Minervini}
\affiliation{Institute of Physics, {\'E}cole Polytechnique F\'ed\'erale de Lausanne (EPFL), Lausanne, Switzerland}
\affiliation{School of Electrical and Computer Engineering, Cornell University, Ithaca, New
York 14850, USA}

\author{Madison Chin}
\affiliation{School of Electrical and Computer Engineering, Cornell University, Ithaca, New
York 14850, USA}

\author{Jacob Kupperman}
\affiliation{School of Applied and Engineering Physics, Cornell University, Ithaca, New
York 14850, USA}

\author{Nana Liu}
\affiliation{Institute of Natural Sciences, Shanghai Jiao Tong University, Shanghai 200240, China}
\affiliation{School of Mathematical Sciences, Shanghai Jiao Tong University, Shanghai
200240, China}

\author{Ivy Luo}
\affiliation{School of Electrical and Computer Engineering, Cornell University, Ithaca, New
York 14850, USA}

\author{Meghan Ly}
\affiliation{School of Electrical and Computer Engineering, Cornell University, Ithaca, New
York 14850, USA}

\author{Soorya Rethinasamy}
\affiliation{School of Applied and Engineering Physics, Cornell University, Ithaca, New
York 14850, USA}

\author{Kathie Wang}
\affiliation{School of Electrical and Computer Engineering, Cornell University, Ithaca, New
York 14850, USA}

\author{Mark M. Wilde}
\affiliation{School of Electrical and Computer Engineering, Cornell University, Ithaca, New
York 14850, USA}

\date{\today}

\begin{abstract}
A quantum thermodynamic system is described by a Hamiltonian and a
list of conserved, non-commuting charges, and a fundamental goal is
to determine the minimum energy of the system subject to constraints
on the charges. Recently, {[}Liu \textit{et al}., arXiv:2505.04514{]}~proposed
first- and second-order classical and hybrid quantum-classical algorithms
for solving a dual chemical potential maximization problem, and they
proved that these algorithms converge to global optima by means of
gradient-ascent approaches. In this paper, we benchmark these algorithms
on several problems of interest in thermodynamics, including one-
and two-dimensional quantum Heisenberg models with nearest- and next-nearest
neighbor interactions and with the charges set to the total $x$,
$y$, and $z$ magnetizations. We also offer an alternative compelling
interpretation of these algorithms as methods for designing ground
and thermal states of controllable Hamiltonians, with potential applications
in molecular and material design. Furthermore, we introduce stabilizer
thermodynamic systems as thermodynamic systems based on stabilizer
codes, with the Hamiltonian constructed from a given code's stabilizer
operators and the charges constructed from the code's logical operators.
We benchmark the aforementioned algorithms on several examples of
stabilizer thermodynamic systems, including those constructed from
the one-to-three-qubit repetition code, the perfect one-to-five-qubit
code, and the two-to-four-qubit error-detecting code. Finally, we
observe that the aforementioned hybrid quantum-classical algorithms,
when applied to stabilizer thermodynamic systems, can serve as alternative
methods for encoding quantum information into stabilizer codes at
a fixed temperature, and we provide an effective method for warm-starting
these encoding algorithms whenever a single qubit is encoded into
multiple physical qubits.

\end{abstract}

\maketitle

\tableofcontents

\section{Introduction}

\subsection{Background and motivation}

\label{subsec:Background-and-motivation}Quantum computation has the
potential to cause a paradigm shift in several areas of computer science.
After the initial proposals of quantum speedups for factoring~\cite{Shor1994}
and unstructured search~\cite{Grover1996}, it was postulated that
quantum computation could offer speedups for optimization, with applications
to physics, combinatorial optimization, machine learning, and beyond
(see, e.g.,~\cite{Abbas2024} for a recent review). While this possibility
still remains the subject of ongoing research, concrete quantum algorithms
for solving semidefinite optimization problems have been proposed
starting nearly a decade ago~\cite{Brandao2017,Apeldoorn2019,Brandao2019,vanApeldoorn2020quantumsdpsolvers,Kerenidis2020,Bharti2022,Augustino2023quantuminterior,watts2023quantum,Patti2023quantumgoemans,Patel2024variationalquantum,westerheim2023dualvqequantumalgorithmlower,chen2023qslackslackvariableapproachvariational,liu2025qthermoSDPs}.
Some of these algorithms have guaranteed runtimes~\cite{Brandao2017,Apeldoorn2019,Brandao2019,vanApeldoorn2020quantumsdpsolvers,Kerenidis2020,Augustino2023quantuminterior,watts2023quantum,liu2025qthermoSDPs},
while others are heuristic in nature~\cite{Bharti2022,Patti2023quantumgoemans,Patel2024variationalquantum,westerheim2023dualvqequantumalgorithmlower,chen2023qslackslackvariableapproachvariational},
falling under the umbrella of variational quantum algorithms~\cite{Cerezo2021vqa}. 

Liu \textit{et al}.~recently identified a concrete link between quantum
thermodynamics and semidefinite optimization~\cite{liu2025qthermoSDPs},
having the benefit of bridging two previously unlinked foundational
areas and bringing insights from quantum thermodynamics into the analysis
of semidefinite optimization. As identified in~\cite{liu2025qthermoSDPs},
a special class of semidefinite optimization problems, known as constrained
energy minimization problems, are relevant in quantum thermodynamics.
In particular, these problems are defined in terms of a Hamiltonian
and a list of conserved, non-commuting charges, with the latter imposing
constraints on the state of the physical system. The goal then is
to minimize the energy of the system, while respecting the imposed
constraints. 

Generally speaking, determining the minimum energy of
a quantum physical system is a primary goal in physics~\cite{Lieb2005},
typically being the first step employed in computing energetic properties
of molecules and materials, and thus having applications in materials
science~\cite{Steinhauser2009}, condensed-matter physics~\cite{Continentino2021},
and quantum chemistry~\cite{Deglmann2015}. Going beyond this application,
and as discussed in this paper, constrained energy minimization has
applications in the design of molecules and materials having desired
properties. In particular, the ability to enforce constraints on conserved quantities (such as particle number or spin) allows for the precise targeting of specific electronic symmetry sectors -- a capability that is often challenging for standard unconstrained variational approaches which may collapse to incorrect sectors (e.g., neutral species instead of ions)~\cite{Ryabinkin2018ConstrainedVQ, McClean2016, Rubin2018}.

Beyond establishing a link between quantum thermodynamics and semidefinite
optimization, Ref.~\cite{liu2025qthermoSDPs} proposed various classical
and hybrid quantum-classical (HQC) algorithms for solving semidefinite
programs, being motivated by physical intuition coming from quantum
thermodynamics when applied to the special class of constrained energy
minimization problems. These algorithms, hereafter referred to as
the LMPW algorithms, have provable guarantees for their runtimes
and are based on the observation that parameterized thermal states,
also known as non-Abelian thermal states~\cite{YungerHalpern2016,YungerHalpern2020,Majidy2023}
or quantum Boltzmann machines~\cite{Amin2018,Benedetti2017,Kieferova2017},
are optimal for solving constrained free-energy minimization problems,
which closely approximate the original constrained energy minimization
problems for sufficiently low temperatures. Indeed, Ref.~\cite{liu2025qthermoSDPs}
established that the optimization landscapes of dual chemical potential
maximization problems are concave in the parameters of the thermal
states, so that conventional optimization techniques like gradient
ascent, stochastic gradient ascent, and their variants~\cite{Bubeck2015}
are guaranteed to converge. The HQC algorithms are variational in
nature, using a quantum computer for two purposes only, which include
the preparation of parameterized thermal states and estimating expectations
of observables (i.e., the Hamiltonian and non-commuting charges).
A classical computer is then employed for storing and updating the
parameters of the parameterized thermal states as the HQC algorithm
progresses toward convergence. The LMPW HQC algorithms make use
of first-derivative information to determine the next step in a search
and can additionally incorporate second-derivative information to
improve their convergence to globally optimal solutions.

Ref.~\cite{liu2025qthermoSDPs} also proved that the HQC algorithms
are sample efficient, meaning that they require a number of thermal-state
preparations polynomial in the number of qubits in order to converge.
However, for the HQC algorithms to have efficient overall runtimes,
and not merely efficient sample complexity, the Hamiltonian and charges
should be such that it is possible to efficiently prepare their corresponding
low-temperature parameterized thermal states on a quantum computer.
While there has been much theoretical progress on this topic in recent
years~\cite{chen2023q_Gibbs_sampl,chen2023thermalstatepreparation,rajakumar2024gibbssampling,bergamaschi2024gibbs_sampling,chen2024sim_Lindblad,rouze2024efficientthermalization,bakshi2024hightemperaturegibbsstates,ding2024preparationlowtemperaturegibbs},
it remains largely open to design quantum algorithms that can prepare
the needed thermal states with temperature sufficiently low for accurately
solving constrained energy minimization problems.

Let us also recall an observation from~\cite{liu2025qthermoSDPs},
that the LMPW HQC algorithms for solving constrained free energy
minimization problems can alternatively be viewed as HQC algorithms
for preparing non-Abelian thermal states of thermodynamic systems
with conserved, non-commuting charges. This aspect of the LMPW HQC
algorithms could be useful in fields like materials science, condensed
matter physics, and quantum chemistry, because one expects such states
to arise in these thermodynamic systems and, after preparing them,
one could subsequently measure various observables of interest on
the generated states, in order to determine their properties. 

This capability is particularly timely, as recent developments in quantum thermodynamics have highlighted the rich physics of systems with conserved, non-commuting charges~\cite{campbell2025roadmapquantumthermodynamics, Majidy2023}. Unlike standard commuting charges, non-commutation relations can fundamentally alter many-body phenomena, leading to effects such as the modification of entropy production~\cite{Manzano_2022} and the enhancement of entanglement entropy~\cite{Majidy_2023increase_entropy}. Furthermore, these charges play a critical role in thermalization dynamics; for instance, they can invalidate the standard Eigenstate Thermalization Hypothesis (ETH)~\cite{Murthy_2023} or prevent certain forms of thermalization entirely~\cite{Majidy_2024}. However, while the dynamical consequences of non-commuting charges have been actively explored, the rigorous construction and calculation of the resulting equilibrium states remains a computational challenge. Our work addresses this gap by providing a suite of algorithms to explicitly construct these states via free energy minimization, offering a tool to verify the equilibrium endpoints predicted by these thermodynamic theories.

\subsection{Summary of contributions}

\label{subsec:Summary-of-contributions}In this paper, we apply the
LMPW algorithms to concrete examples of constrained energy minimization
problems of interest in quantum thermodynamics, including the one-
and two-dimensional Heisenberg models~\cite{Goldschmidt2011} with
nearest and next-nearest neighbor interactions, and with the non-commuting
charges set to the total magnetizations in the $x$, $y$, and $z$
directions. Such models describe interacting spin systems on a lattice
and are essential for understanding magnetic materials~\cite{Mattis2006}.

As a first conceptual contribution, we offer an alternative compelling
interpretation of the LMPW algorithms as methods for designing ground
and thermal states of controllable Hamiltonians, such that these states
satisfy desired properties (i.e., constraints on the expectation values
of certain observables). This application of the algorithms could
find use in the design of molecules and materials with desired properties. 

As a key theoretical contribution, our paper also establishes a bridge
between quantum error correction~\cite{Gottesman2009,Devitt2013,Lidar2013,Roffe2019,Bradshaw2025}
and conserved, non-commuting charges in quantum thermodynamics, by
observing that every stabilizer code~\cite{Gottesman1996,Calderbank1997}
can be understood as corresponding to a thermodynamic system with
conserved, non-commuting charges. In particular, we set the Hamiltonian
to be in correspondence with the stabilizer operators of the code,
while the non-commuting charges are in correspondence with the logical
operators of the code. We refer to such a thermodynamic system as
a \textit{stabilizer thermodynamic system}. 

Let us note that stabilizer
Hamiltonians were proposed in~\cite{Kitaev2003} and have been investigated
substantially since then (see, e.g.,~\cite{Bacon2001,Dennis2002,Jordan2006,Bacon2008,Yoshida2011,Young2012,Temme2015,Temme2017,Weinstein2019}).
However, interpreting logical operators in terms of
conserved, non-commuting charges in
a thermodynamic system has not been considered previously, to the
best of our knowledge. While the aforementioned previous works largely studied the passive properties of stabilizer Hamiltonians (e.g., thermal stability and memory lifetime~\cite{Dennis2002, Bacon2006, Yoshida2011feasibility, Weinstein2019}), our work utilizes stabilizer thermodynamic systems for active state preparation. By formally treating the logical operators as non-commuting charges with associated chemical potentials within a thermodynamic ensemble, we
demonstrate that the LMPW
HQC algorithms can serve as thermodynamics-inspired HQC algorithms
for encoding quantum information into quantum error-correcting codes,
which is rather different from the standard approach that employs
an encoding circuit to do so (see, e.g., \cite[Chapter~4]{Gottesman1997}, \cite{GrasslRoettelerBeth2003}, 
and \cite[Section 3.5]{Wilde2008}). The approach is instead based
on constrained energy minimization and the fact that the LMPW HQC
algorithms, as applied in this context, prepare low-temperature thermal
states of stabilizer thermodynamic systems. We envision that this approach
will be useful in experimental scenarios for which it is accessible to prepare
low-temperature thermal states of stabilizer thermodynamic systems.

As further examples of constrained energy minimization, and building
on the aforementioned bridge, we also benchmark the LMPW algorithms
on stabilizer thermodynamic systems constructed from the one-to-three-qubit
repetition code~\cite{Peres1985}, the perfect one-to-five-qubit code
\cite{Laflamme1996,Bennett1996}, and the two-to-four qubit error-detecting
code~\cite{Vaidman1996,Grassl1997}. For each case, we demonstrate
that the LMPW HQC algorithms are capable of steering the state of
the stabilizer thermodynamic  system in such a way that any desired
quantum state is encoded into the code's logical qubits approximately,
with the approximation improving as the temperature of the system
decreases. 

On the one hand, the numerical contributions of our paper can be viewed
as complementing the theoretical developments of~\cite{liu2025qthermoSDPs},
by providing a comprehensive numerical study of the performance of
the LMPW algorithms for constrained energy minimization problems.
The value of this contribution is that it offers a practical assessment
of the theoretical guarantees from~\cite{liu2025qthermoSDPs} for
various problems of interest in quantum thermodynamics. On the other
hand, the first conceptual contribution of our paper, i.e., the design
of ground and thermal states, could have far reaching applications
for the use of quantum computers in molecular and material design.
Additionally, the other theoretical contribution of our paper, i.e.,
establishing a bridge between quantum error correction and conserved,
non-commuting charges in quantum thermodynamics via the notion of
stabilizer thermodynamic  systems, is conceptually enlightening and
the related encoding algorithm could be useful as an alternative approach
for stabilizing quantum information.

\subsection{Overview of numerical simulations}

\label{subsec:Overview-of-numerical}Before proceeding with the details
of our paper, let us first elaborate on some aspects of the numerical simulations
that we conducted. Given that all of the constrained energy minimization
problems we consider can be formulated as semidefinite programs (SDPs)
\cite[Section 4.6.2]{Boyd2004}, the performance of the LMPW algorithms
can be benchmarked against the values returned by standard SDP solvers.
These solvers are based on the interior-point method~\cite{Alizadeh1998,Todd1998,Jiang2020}
and are so reliable that we can view their outputs as essentially
the true optimal values of the constrained energy minimization problems.
We test out four variants of the LMPW algorithms, which include
a first-order classical algorithm, a second-order classical algorithm,
a first-order HQC algorithm, and a second-order HQC algorithm. Considering
the classical algorithms in addition to the HQC algorithms is useful,
as the classical algorithms can be understood as fully noiseless simulators
of the HQC algorithms, such that the only kind of error involved when
using them is the algorithmic error and not any other kind of error,
e.g., shot noise arising from estimating expectations of observables.
For this reason, we can benchmark the HQC algorithms against the classical
algorithms, in addition to the true values generated by standard SDP
solvers.

We conducted all of our simulations for small system sizes, between
three to six qubits. These sizes are reasonable for providing a
sense of the performance of the LMPW algorithms, while taking a
reasonable amount of time for them to converge when conducted on readily
available laptop computers. In order to speed up the second-order HQC algorithms, which require estimating the Hessian, we parallelized the implementation of them in order to use multiple processor cores. All of the software code needed to reproduce
the results of our paper is publicly available as a zenodo repository \url{https://zenodo.org/records/17285628}. We leave it to future research
to scale up the examples to much larger numbers of qubits, possibly
by using tensor-network methods~\cite{Alhambra2021,Lu2025,Cocchiarella2025},
and to conduct the simulations on actual quantum computers, incorporating
techniques like error correction~\cite{Lidar2013} and mitigation
\cite{Cai2023} in order to reduce the effects of noise.

\subsubsection{Method for thermal state preparation}

\label{subsec:Methods-for-thermal}As mentioned in Section~\ref{subsec:Background-and-motivation},
a critical subroutine used by the first- and second-order HQC algorithms
is thermal state preparation, and there has been much development
on this topic in recent years~\cite{chen2023q_Gibbs_sampl,chen2023thermalstatepreparation,rajakumar2024gibbssampling,bergamaschi2024gibbs_sampling,chen2024sim_Lindblad,rouze2024efficientthermalization,bakshi2024hightemperaturegibbsstates,ding2024preparationlowtemperaturegibbs}.
However, in our paper, we do not make use of these more recent advanced
methods. Instead, we employ a simple method for thermal state preparation, consisting of feeding the numerical
values of the desired thermal state to Qiskit's DensityMatrix routine
\cite{JavadiAbhari2024} and employing its functionality for generic
state preparation when given a description of the state. Although
this method is elementary, it has value in demonstrating the performance
of the algorithms under the assumption that a robust method for thermal
state preparation exists. Indeed, the literature on thermal-state
preparation algorithms is rather active and advances are being made
rapidly, such that we expect continued progress in this direction
during the near term. As such, any of these alternative thermal
state preparation algorithms can be plugged
in to our approach, and further study of the LMPW HQC algorithms'
performance can be conducted under them.

When using the aforementioned elementary method for thermal state preparation, it is essential to make use of the following identity, holding for all $b \in\mathbb{R} $ and every Hermitian matrix $A$, in order for this approach to be numerically stable, especially at low temperature $T$:
\begin{equation}
      \frac{\exp(A)}{\operatorname{Tr}[\exp(A)]} = \frac{\exp(A+b I)}{\operatorname{Tr}[\exp(A+bI)]}.
      \label{eq:numerically-stable-thermal-state}
\end{equation}
For the specific case of the algorithms considered in our paper, we make the choices $A = -\frac{1}{T}(H-\mu\cdot Q)$ and $b = \frac{1}{T} \lambda_{\min}\!\left(H-\mu\cdot Q\right)$, so that the identity in~\eqref{eq:numerically-stable-thermal-state} applies as follows to our case:
\begin{multline}
    \frac{\exp\!\left(-\frac{1}{T}(H-\mu\cdot Q)\right)}{\operatorname{Tr}\!\left[\exp\!\left(-\frac{1}{T}(H-\mu\cdot Q)\right)\right]} =\\ \frac{\exp\!\left(-\frac{1}{T}(H-\mu\cdot Q) +bI\right)}{\operatorname{Tr}\!\left[\exp\!\left(-\frac{1}{T}(H-\mu\cdot Q)+bI \right)\right]}
    \label{eq:specific-thermal-state-stabilize}
\end{multline}
In all of our numerical implementations, we make use of the right-hand side of~\eqref{eq:specific-thermal-state-stabilize} when encoding the numerical values of a thermal state at temperature $T$. Specifically, we evaluate the matrix $-\frac{1}{T}(H-\mu\cdot Q) +bI$, apply the matrix exponential function, and then normalize the resulting matrix by its trace. Adding the specific value of $b$ given above has the effect of balancing out extremely low eigenvalues of $-\frac{1}{T}(H-\mu\cdot Q)$, so that the result of the matrix exponential function is numerically stable (not doing so leads to numerical instabilities).
See Sections~\ref{subsec:Quantum-thermodynamics-in} and~\ref{subsec:Review-of-constrained} for more details of the thermal state in~\eqref{eq:specific-thermal-state-stabilize}. 


\subsection{Paper organization}

\label{subsec:Paper-organization}Our paper is organized as follows.
Section~\ref{sec:Review} provides review material required for understanding
the rest of the paper. In particular, Section~\ref{subsec:Quantum-thermodynamics-in}
reviews quantum thermodynamics in the presence of conserved, non-commuting
charges, Section~\ref{subsec:LMPW-algorithms-for} reviews the LMPW
classical and HQC algorithms for constrained energy minimization,
Section~\ref{subsec:Quantum-Heisenberg-models} reviews quantum Heisenberg
models, and Section~\ref{subsec:Quantum-error-correction} reviews
quantum error correction codes with a focus on stabilizer codes.

Our contributions start with Section~\ref{sec:Design-of-Hamiltonians},
where we offer an alternative interpretation of the LMPW algorithms
as methods for designing Hamiltonians with ground and thermal states
satisfying desired properties. In Section~\ref{sec:Stabilizer-codes-as},
we observe that every stabilizer code corresponds to a stabilizer
thermodynamic system, in which the Hamiltonian is formed from the
code's stabilizer operators and the conserved, non-commuting charges
are formed from the code's logical operators. In Section~\ref{sec:Encoding-quantum-information},
we argue how the LMPW HQC algorithms for constrained energy minimization
can be used to encode quantum information into stabilizer codes when
applied to stabilizer thermodynamic systems. We then detail our simulation
results in Section~\ref{sec:Simulation-results}, exhibiting the performance
of the LMPW classical and HQC algorithms on various quantum Heisenberg
models and stabilizer thermodynamic systems. We finally conclude in
Section~\ref{sec:Conclusion} with a summary of our findings and suggestions
for future research.

\section{Review}

\label{sec:Review}

\subsection{Quantum thermodynamics with conserved, non-commuting charges}

\label{subsec:Quantum-thermodynamics-in}

Recall that a qudit is a
$d$-dimensional quantum system, where $d\in\mathbb{N}$, and a qubit
is a qudit such that $d=2$. A thermodynamic system of $n\in\mathbb{N}$
qudits is described by a Hamiltonian $H$ and a tuple
\begin{equation}
Q\coloneqq(Q_{1},\ldots,Q_{c})
\end{equation}
of non-commuting charges, where $c\in\mathbb{N}$. Each of these matrices
is Hermitian and has dimension $d^{n}\times d^{n}$. The charges are
conserved if each of them commutes with the Hamiltonian; i.e., if
\begin{equation}
[H,Q_{i}]=0\quad\text{for all}\quad i\in\left[c\right]\equiv\{1,\ldots,c\}.\label{eq:conserved-charges-def}
\end{equation}
However, it is not required that the charges in $Q$ form a commuting
set, and much of the recent interest in quantum thermodynamics regards
the scenario in which the charges do not form a commuting set~\cite{YungerHalpern2016,Guryanova2016,Lostaglio2017,YungerHalpern2020,Anshu2021,Kranzl2023,Majidy2023}.

A charge $Q_{i}$ is extensive if it can be written in the following
additive form:
\begin{equation}
Q_{i}=\sum_{j=1}^{n}C_{i}^{(j)},
\label{eq:def-extensive-charges}
\end{equation}
where here and throughout we adopt the following shorthand:
\begin{equation}
C_{i}^{(j)}\equiv I^{(1)}\otimes\cdots\otimes I^{(j-1)}\otimes C_{i}^{(j)}\otimes I^{(j+1)}\otimes\cdots\otimes I^{(n)},\label{eq:shorthand-omit-identities}
\end{equation}
which leaves tensor products with identity operators implicit if not
written explicitly. In~\eqref{eq:shorthand-omit-identities}, the
superscript of each operator indicates the system on which the operator
acts, $I^{(k)}$ denotes the $d\times d$ identity operator acting
on system $k\in\left[n\right]$, and $C_{i}^{(j)}$ denotes a single-site
charge operator acting on the $j$th system. In typical cases of extensive
charges, we have that $C_{i}^{(j)}=C_{i}^{(k)}$ for all systems $j,k\in[n]$,
so that the charge $Q_{i}$ is equal to the sum of the $i$th charges
of each system. 

Examples of conserved, non-commuting, and extensive charges that we
consider in the quantum Heisenberg model (see Section~\ref{subsec:Quantum-Heisenberg-models})
include the total $X$, $Y$, and $Z$ magnetizations:
\begin{equation}
X_{\tot}\coloneqq\sum_{j=1}^{n}X^{(j)},\quad Y_{\tot}\coloneqq\sum_{j=1}^{n}Y^{(j)},\quad Z_{\tot}\coloneqq\sum_{j=1}^{n}Z^{(j)},
\end{equation}
where $X^{(j)}$, $Y^{(j)}$, and $Z^{(j)}$ denote the Pauli operators
acting on system $j$ and are defined in terms of
\begin{equation}
X\coloneqq\begin{bmatrix}0 & 1\\
1 & 0
\end{bmatrix},\qquad Y\coloneqq\begin{bmatrix}0 & -i\\
i & 0
\end{bmatrix},\qquad Z\coloneqq\begin{bmatrix}1 & 0\\
0 & -1
\end{bmatrix}.\label{eq:pauli-defs}
\end{equation}

A parameterized thermal state of the thermodynamic system described
by $H$ and $(Q_{1},\ldots,Q_{c})$, at temperature $T>0$, is given
by
\begin{equation}
\rho_{T}(\mu)\coloneqq\frac{\exp\!\left(-\frac{1}{T}\left(H-\mu\cdot Q\right)\right)}{Z_{T}(\mu)},\label{eq:param-thermal-state}
\end{equation}
where $\mu\coloneqq\left(\mu_{1},\ldots,\mu_{c}\right)\in\mathbb{R}^{c}$
is the chemical potential vector, we have used the shorthand
\begin{equation}
\mu\cdot Q\coloneqq\sum_{i=1}^{c}\mu_{i}Q_{i},
\end{equation}
and the partition function $Z_{T}(\mu)$ is defined as
\begin{equation}
Z_{T}(\mu)\coloneqq\Tr\!\left[\exp\!\left(-\frac{1}{T}\left(H-\mu\cdot Q\right)\right)\right].\label{eq:partition-func-param-thermal-states}
\end{equation}
The chemical potential vector $\mu$ parameterizes the thermal state
$\rho_{T}(\mu)$ and can be used for enabling $\rho_{T}(\mu)$ to
satisfy constraints on expectation values of the charges in $Q$ (see,
e.g., \cite[Appendix B]{liu2025qthermoSDPs}). The state $\rho_{T}(\mu)$
is known in the quantum thermodynamics literature as a non-Abelian
thermal state~\cite{YungerHalpern2016,YungerHalpern2020,Majidy2023}
and in the quantum machine learning literature as a quantum Boltzmann
machine~\cite{Amin2018,Benedetti2017,Kieferova2017}. Let us note
that the state $\rho_{T}(\mu)$ in~\eqref{eq:param-thermal-state}
is well defined even if the charges are not conserved, but much of
the focus in the quantum thermodynamics literature is on the case
in which they are conserved (i.e., Eq.~\eqref{eq:conserved-charges-def}
is satisfied).

A low-temperature thermal state is an approximation of a Hamiltonian's
ground state, and vice versa, a Hamiltonian's ground state is an approximation
of a low-temperature thermal state. In Appendix~\ref{sec:Closeness-of-low-temperature},
we provide quantitative forms of this approximation in terms of three
well known distinguishability measures for quantum states, including
normalized trace distance, fidelity, quantum relative entropy, and
quantum R\'enyi relative entropies.

All of the thermodynamic models that we deal with in this paper, including
quantum Heisenberg models and stabilizer thermodynamic systems, satisfy
the charge conservation constraint in~\eqref{eq:conserved-charges-def}.
However, we should note that the developments in~\cite{liu2025qthermoSDPs}
do not depend on this constraint holding, nor does the interpretation
of the LMPW algorithms given in Section~\ref{sec:Design-of-Hamiltonians}.

\subsection{LMPW algorithms for constrained energy minimization}

\label{subsec:LMPW-algorithms-for}In this section, we review the
problem of constrained energy minimization, as well as the LMPW
classical and HQC algorithms for solving it~\cite{liu2025qthermoSDPs}.

\subsubsection{Review of constrained energy minimization}

\label{subsec:Review-of-constrained}

Let $H$ be a Hamiltonian, and
let $Q\coloneqq\left(Q_{1},\ldots,Q_{c}\right)$ be a tuple of non-commuting
charges, where $c\in\mathbb{N}$ and $H,Q_{1},\ldots,Q_{c}$ are $d^{n}\times d^{n}$
Hermitian matrices, corresponding to observables for $n$ qudit systems.
As mentioned previously, all of the developments in~\cite{liu2025qthermoSDPs}
do not require the charges in $Q$ to be conserved, and so we do not
make such an assumption in this section.

Given $H$, $Q$, and a tuple $q\coloneqq\left(q_{1},\ldots,q_{c}\right)\in\mathbb{R}^{c}$
of constraint values, constrained energy minimization corresponds
to the following optimization problem:
\begin{equation}
E(\mathcal{Q},q)\coloneqq\min_{\rho\in\mathcal{D}_{d^{n}}}\left\{ \Tr[H\rho]:\Tr[Q_{i}\rho]=q_{i}\ \forall i\in\left[c\right]\right\} ,\label{eq:constrained-energy-min-def}
\end{equation}
where $\mathcal{D}_{d^{n}}$ denotes the set of $d^{n}\times d^{n}$
density matrices and $\mathcal{Q}\equiv\left(H,Q_{1},\ldots,Q_{c}\right)$. Clearly, not every tuple of constraints can be satisfied, in which case the value of $E(\mathcal{Q},q)$ is $+\infty$. Going forward, in order to avoid this singularity, we therefore assume that the tuple of constraints, $\Tr[Q_{i}\rho]=q_{i}\ \forall i\in\left[c\right]$, is such that there exists at least one state $\rho$ satisfying them.

By introducing a temperature $T\geq0$, one can approximate this problem
by the following constrained free energy minimization problem \cite[Section II-B]{liu2025qthermoSDPs}:
\begin{multline}
F_{T}(\mathcal{Q},q)\coloneqq\\
\min_{\rho\in\mathcal{D}_{d^{n}}}\left\{ \Tr[H\rho]-TS(\rho):\Tr[Q_{i}\rho]=q_{i}\ \forall i\in\left[c\right]\right\} ,\label{eq:constrained-free-energy-min-def}
\end{multline}
where $S(\rho)\coloneqq-\Tr[\rho\ln\rho]$ is the von Neumann entropy.
By invoking properties of von Neumann entropy, the following approximation
holds \cite[Lemma 4]{liu2025qthermoSDPs}:
\begin{equation}
E(\mathcal{Q},q)\geq F_{T}(\mathcal{Q},q)\geq E(\mathcal{Q},q)-nT\ln d,
\end{equation}
so that one has $\varepsilon\geq0$ additive error in the approximation
by picking $T=\frac{\varepsilon}{n\ln d}$ \footnote{Let us note that the dimension of the system we are considering here is $d^n$, instead of $d$ as considered previously in \cite{liu2025qthermoSDPs}}. Observe that this choice
implies that the temperature is low, i.e., proportional to the desired
error $\varepsilon$ and inversely proportional to the number of qudits
in the system.

By invoking Lagrangian duality and properties of quantum
relative entropy, one can rewrite the constrained free energy minimization
problem in terms of the following dual chemical potential maximization
problem \cite[Appendix B]{liu2025qthermoSDPs}:
\begin{equation}
F_{T}(\mathcal{Q},q)=\sup_{\mu\in\mathbb{R}^{c}}f(\mu),\label{eq:free-energy-duality}
\end{equation}
where
\begin{equation}
f(\mu)\coloneqq\mu\cdot q-T\ln Z_{T}(\mu)
\label{eq:obj-func-chemical-pot-max}
\end{equation}
and the partition function $Z_{T}(\mu)$ is defined in~\eqref{eq:partition-func-param-thermal-states}.
As part of the proof of~\eqref{eq:free-energy-duality}, one deduces
that parameterized thermal states of the form in~\eqref{eq:param-thermal-state}
are optimal for the constrained free energy minimization problem in~\eqref{eq:constrained-free-energy-min-def} (as
mentioned before, these states are also known as non-Abelian thermal
states or quantum Boltzmann machines). Additionally, the objective
function $f(\mu)$ in~\eqref{eq:free-energy-duality} is concave in
$\mu$ and the spectral norm of its Hessian is bounded from above
by the following smoothness parameter \cite[Lemma 9]{liu2025qthermoSDPs}:
\begin{equation}
L\coloneqq\frac{2}{T}\sum_{i\in\left[c\right]}\left\Vert Q_{i}\right\Vert ^{2}=\frac{2n\ln d}{\varepsilon}\sum_{i\in\left[c\right]}\left\Vert Q_{i}\right\Vert ^{2},\label{eq:Hessian-bound}
\end{equation}
implying that standard optimization methods like gradient ascent and
its variants~\cite{Bubeck2015} are guaranteed to converge to a globally
optimal solution with a number of steps related to the smoothness
parameter $L$ in~\eqref{eq:Hessian-bound}.

The gradient $\nabla_{\mu}f(\mu)$ of the objective function $f(\mu)$
has the following entries \cite[Eq.~(18)]{liu2025qthermoSDPs}:
\begin{equation}
\frac{\partial}{\partial\mu_{i}}f(\mu)=q_{i}-\Tr[Q_{i}\rho_{T}(\mu)].\label{eq:gradient-elements}
\end{equation}
This form of the gradient allows for calculation by means of a classical
algorithm for small system sizes or for estimation by means of a quantum
algorithm equipped with subroutines for thermal state preparation
of $\rho_{T}(\mu)$ and measurement of each observable $Q_{i}$.

\subsubsection{LMPW first-order classical algorithm}

With the background from Section~\ref{subsec:Review-of-constrained}
in hand, let us now recall the LMPW first-order classical algorithm
for constrained energy minimization \cite[Algorithm 1]{liu2025qthermoSDPs}:
\begin{lyxalgorithm}
\label{alg:classical-LMPW}The algorithm proceeds according to the
following steps:
\begin{enumerate}
\item Set $m\leftarrow0$, and initialize $\mu^{m}\leftarrow\left(0,\ldots,0\right)$,
the temperature $T\leftarrow\frac{\varepsilon}{n\ln d}$, the step
size $\eta\in(0,L^{-1}]$, where $L$ is defined in~\eqref{eq:Hessian-bound},
and the stopping parameter $\delta>0$.
\item Increment $m$, and set
\begin{equation}
\mu^{m}\leftarrow\mu^{m-1}+\eta\nabla_{\mu}f(\mu^{m-1}),\label{eq:gradient-update-LMPW}
\end{equation}
where the gradient $\nabla_{\mu}f(\mu^{m-1})$ is defined from~\eqref{eq:gradient-elements}.
\item Repeat Step 2 until $\left\Vert \nabla_{\mu}f(\mu^{m})\right\Vert \leq\delta$.
\item Output
\begin{equation}
\mu^{m}\cdot q+\Tr[\left(H-\mu^{m}\cdot Q\right)\rho_{T}(\mu^{m})]\label{eq:classical-LMPW-alg-output}
\end{equation}
as an approximation of $E(\mathcal{Q},q)$ in~\eqref{eq:constrained-energy-min-def}.
\end{enumerate}
\end{lyxalgorithm}

Recall from \cite[Eq.~(D2)]{liu2025qthermoSDPs} that the following
equality holds:
\begin{equation}
f(\mu)=\mu\cdot q+\Tr[\left(H-\mu\cdot Q\right)\rho_{T}(\mu)]-TS(\rho_{T}(\mu)),\label{eq:actual-fmu-val}
\end{equation}
so that the output of Algorithm~\ref{alg:classical-LMPW} in~\eqref{eq:classical-LMPW-alg-output}
is an approximation of $f(\mu^{m})$ that omits the last term $-TS(\rho_{T}(\mu))$
in~\eqref{eq:actual-fmu-val}. As stated around \cite[Eq.~(21)]{liu2025qthermoSDPs},
we could alternatively output $f(\mu^{m})$ as an approximation of
the desired solution $E(\mathcal{Q},q)$, but the output in~\eqref{eq:classical-LMPW-alg-output}
is consistent with the output of the LMPW HQC algorithms. For our
purposes here, this allows for a more direct comparison of the classical
first-order algorithm above and the first-order HQC algorithm recalled
in Section~\ref{subsec:LMPW-first-order-hybrid}.

We have phrased the LMPW classical algorithm slightly differently
from how it was presented previously in \cite[Algorithm 1]{liu2025qthermoSDPs},
in order to simplify its presentation. The main difference is the
introduction of the stopping parameter $\delta$ above, as a simple
way of ending the algorithm after convergence has been achieved (i.e.,
when the norm of the gradient is small, further changes to $\mu$
are small and convergence near a globally optimal solution is guaranteed
at this point by \cite[Corollary 3.5]{garrigos2024handbookconvergencetheoremsstochastic}). 

To be implemented on a classical computer, Algorithm~\ref{alg:classical-LMPW}
requires computation of matrix multiplications, traces, and matrix
exponentials in order to evaluate the gradient in~\eqref{eq:gradient-elements}
and the algorithm's output in~\eqref{eq:classical-LMPW-alg-output}, while adhering to the approach from~\eqref{eq:specific-thermal-state-stabilize} for encoding the numerical values of the thermal states.
This is possible using standard tools like NumPy in Python~\cite{Harris2020},
and it is what we used in our realizations of the classical algorithm,
the results of which are reported in Section~\ref{sec:Simulation-results}. 

Let us also note that variations of gradient ascent can be employed
instead of the standard approach presented in Algorithm~\ref{alg:classical-LMPW}.
These include momentum~\cite{Polyak1964}, ADAM~\cite{Kingma2017},
and Nesterov accelerated gradient~\cite{Nesterov1983}. However, the
preferred approach here is actually Nesterov accelerated gradient,
because it converges most quickly among all of these methods and is
in fact provably optimal whenever the objective function is smooth
and concave~\cite{Nesterov1983}, as it is in the constrained free-energy
minimization problem.

Before moving on to the LMPW first-order HQC algorithm reviewed
in Section~\ref{subsec:LMPW-first-order-hybrid}, let us note that
the LMPW first-order classical algorithm can be considered a noiseless
simulation of the first-order HQC algorithm. The main reason for this
is that the classical algorithm numerically computes the expectation
values $\Tr[H\rho_{T}(\mu)]$ and $\Tr[Q_{i}\rho_{T}(\mu)]$, for
all $i\in\left[c\right]$, whereas the HQC algorithm estimates these
expectation values through sampling on a quantum computer. As such,
the classical algorithm can serve as a benchmark for the HQC algorithm,
in addition to the benchmark provided by standard SDP solvers.

\subsubsection{LMPW first-order hybrid quantum-classical algorithm}

\label{subsec:LMPW-first-order-hybrid}

As proposed in~\cite{liu2025qthermoSDPs}, Algorithm~\ref{alg:classical-LMPW}
can be modified to a first-order HQC algorithm by replacing the gradient
and output calculations in~\eqref{eq:gradient-update-LMPW} and~\eqref{eq:classical-LMPW-alg-output},
respectively, with gradient and output estimations that are realized
by a quantum computer. Estimating these quantities requires two capabilities:
preparing the thermal state $\rho_{T}(\mu)$ of temperature $T$ and
measuring the observables $H,Q_{1},\ldots,Q_{c}$. The former task
is the bottleneck for the algorithm, given that measuring $H,Q_{1},\ldots,Q_{c}$
can be performed efficiently for observables of physical interest.
As mentioned in Section~\ref{subsec:Methods-for-thermal}, here we
focus on a basic thermal state preparation routine: direct preparation
from the description of the density matrix from a quantum circuit that encodes its numerical values,
while leaving it as an open direction
to implement more advanced methods.

Algorithm~\ref{alg:classical-LMPW} is then modified as follows,
to become the LMPW first-order HQC algorithm for constrained energy
minimization:
\begin{lyxalgorithm}
\label{alg:HQC-LMPW-1st-order}The algorithm proceeds according
to the following steps:
\begin{enumerate}
\item Set $m\leftarrow0$, and initialize $\mu^{m}\leftarrow\left(0,\ldots,0\right)$,
the temperature $T\leftarrow\frac{\varepsilon}{n\ln d}$, the step
size $\eta\in(0,L^{-1}]$, where $L$ is defined in~\eqref{eq:Hessian-bound},
and the stopping parameter $\delta>0$.
\item Increment $m$. For all $i\in\left[c\right]$, set $\widetilde{Q}_{i}$
as an estimate of $\Tr[Q_{i}\rho_{T}(\mu^{m-1})]$, and set
\begin{equation}
\mu_{i}^{m}\leftarrow\mu_{i}^{m-1}+\eta\left(q_{i}-\widetilde{Q}_{i}\right).\label{eq:gradient-update-HQC-LMPW-1st-order}
\end{equation}
\item Repeat Step 2 until $\sqrt{\sum_{i\in\left[c\right]}\left(q_{i}-\widetilde{Q}_{i}\right)^{2}}\leq\delta$.
\item Set $\widetilde{H}$ as an estimate of $\Tr[H\rho_{T}(\mu^{m})]$,
and set $\widetilde{Q}_{i}$ as an estimate of $\Tr[Q_{i}\rho_{T}(\mu^{m})]$.
Output
\begin{equation}
\mu^{m}\cdot q+\widetilde{H}-\sum_{i\in\left[c\right]}\mu_{i}^{m}\widetilde{Q}_{i}\label{eq:HQC-LMPW-alg-output-1st-order}
\end{equation}
as an approximation of $E(\mathcal{Q},q)$ in~\eqref{eq:constrained-energy-min-def}.
\end{enumerate}
\end{lyxalgorithm}

As presented above, Algorithm~\ref{alg:HQC-LMPW-1st-order} is a
simplification of the first-order HQC algorithm proposed in \cite[Algorithm~2]{liu2025qthermoSDPs}.
Since the gradient calculations are replaced by gradient estimates,
the latter are random variables, and the algorithm thus employs stochastic
gradient ascent in its search for a globally optimal solution (see,
e.g., \cite[Chapter 5]{garrigos2024handbookconvergencetheoremsstochastic}
and \cite[Section 6]{Bubeck2015}). Since one of the goals of~\cite{liu2025qthermoSDPs}
was to prove that the algorithm converges to a global optimum, the
presentation of the algorithm was more involved in order to provide
a theoretical guarantee on its performance. However, in practice,
we find in our simulations that the above simplification of the algorithm
works well for converging near the global optimum.

The first-order HQC algorithm proposed in \cite[Algorithm 2]{liu2025qthermoSDPs}
focused on the case in which the Hamiltonian $H$ and the non-commuting
charges $Q_{1},\ldots,Q_{c}$ can be written as a real linear combination
of Pauli matrices. As described there, one can invoke a Monte Carlo
sampling procedure in order to measure these observables. Such a procedure
is useful when the coefficients in the linear combination of Pauli
matrices differ markedly, because terms with less weight will be sampled
less frequently, while terms with more weight will be sampled more
frequently. However, in all of our example problems, the coefficients
have equal weight, so that there is no benefit from performing such
a Monte Carlo sampling procedure and we instead simply measure each
term in the observables sufficiently many times in order to obtain
reasonable estimates.

\subsubsection{LMPW second-order classical algorithm}

It is possible to modify Algorithm~\ref{alg:classical-LMPW} to
become a second-order classical algorithm, i.e., incorporating second-derivative
information to aid the search \cite[Section VI]{liu2025qthermoSDPs}.
In this case, the Hessian of the objective function $f(\mu)$ has
the following matrix elements \cite[Eqs.~(24)--(25)]{liu2025qthermoSDPs}:
\begin{multline}
\frac{\partial^{2}}{\partial\mu_{i}\partial\mu_{j}}f(\mu)=-\frac{1}{T}\int_{0}^{1}ds\ \Tr[\rho_{T}(\mu)^{1-s}Q_{i}\rho_{T}(\mu)^{s}Q_{j}]\\
+\frac{1}{T}\Tr[Q_{i}\rho_{T}(\mu)]\Tr[Q_{j}\rho_{T}(\mu)].\label{eq:Hessian-matrix-real-integral}
\end{multline}
As noted in \cite[Lemmas~6 and 8]{liu2025qthermoSDPs}, the Hessian is negative semidefinite, implying that the objective function $f(\mu)$ in~\eqref{eq:obj-func-chemical-pot-max} is concave.
Let us denote the Hessian matrix evaluated at $\nu\in\mathbb{R}^{c}$
by $\nabla_{\mu}^{2}f(\nu)$, so that it has matrix elements
\begin{multline}
\left[\nabla_{\mu}^{2}f(\nu)\right]_{i,j}=-\frac{1}{T}\int_{0}^{1}ds\ \Tr[\rho_{T}(\nu)^{1-s}Q_{i}\rho_{T}(\nu)^{s}Q_{j}]\\
+\frac{1}{T}\Tr[Q_{i}\rho_{T}(\nu)]\Tr[Q_{j}\rho_{T}(\nu)].\label{eq:Hessian-elements-w-notation}
\end{multline}
One can then use the explicit expression in~\eqref{eq:Hessian-elements-w-notation}
to calculate the Hessian at each iteration of Algorithm~\ref{alg:classical-LMPW}.
This leads to the following modified classical algorithm for constrained
energy minimization:
\begin{lyxalgorithm}
\label{alg:classical-LMPW-2nd-order}The algorithm proceeds according
to the following steps:
\begin{enumerate}
\item Set $m\leftarrow0$, and initialize $\mu^{m}\leftarrow\left(0,\ldots,0\right)$,
the temperature $T\leftarrow\frac{\varepsilon}{n\ln d}$, the step
size $\eta>0$, 
and the stopping parameter $\delta>0$.
\item Increment $m$, and solve for $\Delta$ in
\begin{equation}
\nabla_{\mu}^{2}f(\mu^{m-1})\Delta=\nabla_{\mu}f(\mu^{m-1}),
\end{equation}
where the gradient $\nabla_{\mu}f(\mu^{m-1})$ is defined from~\eqref{eq:gradient-elements}
and the Hessian matrix $\nabla_{\mu}^{2}f(\mu^{m-1})$ from~\eqref{eq:Hessian-elements-w-notation}.
Set
\begin{equation}
\mu^{m}\leftarrow\mu^{m-1}-\eta\Delta.\label{eq:gradient-update-LMPW-2nd-order}
\end{equation}
\item Repeat Step 2 until $\left\Vert \nabla_{\mu}f(\mu^{m})\right\Vert \leq\delta$.
\item Output
\begin{equation}
\mu^{m}\cdot q+\Tr[\left(H-\mu^{m}\cdot Q\right)\rho_{T}(\mu^{m})]\label{eq:classical-LMPW-alg-output-2nd-order}
\end{equation}
as an approximation of $E(\mathcal{Q},q)$ in~\eqref{eq:constrained-energy-min-def}.
\end{enumerate}
\end{lyxalgorithm}

Due to the Hessian being negative semidefinite, the minus sign in~\eqref{eq:gradient-update-LMPW-2nd-order} effectively cancels out with the implicit minus sign from the Hessian $\nabla_{\mu}^{2}f(\mu^{m-1})$, so that each gradient update always proceeds in a positive direction and corresponds to an ascent procedure.

The algorithm can be sensitive to the step size $\eta$. The basic form of a Newton update simply sets $\eta = 1$. However, such an update does not always lead to convergence, and instead we incorporate a backtracking approach based on the value of the norm of the gradient. If, at a certain iteration, the gradient norm increases, this indicates that the algorithm is not making progress, and we reduce the step size in a multiplicative way until the gradient norm decreases. At that point we continue with larger step sizes while making progress.

By incorporating second-derivative information, Algorithm~\ref{alg:classical-LMPW-2nd-order}
typically requires fewer iterations to converge than does Algorithm
\ref{alg:classical-LMPW}. However, each iteration of Algorithm
\ref{alg:classical-LMPW-2nd-order} is computationally more expensive
because it requires the calculation of the Hessian matrix.

Let us note that Algorithm~\ref{alg:classical-LMPW-2nd-order} can
be used as a direct benchmark for the performance of the second-order
HQC algorithm detailed in Section~\ref{subsec:LMPW-second-order-hybrid}.
As before, it can be considered a noiseless simulation of Algorithm
\ref{alg:HQC-LMPW-2nd-order}, in the sense that it does not involve
shot noise error that occurs from having to estimate the gradient
and the Hessian of the objective function $f(\mu)$.

\subsubsection{LMPW second-order hybrid quantum-classical algorithm}

\label{subsec:LMPW-second-order-hybrid}Finally, let us recall a
second-order HQC algorithm for constrained energy minimization, as
proposed in \cite[Section VI]{liu2025qthermoSDPs}. A key observation,
originally made in the context of a general-purpose natural gradient
algorithm for quantum Boltzmann machines \cite[Theorem 2]{patel2024naturalgradientparameterestimation},
is that the Hessian matrix in~\eqref{eq:Hessian-matrix-real-integral}
admits the following alternative expression \cite[Lemma 6]{liu2025qthermoSDPs}:
\begin{multline}
\frac{\partial^{2}}{\partial\mu_{i}\partial\mu_{j}}f(\mu)=-\frac{1}{2T}\Tr[\left\{ \Phi_{\mu}(Q_{i}),Q_{j}\right\} \rho_{T}(\mu)]\\
+\frac{1}{T}\Tr[Q_{i}\rho_{T}(\mu)]\Tr[Q_{j}\rho_{T}(\mu)],\label{eq:Hessian-matrix-fourier-integral}
\end{multline}
where $\Phi_{\mu}$ is the following quantum channel~\cite{Hastings2007,Kim2012,Anshu2021,Coopmans2024,patel2024quantumboltzmannmachinelearning,patel2024naturalgradientparameterestimation}:
\begin{equation}
\Phi_{\mu}(\cdot)\coloneqq\int_{-\infty}^{\infty}dt\ p(t)\,e^{-i\left(H-\mu\cdot Q\right)t/T}(\cdot)e^{i\left(H-\mu\cdot Q\right)t/T},
\end{equation}
and $p(t)$ is the following high-peak tent probability density defined
for $t\in\mathbb{R}$ \cite[Eq.~(9)]{patel2024quantumboltzmannmachinelearning}:
\begin{equation}
p(t)\coloneqq\frac{2}{\pi}\ln\left|\coth(\pi t/2)\right|.
\end{equation}
The elements of the Hessian matrix in~\eqref{eq:Hessian-matrix-fourier-integral}
break up into two terms (the first in the first line of~\eqref{eq:Hessian-matrix-fourier-integral}
and the second in the second line of~\eqref{eq:Hessian-matrix-fourier-integral}).
The second term can be estimated by measuring the observable $Q_{i}\otimes Q_{j}$
with respect to the state $\rho_{T}(\mu)\otimes\rho_{T}(\mu)$. The
first term can be estimated by a combination of several subroutines,
including classical random sampling from $p(t)$, Hamiltonian simulation
\cite{lloyd1996universal,childs2018toward} according to $e^{-i\left(H-\mu\cdot Q\right)t/T}$,
and the Hadamard test~\cite{Cleve1998}. The precise procedure is
detailed in \cite[Figure 3 and Appendix H]{liu2025qthermoSDPs}.

Let us now write down the explicit steps of the LMPW second-order
HQC algorithm for constrained energy minimization:
\begin{lyxalgorithm}
\label{alg:HQC-LMPW-2nd-order}The algorithm proceeds according
to the following steps:
\begin{enumerate}
\item Set $m\leftarrow0$, and initialize $\mu^{m}\leftarrow\left(0,\ldots,0\right)$,
the temperature $T\leftarrow\frac{\varepsilon}{n\ln d}$, the step
size $\eta>0$, 
and the stopping parameter $\delta>0$.

\item Increment $m$. For all $i\in\left[c\right]$, set
\begin{equation}
\widetilde{\nabla}\leftarrow\left(q_{i}-\widetilde{Q}_{i}\right)_{i\in[c]}
\end{equation}
as an estimate of the gradient, where $\widetilde{Q}_{i}$ is an estimate
of $\Tr[Q_{i}\rho_{T}(\mu^{m-1})]$. Set $\widetilde{\nabla^{2}}$
as an estimate of the Hessian matrix in~\eqref{eq:Hessian-matrix-fourier-integral},
by following \cite[Algorithm~4]{liu2025qthermoSDPs}. Solve for $\Delta$
in
\begin{equation}
\widetilde{\nabla^{2}}\Delta=\widetilde{\nabla},
\end{equation}
and set
\begin{equation}
\mu^{m}\leftarrow\mu^{m-1}-\eta\Delta,\label{eq:gradient-update-LMPW-2nd-order-1}
\end{equation}

\item Repeat Step 2 until $\sqrt{\sum_{i\in\left[c\right]}\left(q_{i}-\widetilde{Q}_{i}\right)^{2}}\leq\delta$.
\item Set $\widetilde{H}$ as an estimate of $\Tr[H\rho_{T}(\mu^{m})]$,
and set $\widetilde{Q}_{i}$ as an estimate of $\Tr[Q_{i}\rho_{T}(\mu^{m})]$
for all $i\in\left[c\right]$. Output
\begin{equation}
\mu^{m}\cdot q+\widetilde{H}-\sum_{i\in\left[c\right]}\mu_{i}^{m}\widetilde{Q}_{i}\label{eq:HQC-LMPW-alg-output-2nd-order}
\end{equation}
as an approximation of $E(\mathcal{Q},q)$ in~\eqref{eq:constrained-energy-min-def}.
\end{enumerate}
\end{lyxalgorithm}

In addition to the backtracking approach mentioned after Algorithm~\ref{alg:classical-LMPW-2nd-order}, we also regularize the estimate of the Hessian matrix, by adding a scaled identity matrix to it, in order to ensure that it is negative semidefinite. By employing both backtracking and regularization, we find in our simulations that the algorithm converges in all cases. Without doing so, the estimate of the Hessian matrix need not be negative semidefinite, due to noise inherent in the sampling process, and the algorithm need not converge in such a case. 

The implementation of the channel $\Phi_\mu$ can be simplified for estimating the Hessian when the charges are conserved (i.e., when~\eqref{eq:conserved-charges-def} holds). To see this, observe that the term $\Phi_\mu(Q_i)$ appears in \eqref{eq:Hessian-matrix-fourier-integral}, which can be written as follows:
\begin{align}
    \Phi_\mu(Q_i) & = \int_{-\infty}^{\infty}dt\ p(t)\,e^{-i\left(H-\mu\cdot Q\right)t/T}Q_ie^{i\left(H-\mu\cdot Q\right)t/T} \\
    & = \int_{-\infty}^{\infty}dt\ p(t)\,e^{i\mu\cdot Qt/T}Q_i e^{-i\mu\cdot Qt/T},
    \label{eq:simplify-phi-conserved-2}
\end{align}
where the second equality follows by employing the assumption in \eqref{eq:conserved-charges-def} and the fact that 
\begin{equation}
    e^{A+B} = e^{A}e^{B}
    \label{eq:exp-sum-commute}
\end{equation}
for Hermitian $A$ and $B$ satisfing $[A,B]=0$. If the charges are extensive in addition to being conserved (i.e., if~\eqref{eq:def-extensive-charges} holds), the implementation of the channel $\Phi_\mu$ can be simplified even more. This is the case for the quantum Heisenberg models reviewed in Section~\ref{subsec:Quantum-Heisenberg-models}. In this case, it follows that
\begin{equation}
    \Phi_\mu(Q_i) = \sum_{j=1}^{n}\int_{-\infty}^{\infty}dt\ p(t)\, e^{\frac{it}{T}  \mu\cdot  C^{(j)}} C_{i}^{(j)}  e^{-\frac{it}{T}  \mu \cdot  C^{(j)} },
    \label{eq:simplify-phi-conserved-extensive-charges}
\end{equation}
where $C^{(j)} \equiv (C_{1}^{(j)}, \ldots, C_{c}^{(j)})$. This simplification has the effect of reducing a multi-site  
Hamiltonian evolution, acting on all systems, to single-site evolutions, which act only on single systems, and can significantly reduce the circuit depth. See Appendix~\ref{app:simplifying-phi-mu} for a proof of~\eqref{eq:simplify-phi-conserved-extensive-charges}.

\subsection{Quantum Heisenberg models}

\label{subsec:Quantum-Heisenberg-models}As mentioned in Section~\ref{subsec:Summary-of-contributions},
quantum Heisenberg models describe interacting spin systems on a lattice
and are essential for understanding magnetic materials~\cite{Mattis2006}.
They are key examples considered for our simulations of the LMPW
algorithms, the results of which we report in Section~\ref{sec:Simulation-results}.
A general form for a quantum Heisenberg model involves qubits placed
at the vertices of a graph $G$ with vertex set $V$ and edge set
$E$~\cite{Goldschmidt2011}. The Hamiltonian is defined in terms
of the graph $G$ and a list $\vec{J}\coloneqq(J_{i,j})_{\left\{ i,j\right\} \in E}$
of coupling coefficients as follows:
\begin{multline}
H(G,\vec{J})\coloneqq\\
\sum_{\left\{ i,j\right\} \in E}J_{i,j}\left(X^{(i)}\otimes X^{(j)}+Y^{(i)}\otimes Y^{(j)}+Z^{(i)}\otimes Z^{(j)}\right).\label{eq:heisenberg-model}
\end{multline}
The total magnetizations in the $x$, $y$, and $z$ directions are
defined as follows:
\begin{align}
X(G) & \coloneqq\sum_{i\in V}X^{(i)},\label{eq:heisenberg-model-charge-x}\\
Y(G) & \coloneqq\sum_{i\in V}Y^{(i)},\\
Z(G) & \coloneqq\sum_{i\in V}Z^{(i)}.\label{eq:heisenberg-model-charge-z}
\end{align}

The following commutation relations hold for every graph $G$ and
for every list $\vec{J}$ of coupling coefficients:
\begin{align}
\left[H(G,\vec{J}),X(G)\right] & =\left[H(G,\vec{J}),Y(G)\right]\label{eq:heis-comm-1}\\
 & =\left[H(G,\vec{J}),Z(G)\right]\\
 & =0.\label{eq:eq:heis-comm-last}
\end{align}
These are a direct consequence of the following commutation relations
holding:
\begin{align}
\left[X\otimes X+Y\otimes Y+Z\otimes Z,X\otimes I+I\otimes X\right] & =0,\\
\left[X\otimes X+Y\otimes Y+Z\otimes Z,Y\otimes I+I\otimes Y\right] & =0,\\
\left[X\otimes X+Y\otimes Y+Z\otimes Z,Z\otimes I+I\otimes Z\right] & =0,
\end{align}
which in turn follow directly from applying the equalities
\begin{align}
XY & =-YX=iZ,\\
YZ & =-ZY=iX,\\
ZX & =-XZ=iY.
\end{align}
Clearly, the total magnetization operators $X(G)$, $Y(G)$, and $Z(G)$
do not form a commuting set. 

As a consequence of the commutation relations in~\eqref{eq:heis-comm-1}--\eqref{eq:eq:heis-comm-last},
it follows that every quantum Heisenberg model realizes a thermodynamic
system with conserved, non-commuting charges, in which $H(G,\vec{J})$
is the Hamiltonian and $X(G)$, $Y(G)$, and $Z(G)$ are the conserved,
non-commuting charges. Several examples of quantum Heisenberg models
were closely studied in various prior works on quantum thermodynamics
with conserved, non-commuting charges~\cite{YungerHalpern2020,Kranzl2023,Majidy2023}.

Special cases of the general model presented in~\eqref{eq:heisenberg-model}--\eqref{eq:eq:heis-comm-last}
include the one- and two-dimensional quantum Heisenberg models with
nearest-neighbor and next-nearest-neighbor interactions. In the
case of a one-dimensional model, the qubits are arranged on a line
with edges placed between nearest neighbors, for the nearest neighbor
case, and edges additionally placed between next-nearest neighbors,
for the next-nearest neighbor case. To simplify the model further,
we set the coupling coefficients to be a fixed value $J$ between
nearest neighbors, and for next-nearest neighbors, we set the coupling
coefficients to be fixed at a constant value, say, $\lambda J$, where
$\lambda\in[0,1]$. For a two-dimensional model, the main difference
is that the qubits are arranged on a square lattice, and we take similar
conventions, as mentioned above, for the nearest-neighbor and next-nearest-neighbor cases. See Figure~\ref{fig:heisenberg-models} for a depiction of these models.

\begin{figure}
\includegraphics[width=\linewidth]{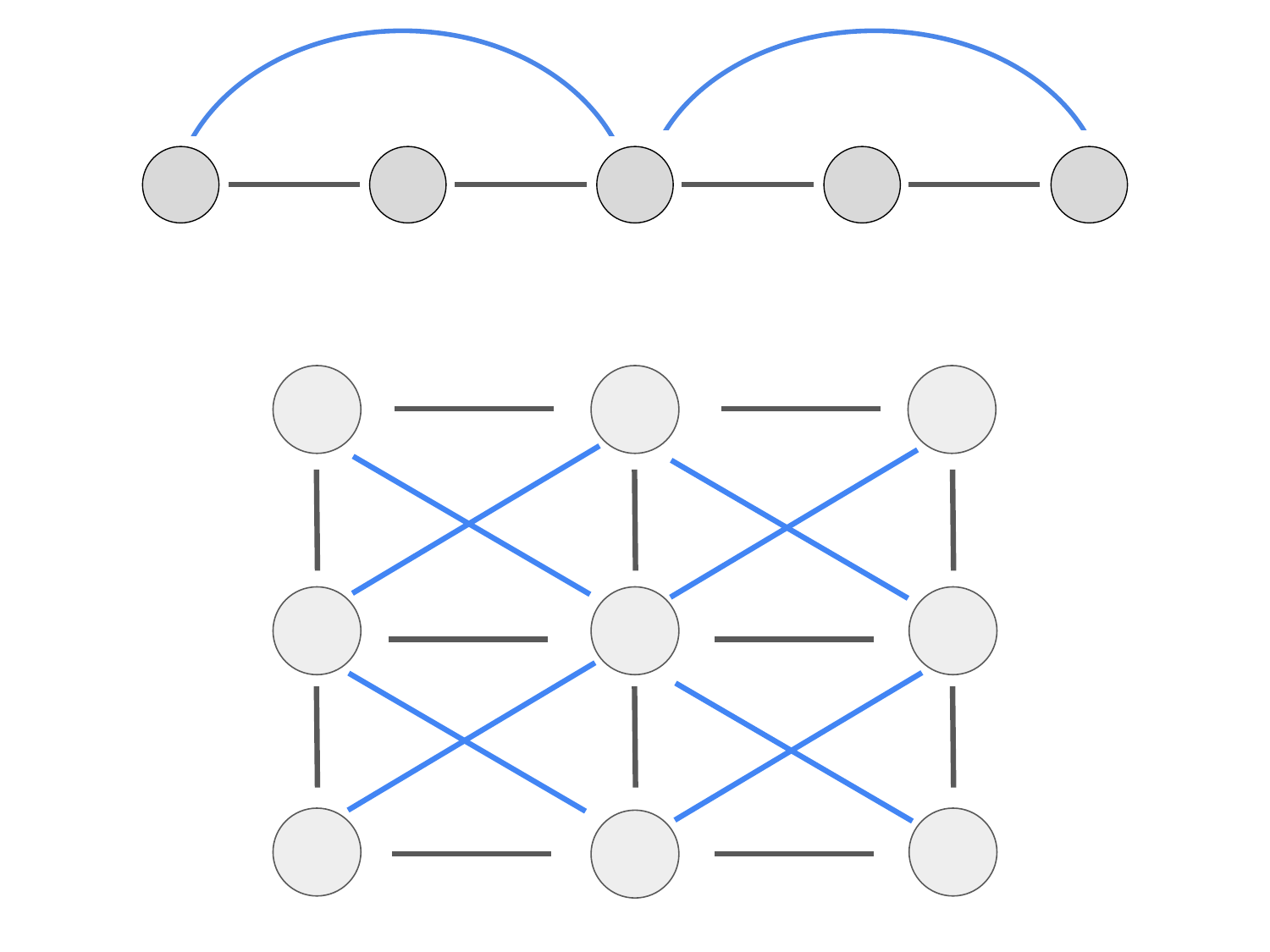}
\caption{Depiction of one- and two-dimensional Heisenberg models. (Top) Qubits of the one-dimensional model are arranged on a line, with edges placed between nearest neighbors (black lines) and additionally between next-nearest neighbors (blue curves).
(Bottom) Qubits of the two-dimensional model are arranged on a square lattice, with nearest-neighbor and next-nearest-neighbor interactions shown in black and blue lines, respectively.}
\label{fig:heisenberg-models}
\end{figure}

\subsection{Stabilizer codes for quantum error correction}

\label{subsec:Quantum-error-correction}Here we provide a brief review
of stabilizer codes~\cite{Gottesman1996,Calderbank1997} (see also
\cite{Gottesman1997}), which suffices for our purposes in this paper,
and we point to various reviews for further exposition~\cite{Gottesman2009,Devitt2013,Lidar2013,Roffe2019,Bradshaw2025}.

Let $\mathcal{P}^{n}$ denote the Pauli group for $n\in\mathbb{N}$
qubits. The Pauli group for a single qubit has the generating set
$\left\langle X,Y,Z\right\rangle $, where $X$, $Y$, and $Z$ are
defined in~\eqref{eq:pauli-defs}, and the Pauli group on $n$ qubits
is generated by taking tensor products of the elements of this generating
set.

Let $k\in\mathbb{N}$ be such that $k\leq n$. A stabilizer $\mathcal{S}$
is an Abelian subgroup of the Pauli group $\mathcal{P}^{n}$ not containing
$-I$ and is generated by $n-k$ commuting Pauli operators $S_{1},\ldots,S_{n-k}$:
\begin{equation}
\mathcal{S}\coloneqq\left\langle S_{1},\ldots,S_{n-k}\right\rangle .
\end{equation}
By definition, the stabilizer operators $S_{1},\ldots,S_{n-k}$ in
the generating set are independent, in the sense that it is not possible
to realize any other one in the generating set by multiplying any
two of them in it. The codespace $\mathcal{C}$ is the vector space
stabilized by $\mathcal{S}$; i.e., it is equal to the simultaneous
$+1$-eigenspace of the stabilizer operators in $\mathcal{S}$:
\begin{equation}
\mathcal{C}\coloneqq\SPAN\!\left\{ |\psi\rangle\in\left(\mathbb{C}^{2}\right)^{\otimes n}:S|\psi\rangle=|\psi\rangle\ \ \forall S\in\mathcal{S}\right\} .
\end{equation}
The condition that $-I\notin\mathcal{S}$ ensures that $\mathcal{S}$
has a simultaneous $+1$-eigenspace. The codespace $\mathcal{C}$
is the subspace of $\left(\mathbb{C}^{2}\right)^{\otimes n}$ into
which quantum information is encoded. The projection onto the codespace
is given by
\begin{align}
\Pi_{\mathcal{C}} & \coloneqq\frac{1}{2^{n-k}}\sum_{S\in\mathcal{S}}S=\prod_{i=1}^{n-k}\frac{I^{\otimes n}+S_{i}}{2},
\end{align}
where the latter product is over $S_{1},\ldots,S_{n-k}$, i.e., those
in the generating set for $\mathcal{S}$.

The set $\mathcal{L}$ of logical operators for the stabilizer $\mathcal{S}$
consists of the elements of the normalizer of $\mathcal{S}$:
\begin{equation}
\mathcal{L}\coloneqq N(\mathcal{S})\coloneqq\left\{ L\in\mathcal{P}^{n}:L\mathcal{S}=\mathcal{S}L\right\} ,\label{eq:logical-ops-normalizer}
\end{equation}
where $L\mathcal{S}$ is the group realized by left multiplying all
elements of $\mathcal{S}$ by $L$. For a stabilizer code, the normalizer
$N(\mathcal{S})$ of $\mathcal{S}$ coincides with its centralizer
\cite[Section 3.2]{Gottesman1997}:
\begin{equation}
N(\mathcal{S})=C(\mathcal{S})\coloneqq\left\{ L\in\mathcal{P}^{n}:LS=SL\ \ \forall S\in\mathcal{S}\right\} .\label{eq:centralizer-stab}
\end{equation}
The logical operators of $\mathcal{S}$ are generated by $2k$ operators
that are labeled by $\overline{X}_{1},\overline{Z}_{1},\ldots,\overline{X}_{k},\overline{Z}_{k}$,
with this notation being used because their commutation relations
are the same as those of the standard Pauli operators $X_{1},Z_{1},\ldots,X_{k},Z_{k}$
acting on $k$ qubits:
\begin{align}
\left[\overline{X}_{i},\overline{X}_{j}\right] & =0\quad\forall i,j\in[k],\label{eq:logical-ops-comm-rels-1}\\
\left[\overline{Z}_{i},\overline{Z}_{j}\right] & =0\quad\forall i,j\in[k],\\
\left[\overline{X}_{i},\overline{Z}_{j}\right] & =0\quad\forall i,j\in[k]:i\neq j,\\
\left\{ \overline{X}_{i},\overline{Z}_{i}\right\}  & =0\quad\forall i\in[k].\label{eq:logical-ops-comm-rels-last}
\end{align}
As such, the logical operators can be used to manipulate the $k$
logical qubits encoded in the stabilizer code. In correspondence with
the equality $iXZ=Y$ that holds for the Pauli operators, we can also
define the logical $\overline{Y}_{i}$ operator for each $i\in\left[k\right]$
as follows:
\begin{equation}
\overline{Y}_{i}\coloneqq i\overline{X}_{i}\overline{Z}_{i}.
\end{equation}

A state $\rho$ of $k$ qubits is uniquely characterized by its Pauli
coefficients $\left(r_{i_{1},\ldots,i_{k}}\right)_{i_{1},\ldots,i_{k}\in\{0,1,2,3\}}$:
\begin{equation}
\rho=\frac{1}{2^{k}}\sum_{i_{1},\ldots,i_{k}\in\{0,1,2,3\}}r_{i_{1},\ldots,i_{k}}\sigma_{i_{1}}\otimes\cdots\otimes\sigma_{i_{k}},\label{eq:pauli-expansion-k-qubits}
\end{equation}
where we have introduced the notation $\sigma_{0}\equiv I$, $\sigma_{1}\equiv X$,
$\sigma_{2}\equiv Y$, and $\sigma_{3}\equiv Z$. Similar to this,
a state $\overline{\rho}$ of $k$ logical qubits encoded in the codespace
$\mathcal{C}$ of a stabilizer code $\mathcal{S}$ is uniquely characterized
by its logical Pauli coefficients $\left(r_{i_{1},\ldots,i_{k}}\right)_{i_{1},\ldots,i_{k}\in\{0,1,2,3\}}$:
\begin{equation}
\overline{\rho}=\Pi_{\mathcal{C}}\left(\frac{1}{2^{k}}\sum_{i_{1},\ldots,i_{k}\in\{0,1,2,3\}}r_{i_{1},\ldots,i_{k}}\overline{\sigma}_{i_{1},1}\cdots\overline{\sigma}_{i_{k},k}\right)\Pi_{\mathcal{C}},\label{eq:encoded-pauli-expansion}
\end{equation}
where we have introduced the notation $\sigma_{0,i}\equiv\overline{I}_{i}=I^{\otimes n}$,
$\sigma_{1,i}\equiv\overline{X}_{i}$, $\sigma_{2}\equiv\overline{Y}_{i}$,
and $\sigma_{3}\equiv\overline{Z}_{i}$ for the logical operators
of the code. For clarity, we include a short derivation of~\eqref{eq:encoded-pauli-expansion}
in Appendix~\ref{sec:Derivation-of-Equation}.

\section{Design of Hamiltonians, ground states, and thermal states}

\label{sec:Design-of-Hamiltonians}In this section, we offer alternative
interpretations and applications of the constrained energy minimization
and constrained free energy minimization problems in~\eqref{eq:constrained-energy-min-def}
and~\eqref{eq:constrained-free-energy-min-def}, respectively. Namely,
we show how applying the LMPW algorithms to these problems leads
to the design of Hamiltonians with ground and thermal states satisfying
fixed constraints on the expectation values of certain observables.

In this way, the LMPW algorithms could be useful for the design
of molecules and materials satisfying desired properties. Specifically, they enable the precise targeting of specific electronic symmetry sectors (defined by particle number, spin, or other charges), which is a prerequisite for calculating the properties of molecular isomers and magnetic materials. As noted in recent literature regarding the Constrained Variational Quantum Eigensolver~\cite{Ryabinkin2018ConstrainedVQ}, standard unconstrained energy minimization is often insufficient for molecular design because the optimization can collapse to an incorrect sector (e.g., a neutral species instead of a cation). While penalty methods are commonly employed to enforce these constraints by energetically punishing unwanted sectors~\cite{McClean2016,McClean2017,Rubin2018,Higgott2019}, our framework offers a rigorous thermodynamic alternative that naturally targets specific sectors via chemical potentials, without requiring the manual tuning of penalty strengths. This approach connects directly to Reduced Density Matrix (RDM) theory~\cite{Mazziotti2012}, where constrained optimization is central to capturing electron correlation effects.

Let us note
that the design of ground states has been considered previously in
the context of the density matrix renormalization group, in the special
case that there is a single chemical potential that constrains a single
conserved charge of the system~\cite{Kuehner1998,Weir2010,Kiely2022}.

To begin with, suppose that the Hamiltonian of a system of $n$ qudits
is of the following form:
\begin{equation}
H-\mu\cdot Q,
\end{equation}
where $H$ is an arbitrary $d^{n}\times d^{n}$ Hermitian matrix,
\begin{equation}
Q\equiv\left(Q_{1},\ldots,Q_{c}\right)
\end{equation}
is a tuple of $d^{n}\times d^{n}$ Hermitian matrices, and $\mu\in\mathbb{R}^{c}$.
Contrary to the setting of Section~\ref{subsec:Review-of-constrained},
here we enforce no other constraints on $H$, $Q$, and $\mu$ (namely,
we do not enforce the conservation constraint in~\eqref{eq:conserved-charges-def}).
Additionally, in distinction to the scenario put forth in Section
\ref{subsec:Review-of-constrained}, now we no longer interpret $Q$
as being separate from the system Hamiltonian, but here instead we
consider $Q$ to be a part of it. That is, we interpret $\mu$ as
a vector of controllable parameters that one can vary in order to
modify the Hamiltonian, such that it will satisfy fixed constraints
on the expectations of the observables in $Q$.

As an example, the quantum Heisenberg model with an external magnetic
field is as follows:
\begin{equation}
H(G,\vec{J})-\mu_{x}X(G)-\mu_{y}Y(G)-\mu_{z}Z(G),\label{eq:heisenberg-model-controllable-ham}
\end{equation}
where $H(G,\vec{J})$ is defined in~\eqref{eq:heisenberg-model},
the operators $X(G)$, $Y(G)$, and $Z(G)$ are defined in~\eqref{eq:heisenberg-model-charge-x}--\eqref{eq:heisenberg-model-charge-z},
and $\left(\mu_{x},\mu_{y},\mu_{z}\right)\in\mathbb{R}^{3}$ is a
vector of control parameters that determines the orientation and strength
of an external magnetic field. One could then vary the control parameters
$\mu_{x}$, $\mu_{y}$, and $\mu_{z}$ such that a ground or thermal
state $\rho$ of~\eqref{eq:heisenberg-model-controllable-ham} satisfies
the following constraints on the total magnetizations:
\begin{equation}
\Tr[X(G)\rho]=q_{x},\quad\Tr[Y(G)\rho]=q_{y},\quad\Tr[Z(G)\rho]=q_{z},\label{eq:constraints-for-heisenberg-model}
\end{equation}
where $\left(q_{x},q_{y},q_{z}\right)\in\mathbb{R}^{3}$ and under
the assumption that there exists at least one state $\rho$ satisfying
the constraints in~\eqref{eq:constraints-for-heisenberg-model}.

\subsection{Design of ground states}

\label{subsec:Design-of-ground}Let $q\equiv\left(q_{1},\ldots,q_{c}\right)\in\mathbb{R}^{c}$
be a vector of constraints. Suppose now that we would like to tune
the control parameters $\mu$ of the Hamiltonian $H-\mu\cdot Q$,
such that there is a state in its ground space satisfying the following
constraints:
\begin{equation}
\Tr[Q_{i}\rho]=q_{i}\qquad\forall i\in\left[c\right].\label{eq:constraints-for-ground-space}
\end{equation}
We refer to this problem as ``design of ground states.'' As mentioned previously in Section~\ref{subsec:Review-of-constrained}, not every tuple of constraints can be satisfied, in which case it is impossible to design a ground state (or any state for that matter) satisfying the constraints in~\eqref{eq:constraints-for-ground-space}. Thus, going forward from here, in order to avoid this singularity, we assume that there exists at least one state satisfying  the tuple of constraints in~\eqref{eq:constraints-for-ground-space}.

Our claim is that this is possible by solving the following constrained
energy minimization problem and its dual:
\begin{equation}
\min_{\rho\in\mathcal{D}_{d^{n}}}\left\{ \Tr[H\rho]:\Tr[Q_{i}\rho]=q_{i}\ \forall i\in\left[c\right]\right\} .\label{eq:constrained-energy-min-design-ham}
\end{equation}
To see this, let us suppose that there exists at least one state satisfying
the constraints in~\eqref{eq:constraints-for-ground-space}, so that
the minimum value in~\eqref{eq:constrained-energy-min-design-ham}
is finite, and then determine the dual of the constrained energy minimization
problem in~\eqref{eq:constrained-energy-min-design-ham} by applying
the Lagrange multiplier method, as follows:
\begin{align}
 & \min_{\rho\in\mathcal{D}_{d^{n}}}\left\{ \Tr[H\rho]:\Tr[Q_{i}\rho]=q_{i}\ \forall i\in\left[c\right]\right\} \nonumber \\
 & =\min_{\rho\in\mathcal{D}_{d^{n}}}\left\{ \Tr[H\rho]+\sup_{\mu\in\mathbb{R}^{c}}\sum_{i\in\left[c\right]}\mu_{i}\left(q_{i}-\Tr[Q_{i}\rho]\right)\right\} \label{eq:proof-lagrange-mults}\\
 & =\min_{\rho\in\mathcal{D}_{d^{n}}}\sup_{\mu\in\mathbb{R}^{c}}\left\{ \Tr[H\rho]+\sum_{i\in\left[c\right]}\mu_{i}\left(q_{i}-\Tr[Q_{i}\rho]\right)\right\} \label{eq:sup-to-outside}\\
 & =\min_{\rho\in\mathcal{D}_{d^{n}}}\sup_{\mu\in\mathbb{R}^{c}}\left\{ \mu\cdot q+\Tr[\left(H-\mu\cdot Q\right)\rho]\right\} \label{eq:algebraic-simplifications}\\
 & =\sup_{\mu\in\mathbb{R}^{c}}\min_{\rho\in\mathcal{D}_{d^{n}}}\left\{ \mu\cdot q+\Tr[\left(H-\mu\cdot Q\right)\rho]\right\} \\
 & =\sup_{\mu\in\mathbb{R}^{c}}\left\{ \mu\cdot q+\lambda_{\min}(H-\mu\cdot Q)\right\} .\label{eq:duality-proof-last}
\end{align}
The first equality in~\eqref{eq:proof-lagrange-mults} follows by
introducing $\mu$ as a vector of Lagrange multipliers to enforce
the constraints in~\eqref{eq:constraints-for-ground-space}. The equality
in~\eqref{eq:sup-to-outside} follows by moving the supremum to the
outside, and the equality in~\eqref{eq:algebraic-simplifications}
follows by simple algebra. The penultimate equality follows from the
Sion minimax theorem for extended reals \cite[Theorem~2.11]{BenDavid2023}.
The final equality follows by bringing in the minimum over all density
matrices and recognizing that
\begin{equation}
\min_{\rho\in\mathcal{D}_{d^{n}}}\Tr[\left(H-\mu\cdot Q\right)\rho]=\lambda_{\min}(H-\mu\cdot Q),
\end{equation}
where $\lambda_{\min}(A)$ denotes the minimum eigenvalue of a Hermitian
operator $A$. In the theory of semidefinite programming, the equality
between the first line and the last line above is known as strong
duality \cite[Section~5.9.1]{Boyd2004}.

The development in~\eqref{eq:proof-lagrange-mults}--\eqref{eq:duality-proof-last}
demonstrates that, in the case that the constrained energy minimization
problem in~\eqref{eq:constrained-energy-min-design-ham} is feasible
(i.e., there exists at least one state satisfying the constraints
in~\eqref{eq:constraints-for-ground-space}), then its minimum value
is finite and, for the dual optimization problem in~\eqref{eq:duality-proof-last},
there exists an optimal choice of $\mu$, which we denote by $\mu^{*}$.
Furthermore, for the optimal choice $\mu^{*}$, there is a corresponding
Hamiltonian, $H-\mu^{*}\cdot Q$, such that there exists a state in
its ground space satisfying the constraints in~\eqref{eq:constraints-for-ground-space}.
To see this, let $\rho^{*}$ be a state that is optimal for~\eqref{eq:constrained-energy-min-design-ham},
and let $\mu^{*}$ be optimal for~\eqref{eq:duality-proof-last}.
Then, the complementary slackness conditions of semidefinite programming
imply that the following equality holds:
\begin{equation}
\left(H-\mu^{*}\cdot Q\right)\rho^{*}=\lambda_{\min}(H-\mu^{*}\cdot Q)\rho^{*}.\label{eq:complementary-slack-cond}
\end{equation}
See Appendix~\ref{sec:Derivation-of-Equation-1} for a derivation
of~\eqref{eq:complementary-slack-cond}. After taking a trace, Eq.~\eqref{eq:complementary-slack-cond}
implies that $\rho^{*}$ is a state in the ground space of $H-\mu^{*}\cdot Q$.

In this way, by solving the constrained energy minimization problem
in~\eqref{eq:constrained-energy-min-design-ham} and its dual in~\eqref{eq:duality-proof-last},
we determine an optimal choice $\mu^{*}$ of the controllable parameters
for the Hamiltonian $H-\mu^{*}\cdot Q$, such that there exists a
state $\rho^{*}$ in the ground space of $H-\mu^{*}\cdot Q$ satisfying
the constraints in~\eqref{eq:constrained-energy-min-design-ham}.
As such, by solving the primal and dual optimization problems, we
can design Hamiltonians with a ground state satisfying desirable properties,
i.e., the constraints in~\eqref{eq:constraints-for-ground-space}.
The LMPW algorithms solve the constrained energy minimization problem
approximately, by instead solving a free-energy minimization problem
at low temperature, as reviewed in Section~\ref{subsec:LMPW-algorithms-for},
and preparing a low-temperature thermal state close to the ground
state. After the LMPW HQC algorithm prepares this state, other observables
could then be measured on it, in order to learn various properties
of it.

Next, in Section~\ref{subsec:Design-of-thermal}, we specifically
discuss the free-energy minimization problem as a way of designing
thermal states with desired properties.

\subsection{Design of thermal states}

\label{subsec:Design-of-thermal}The development from Section~\ref{subsec:Design-of-ground}
extends to the finite-temperature setting and the design of Hamiltonians
having thermal states with desirable properties. We refer to this
problem as ``design of thermal states'' and discuss it in more detail
in this section.

As before, suppose that the controllable Hamiltonian is $H-\mu\cdot Q$,
and the goal is find values of the control parameters in $\mu$ such
that the temperature-$T$ thermal state of $H-\mu\cdot Q$ satisfies
the constraints in~\eqref{eq:constraints-for-ground-space}. Recall
that the temperature-$T$ thermal state of $H-\mu\cdot Q$ has the
form given in~\eqref{eq:param-thermal-state}.

Our claim is that this is possible by solving the dual of the following
constrained free energy minimization problem:
\begin{equation}
\min_{\rho\in\mathcal{D}_{d^{n}}}\left\{ \Tr[H\rho]-TS(\rho):\Tr[Q_{i}\rho]=q_{i}\ \forall i\in\left[c\right]\right\} .\label{eq:constrained-free-energy-min-design-ham}
\end{equation}
To see this, let us revisit the derivation from \cite[Appendix~B]{liu2025qthermoSDPs},
and before doing so, let us recall the quantum relative entropy $D(\omega\|\tau)$
of states $\omega$ and $\tau$~\cite{Umegaki1962}:
\begin{align}
D(\omega\|\tau) & \coloneqq\Tr\!\left[\omega\left(\ln\omega-\ln\tau\right)\right]\\
 & =-S(\rho)-\Tr\!\left[\omega\ln\tau\right],\label{eq:2nd-expr-rel-ent}
\end{align}
with the property that $D(\omega\|\tau)\geq0$ for all states $\omega$
and $\tau$ and $D(\omega\|\tau)=0$ if and only if $\omega=\tau$.
As in the previous section, let us suppose that there exists at least
one state satisfying the constraints in~\eqref{eq:constraints-for-ground-space},
so that the minimum value in~\eqref{eq:constrained-free-energy-min-design-ham}
is finite. Now consider that
\begin{align}
 & \min_{\rho\in\mathcal{D}_{d^{n}}}\left\{ \Tr[H\rho]-TS(\rho):\Tr[Q_{i}\rho]=q_{i}\ \forall i\in\left[c\right]\right\} \nonumber \\
 & =\min_{\rho\in\mathcal{D}_{d^{n}}}\left\{ \begin{array}{cc}
      \Tr[H\rho]-TS(\rho)+  \\
      \sup_{\mu\in\mathbb{R}^{c}}\sum_{i\in\left[c\right]}\mu_{i}\left(q_{i}-\Tr[Q_{i}\rho]\right)
 \end{array}\right\} \label{eq:eq:proof-lagrange-mults-thermal}\\
 & =\min_{\rho\in\mathcal{D}_{d^{n}}}\sup_{\mu\in\mathbb{R}^{c}}\left\{ \begin{array}{cc}
      \Tr[H\rho]-TS(\rho)+  \\
      \sum_{i\in\left[c\right]}\mu_{i}\left(q_{i}-\Tr[Q_{i}\rho]\right)
 \end{array} \right\} \\
 & =\min_{\rho\in\mathcal{D}_{d^{n}}}\sup_{\mu\in\mathbb{R}^{c}}\left\{ \mu\cdot q-TS(\rho)+\Tr[\left(H-\mu\cdot Q\right)\rho]\right\} \\
 & =\sup_{\mu\in\mathbb{R}^{c}}\min_{\rho\in\mathcal{D}_{d^{n}}}\left\{ \mu\cdot q-TS(\rho)+\Tr[\left(H-\mu\cdot Q\right)\rho]\right\} \\
 & =\sup_{\mu\in\mathbb{R}^{c}}\min_{\rho\in\mathcal{D}_{d^{n}}}\left\{ \mu\cdot q+TD(\rho\|\rho_{T}(\mu))-T\ln Z_{T}(\mu)\right\} \label{eq:relative-entropy-recognize}\\
 & =\sup_{\mu\in\mathbb{R}^{c}}\left\{ \mu\cdot q-T\ln Z_{T}(\mu)+T\min_{\rho\in\mathcal{D}_{d^{n}}}D(\rho\|\rho_{T}(\mu))\right\} \\
 & =\sup_{\mu\in\mathbb{R}^{c}}\left\{ \mu\cdot q-T\ln Z_{T}(\mu)\right\} .\label{eq:last-step-legendre-duality}
\end{align}
The first four equalities follow from steps similar to the first few
steps of~\eqref{eq:proof-lagrange-mults}--\eqref{eq:duality-proof-last}.
The equality in~\eqref{eq:relative-entropy-recognize} follows from
recognizing the quantum relative entropy, i.e.,
\begin{align}
 & -TS(\rho)+\Tr[\left(H-\mu\cdot Q\right)\rho]\nonumber \\
 & =T\left(-S(\rho)-\Tr\!\left[\rho\left(-\frac{1}{T}\left(H-\mu\cdot Q\right)\right)\right]\right)\\
 & =T\left(-S(\rho)-\Tr\!\left[\rho\ln\exp\!\left(-\frac{1}{T}\left(H-\mu\cdot Q\right)\right)\right]\right)\\
 & =T\left(-S(\rho)-\Tr\!\left[\rho\ln\rho_{T}(\mu)\right]-\ln Z_{T}(\mu)\right)\\
 & =T\left[D(\rho\|\rho_{T}(\mu))-\ln Z_{T}(\mu)\right],
\end{align}
where the last equality follows from~\eqref{eq:2nd-expr-rel-ent}.
The final equality in~\eqref{eq:last-step-legendre-duality} follows
from the property of quantum relative entropy stated after~\eqref{eq:2nd-expr-rel-ent}.

The development in~\eqref{eq:eq:proof-lagrange-mults-thermal}--\eqref{eq:last-step-legendre-duality}
demonstrates that, in the case that the constrained free energy minimization
problem in~\eqref{eq:constrained-free-energy-min-design-ham} is feasible
(i.e., there exists at least one state satisfying the constraints
in~\eqref{eq:constraints-for-ground-space}), then its minimum value
is finite and, for the dual optimization problem in~\eqref{eq:last-step-legendre-duality},
there exists an optimal choice of $\mu$, which we denote by $\mu^{*}$.
Furthermore, for the optimal choice $\mu^{*}$, there is a corresponding
Hamiltonian, $H-\mu^{*}\cdot Q$, such that the thermal state $\rho_{T}(\mu^{*})$
satisfies the constraints in~\eqref{eq:constraints-for-ground-space}.
To see this, observe that $(\mu^{*},\rho_{T}(\mu^{*}))$ is a saddle
point of~\eqref{eq:eq:proof-lagrange-mults-thermal}--\eqref{eq:last-step-legendre-duality},
such that all equalities are attained. Indeed, for the optimal choice
$\mu^{*}$, the first-order optimality conditions for the objective
function $f(\mu)\equiv\mu\cdot q-T\ln Z_{T}(\mu)$ in~\eqref{eq:last-step-legendre-duality}
guarantee that $\nabla_{\mu}f(\mu^{*})=0$, which implies by~\eqref{eq:gradient-elements}
that $q_{i}=\Tr[Q_{i}\rho_{T}(\mu^{*})]$ for all $i\in\left[c\right]$.
Thus, the state $\rho_{T}(\mu^{*})$ satisfies all of the constraints.
Using this condition, one then finds that
\begin{equation}
\Tr[H\rho_{T}(\mu^{*})]-TS(\rho_{T}(\mu^{*}))=\mu^{*}\cdot q-T\ln Z_{T}(\mu^{*}).
\label{eq:optimality-condition-thermal}
\end{equation}
As such, due to the equalities in~\eqref{eq:eq:proof-lagrange-mults-thermal}--\eqref{eq:last-step-legendre-duality}
holding, it follows that $\rho_{T}(\mu^{*})$ is a thermal state that
is optimal for the constrained free energy minimization problem in~\eqref{eq:constrained-free-energy-min-design-ham}.

In this way, by solving the constrained free energy minimization problem
in~\eqref{eq:constrained-free-energy-min-design-ham} and its dual
in~\eqref{eq:last-step-legendre-duality}, we determine the optimal
choice $\mu^{*}$ of the controllable parameters for the Hamiltonian
$H-\mu^{*}\cdot Q$, such that $\rho_{T}(\mu^{*})$ is a temperature-$T$
thermal state of $H-\mu^{*}\cdot Q$ satisfying the constraints in
\eqref{eq:constrained-energy-min-design-ham}. As such, by solving
this dual optimization problem, we can design Hamiltonians with thermal
states satisfying desirable properties, i.e., the constraints in~\eqref{eq:constraints-for-ground-space}.
The LMPW algorithms solve this problem to arbitrary accuracy and
thus can be used for the design of thermal states. If one is exclusively
interested in the design of thermal states, note that it is not necessary
for the algorithms to output the final value of the free energy, but
they instead can exclusively perform the gradient update steps in
Algorithms~\ref{alg:classical-LMPW}--\ref{alg:HQC-LMPW-2nd-order}
until the norm of the gradient is sufficiently small.

\section{Stabilizer thermodynamic systems}

\label{sec:Stabilizer-codes-as}Given the background in Sections~\ref{subsec:Quantum-thermodynamics-in}
and~\ref{subsec:Quantum-error-correction}, now we define a stabilizer
thermodynamic system as follows:
\begin{defn}[Stabilizer thermodynamic system] Let $\mathcal{S}$ denote a stabilizer code that encodes $k$ logical
qubits into $n$ physical qubits, has commuting stabilizer generators
$S_{1},\ldots,S_{n-k}$, and has the set $\mathcal{L}$ of logical
operators. We define a stabilizer thermodynamic system to have a Hamiltonian
$H$ given by
\begin{equation}
H\coloneqq-\sum_{i=1}^{n-k}\gamma_{i}S_{i},\label{eq:stabilizer-Ham}
\end{equation}
where $\gamma_{i}>0$, and conserved, non-commuting charges given by
\begin{equation}
Q_{i}\coloneqq\sum_{j}\alpha_{i,j}L_{i,j},\label{eq:charge-stabilizer-TS}
\end{equation}
where $L_{i,j}\in\mathcal{L}$ and $\alpha_{i,j}\in\mathbb{R}$, for
all $i\in\left[c\right]$ and $j$. 
\end{defn}

For the stabilizer Hamiltonian $H$ in~\eqref{eq:stabilizer-Ham},
its ground space coincides with the codespace of $\mathcal{S}$, by
construction~\cite{Kitaev2003}. Due to the fact that $\left[L_{i,j},S\right]=0$
for all $S\in\mathcal{S}$ (recall~\eqref{eq:logical-ops-normalizer}
and~\eqref{eq:centralizer-stab}), it follows that each $Q_{i}$ in
\eqref{eq:charge-stabilizer-TS} is a conserved charge, i.e., satisfying
$\left[H,Q_{i}\right]=0$. Also, the tuple $\left(Q_{1},\ldots,Q_{c}\right)$
need not form a commuting set, as observed by recalling~\eqref{eq:logical-ops-comm-rels-1}--\eqref{eq:logical-ops-comm-rels-last},
in which case the charges are non-commuting in general.

It is clear from the definition in~\eqref{eq:charge-stabilizer-TS}
that there is great freedom allowed in what is considered a conserved
charge in a stabilizer thermodynamic system. As a first example, the
conserved charges could be encoded total magnetizations, corresponding
to the total magnetizations of the $k$ logical qubits, as follows:
\begin{equation}
Q_{x}\coloneqq\sum_{i=1}^{k}\overline{X}_{i},\quad Q_{y}\coloneqq\sum_{i=1}^{k}\overline{Y}_{i},\quad Q_{z}\coloneqq\sum_{i=1}^{k}\overline{Z}_{i}.
\end{equation}
As another example, the conserved charges could be the following exhaustive
list of all possible $4^{k}$ logical Pauli operators:
\begin{equation}
Q_{i_{1},\ldots,i_{k}}\coloneqq\overline{\sigma}_{i_{1},1}\cdots\overline{\sigma}_{i_{k},k},
\end{equation}
where $i_{\ell}\in\{0,1,2,3\}$ and we have employed the notation
introduced in~\eqref{eq:encoded-pauli-expansion}. In this latter
case, if all of the expectations $\left\langle Q_{i_{1},\ldots,i_{k}}\right\rangle _{\overline{\rho}}=r_{i_{1},\ldots,i_{k}}$
are fixed, where $\overline{\rho}$ is a state in the ground space
of $H$ and $r_{i_{1},\ldots,i_{k}}\in\mathbb{R}$ for all $i_{1},\ldots,i_{k}\in\{0,1,2,3\}$,
then the state of the $k$ logical qubits is completely specified
after doing so (see~\eqref{eq:encoded-pauli-expansion}). This case
leads to a particular way of using the LMPW HQC algorithms for encoding
quantum information into the ground space of $H$, which we discuss
in more detail in Section~\ref{sec:Encoding-quantum-information}.
\begin{example}
\label{exa:repetition}As a simple example of a stabilizer thermodynamic
system, let us consider the one-to-three qubit repetition code, which
has a stabilizer generated by
\begin{equation}
\left\langle Z_{1}Z_{2},Z_{2}Z_{3}\right\rangle ,
\end{equation}
where $Z_{1}Z_{2}\equiv Z\otimes Z\otimes I$ and $Z_{2}Z_{3}\equiv I\otimes Z\otimes Z$,
and logical operators generated by
\begin{equation}
\left\langle X_{1}X_{2}X_{3},Z_{1}\right\rangle ,
\end{equation}
where $X_{1}X_{2}X_{3}\equiv X\otimes X\otimes X$ and $Z_{1}\equiv Z\otimes I\otimes I$.
In this case, the Hamiltonian of the corresponding stabilizer thermodynamic
system is
\begin{equation}
H=-Z_{1}Z_{2}-Z_{2}Z_{3},
\end{equation}
and we can take the conserved, non-commuting charges to be the encoded
magnetizations of the logical qubit in the $x$, $y$, and $z$ directions:
\begin{align}
X_{\tot} & =X_{1}X_{2}X_{3},\\
Y_{\tot} & =iX_{1}X_{2}X_{3}Z_{1}=Y_{1}X_{2}X_{3},\\
Z_{\tot} & =Z_{1}.
\end{align}
For the case of a single logical qubit, these operators coincide precisely
with the $X$, $Y$, and $Z$ logical operators of the code, and constraints
on them specify the state of the encoded qubit. The parameterized
thermal state of this stabilizer thermodynamic system is as follows:
\begin{equation}
\frac{e^{-\frac{1}{T}\left(-Z_{1}Z_{2}-Z_{2}Z_{3}-\mu_{x}X_{1}X_{2}X_{3}-\mu_{y}Y_{1}X_{2}X_{3}-\mu_{z}Z_{1}\right)}}{\Tr\!\left[e^{-\frac{1}{T}\left(-Z_{1}Z_{2}-Z_{2}Z_{3}-\mu_{x}X_{1}X_{2}X_{3}-\mu_{y}Y_{1}X_{2}X_{3}-\mu_{z}Z_{1}\right)}\right]},
\end{equation}
where $\left(\mu_{x},\mu_{y},\mu_{z}\right)\in\mathbb{R}^{3}$ is
the chemical potential vector.
\end{example}

\section{Encoding quantum information into stabilizer thermodynamic systems}

\label{sec:Encoding-quantum-information}The LMPW first- and second-order
HQC algorithms reviewed in Sections~\ref{subsec:LMPW-first-order-hybrid}
and~\ref{subsec:LMPW-second-order-hybrid} can alternatively serve
as methods for encoding quantum information into stabilizer thermodynamic
systems. To see this, consider a stabilizer thermodynamic system corresponding
to a stabilizer code encoding $k$ logical qubits into $n$ physical
qubits, with a stabilizer $\mathcal{S}$ and a set $\mathcal{L}$
of logical operators. Suppose that we would like to encode a particular
state of the $k$ logical qubits. Recall from~\eqref{eq:encoded-pauli-expansion}
that such a state is specified by $4^{k}$ Pauli coefficients in the
tuple $\left(r_{i_{1},\ldots,i_{k}}\right)_{i_{1},\ldots,i_{k}\in\{0,1,2,3\}}$.
Then, after defining the Hamiltonian $H$ as in~\eqref{eq:stabilizer-Ham},
with $\gamma_{i}=1$ for all $i\in[n-k]$, we set up a constrained
energy minimization of the form in~\eqref{eq:constrained-energy-min-def},
as follows:
\begin{equation}
\begin{aligned}\min_{\rho\in\mathcal{D}_{2^{n}}}\Tr[H\rho]\\
\text{subject to}\\
\Tr[\overline{\sigma}_{i_{1},1}\cdots\overline{\sigma}_{i_{k},k}\rho] & =r_{i_{1},\ldots,i_{k}}\:\forall i_{1},\ldots,i_{k}\in\{0,1,2,3\}.
\end{aligned}
\label{eq:constrained-energy-min-stabilizer}
\end{equation}
Evidently, the optimal solution to this problem is the value $-\left(n-k\right)$,
and the optimal state is $\overline{\rho}$, as given in~\eqref{eq:encoded-pauli-expansion}.
Indeed, Hamiltonian $H$ has a minimum eigenvalue of $-\left(n-k\right)$,
corresponding to a ground space that is $2^{k}$-fold degenerate.
The constraints in~\eqref{eq:constrained-energy-min-stabilizer} then
uniquely specify the state of the logical qubits, according to~\eqref{eq:encoded-pauli-expansion}. 

By transforming this constrained energy minimization problem to a
constrained free-energy minimization problem at temperature $T\geq0$,
as reviewed in Section~\ref{subsec:Review-of-constrained}, the LMPW
algorithms converge to an approximate solution of~\eqref{eq:constrained-energy-min-stabilizer}
for a sufficiently low temperature $T$. Furthermore, the HQC algorithms
at the final step generate a parameterized thermal state of the following
form:
\begin{equation}
\frac{e^{-\frac{1}{T}\left(H-\sum_{i_{1},\ldots,i_{k}\in\{0,1,2,3\}}\mu_{i_{1},\ldots,i_{k}}\overline{\sigma}_{i_{1},1}\cdots\overline{\sigma}_{i_{k},k}\right)}}{\Tr\!\left[e^{-\frac{1}{T}\left(H-\sum_{i_{1},\ldots,i_{k}\in\{0,1,2,3\}}\mu_{i_{1},\ldots,i_{k}}\overline{\sigma}_{i_{1},1}\cdots\overline{\sigma}_{i_{k},k}\right)}\right]},
\end{equation}
where the chemical potentials in $\left(\mu_{i_{1},\ldots,i_{k}}\right)_{i_{1},\ldots,i_{k}\in\{0,1,2,3\}}$
ensure that the constraints in~\eqref{eq:constrained-energy-min-stabilizer}
are satisfied. Thus, in this way, the LMPW HQC algorithms
can be viewed as methods for approximately encoding logical qubit
states into stabilizer thermodynamic systems, whenever low-temperature
thermal state preparation methods are experimentally accessible.

\subsection{Warm-starting the encoding of a single qubit}

\label{sec:warm-start}

The LMPW HQC algorithms start from the initial value $\mu=\left(0,\ldots,0\right)$,
as seen in Algorithms~\ref{alg:HQC-LMPW-1st-order} and~\ref{alg:HQC-LMPW-2nd-order}.
Such an initial starting point is used whenever there is no information
about the solution to the problem at hand. However, in the case of
encoding a single logical qubit into multiple physical qubits, we
can provide a much better starting point for the algorithms, which
speeds up the convergence significantly.

To see this, suppose that $\rho$ is a qubit state characterized by
its Bloch vector $r\coloneqq\left(r_{x},r_{y},r_{z}\right)\in\mathbb{R}^{3}$:
\begin{equation}
\rho=\frac{1}{2}\left(I+r_{x}X+r_{y}Y+r_{z}Z\right).\label{eq:bloch-rep}
\end{equation}
If $\left\Vert r\right\Vert =1$, then $\rho$ is a pure state, and
if $\left\Vert r\right\Vert < 1$, then $\rho$ is a mixed state.
Whenever $\rho$ is a mixed state, it can alternatively be written
in terms of a chemical potential vector $\mu\coloneqq\left(\mu_{x},\mu_{y},\mu_{z}\right)\in\mathbb{R}^{3}$
and an inverse temperature parameter $\beta\geq0$ as follows:
\begin{equation}
\rho=\frac{1}{Z_{\beta}(\mu)}\exp\!\left(-\beta\left(\mu_{x}X+\mu_{y}Y+\mu_{z}Z\right)\right),\label{eq:exponential-coord-qubit}
\end{equation}
where
\begin{equation}
Z_{\beta}(\mu)\coloneqq\Tr\!\left[\exp\!\left(-\beta\left(\mu_{x}X+\mu_{y}Y+\mu_{z}Z\right)\right)\right].
\end{equation}
The coordinates specified by $r$ are known as mixture coordinates,
and those specified by $\mu$ and $\beta$ are known as exponential coordinates
\cite[Section~8.2]{Bengtsson2006}. The following theorem, proven
in Appendix~\ref{app:mixture-to-exp-proof}, establishes a mapping
from mixture coordinates to exponential coordinates for a single-qubit
state:
\begin{thm}
\label{thm:mixture-to-exp-qubit}Let $r$ be a Bloch vector corresponding
to a qubit density matrix $\rho$, as written in~\eqref{eq:bloch-rep},
such that $\left\Vert r\right\Vert <1$. Then, by choosing
\begin{align}
\mu & =r,\label{eq:warm-start-qubit-1}\\
\beta & =\frac{\arctanh\!\left(-\left\Vert r\right\Vert \right)}{\left\Vert r\right\Vert },\label{eq:warm-start-qubit-2}
\end{align}
the following equality holds:
\begin{multline}
\frac{1}{2}\left(I+r_{x}X+r_{y}Y+r_{z}Z\right)=\\
\frac{1}{Z_{\beta}(\mu)}\exp\!\left(-\beta\left(\mu_{x}X+\mu_{y}Y+\mu_{z}Z\right)\right).\label{eq:desired-eq-mix-exp}
\end{multline}
\end{thm}

Equipped with Theorem~\ref{thm:mixture-to-exp-qubit}, we can initialize
the LMPW HQC algorithms with the following state:
\begin{equation}
\frac{e^{-\frac{1}{T}\left(H+\beta T\left(\mu_{x}\overline{X}+\mu_{y}\overline{Y}+\mu_{z}\overline{Z}\right)\right)}}{\Tr\!\left[e^{-\frac{1}{T}\left(H+\beta T\left(\mu_{x}\overline{X}+\mu_{y}\overline{Y}+\mu_{z}\overline{Z}\right)\right)}\right]},
\label{eq:warm-start-qubit-encoded}
\end{equation}
where $H$ is the stabilizer Hamiltonian for the code (as in~\eqref{eq:stabilizer-Ham},
with $\gamma_{i}=1$ for all $i\in[n-k]$), $\overline{X}$, $\overline{Y}$,
and $\overline{Z}$ are the logical operators for the encoded logical
qubit, and $\mu$ and $\beta$ are chosen as in~\eqref{eq:warm-start-qubit-1}--\eqref{eq:warm-start-qubit-2}. The extra factor of $T$ multiplying $\beta$ in~\eqref{eq:warm-start-qubit-encoded} ensures that the constraints in~\eqref{eq:constrained-energy-min-stabilizer} are satisfied identically. In our simulations, we find that this initial starting point leads
to immediate convergence of the LMPW HQC algorithms. In fact, the state in \eqref{eq:warm-start-qubit-encoded} is actually optimal, as a special case of the following more general theorem, which is proven in Appendix~\ref{app:proof-solution-stab-thermo}:

\begin{thm}
\label{thm:solution-stab-thermo}
Consider a stabilizer thermodynamic system with Hamiltonian
\begin{equation}
H=-\sum_{i=1}^{n-k}S_{i},
\end{equation}
and logical operators $\overline{\sigma}_{i_{1},1}\cdots\overline{\sigma}_{i_{k},k}$
for all $i_{1},\ldots,i_{k}\in\left\{ 0,1,2,3\right\} $. Suppose
that there is a mixed (non-pure) state $\rho$ of $k$ qubits satisfying
the following constraints:
\begin{equation}
\Tr\!\left[\left(\sigma_{i_{1}}\otimes\cdots\otimes\sigma_{i_{k}}\right)\rho\right]=r_{i_{1},\ldots,i_{k}}\label{eq:k-qubit-constraints-for-state}
\end{equation}
for all $i_{1},\ldots,i_{k}\in\left\{ 0,1,2,3\right\} $. (As indicated
in~\eqref{eq:pauli-expansion-k-qubits}, the constraints in \eqref{eq:k-qubit-constraints-for-state}
completely specify the $k$-qubit state.) Suppose that the coefficients
in $\left(\mu_{i_{1},\ldots,i_{k}}\right)_{i_{1},\ldots,i_{k}\in\left\{ 0,1,2,3\right\} }$
are such that
\begin{equation}
\rho=\frac{\exp\!\left(\sum_{i_{1},\ldots,i_{k}}\mu_{i_{1},\ldots,i_{k}}\sigma_{i_{1}}\otimes\cdots\otimes\sigma_{i_{k}}\right)}{\Tr\!\left[\exp\!\left(\sum_{i_{1},\ldots,i_{k}}\mu_{i_{1},\ldots,i_{k}}\sigma_{i_{1}}\otimes\cdots\otimes\sigma_{i_{k}}\right)\right]}.\label{eq:unencoded-state-thermal-form}
\end{equation}
Then the following state
\begin{equation}
\frac{\exp\!\left(-\frac{1}{T}\left(H-T\sum_{i_{1},\ldots,i_{k}}\mu_{i_{1},\ldots,i_{k}}\overline{\sigma}_{i_{1}}\cdots\overline{\sigma}_{i_{k}}\right)\right)}{\Tr\!\left[\exp\!\left(-\frac{1}{T}\left(H-T\sum_{i_{1},\ldots,i_{k}}\mu_{i_{1},\ldots,i_{k}}\overline{\sigma}_{i_{1}}\cdots\overline{\sigma}_{i_{k}}\right)\right)\right]}
\end{equation}
is the unique solution to the following constrained free energy minimization
problem:
\begin{equation}
\min_{\omega\in\mathcal{D}^{2^{n}}}\left\{ \begin{array}{c}
\Tr[H\omega]-TS(\omega):\\
\Tr\!\left[\left(\overline{\sigma}_{i_{1}}\cdots\overline{\sigma}_{i_{k}}\right)\omega\right]=r_{i_{1},\ldots,i_{k}}
\end{array}\right\} .\label{eq:SDP-QECC}
\end{equation}
\end{thm}

As expanded upon in Remark~\ref{rem:stab-thermo-encoding} of Appendix~\ref{app:proof-solution-stab-thermo}, the intuition behind Theorem~\ref{thm:solution-stab-thermo} is that, if the state of the logical qubits is completely specified by the constraints, then these degrees of freedom are not available for minimizing the free energy of the system, and one can only do so by initializing the ancilla qubits to temperature $T$ thermal states. 

\section{Simulation results}

\label{sec:Simulation-results}

\begin{figure*}
\includegraphics[width=\linewidth]{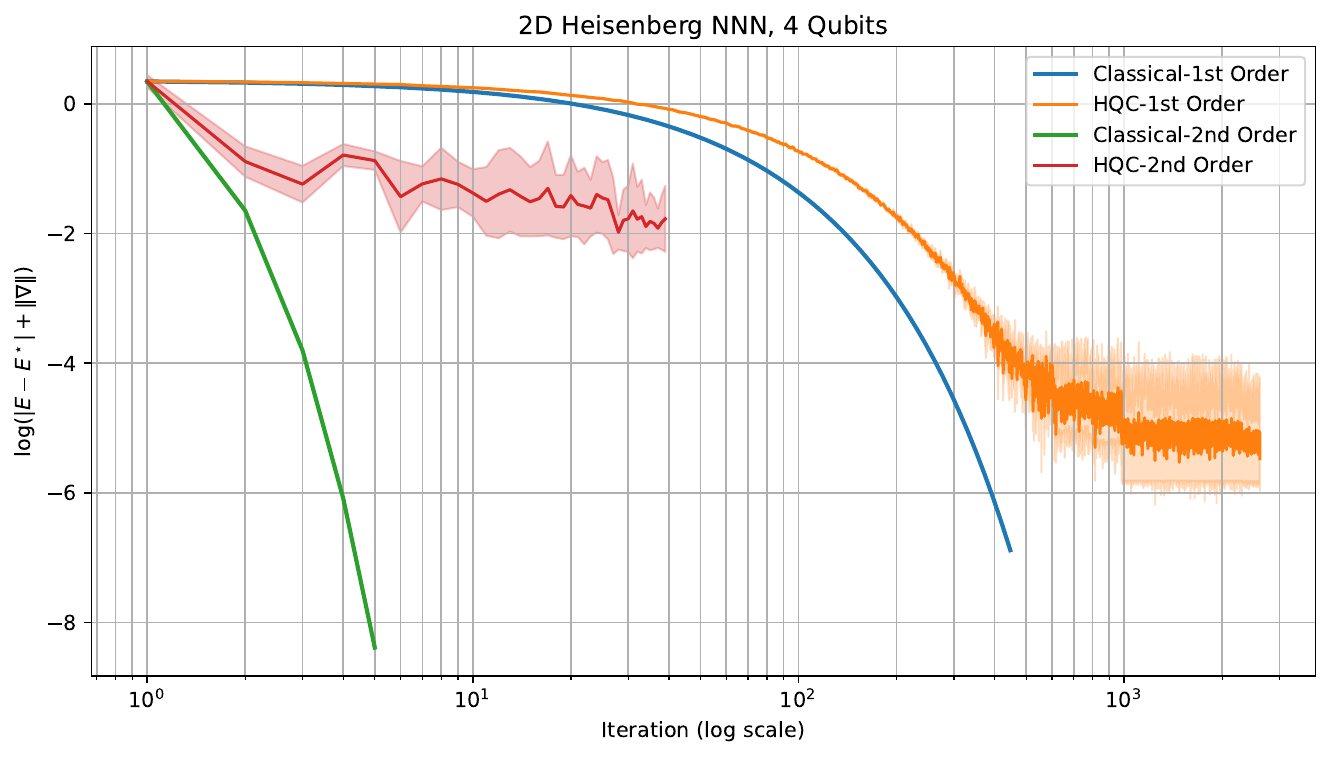}
\caption{The figure depicts the logarithm of the error metric in~\eqref{eq:error-metric} versus the logarithm of the number of iterations,  for the task of constrained energy minimization for the two-dimensional, four-qubit quantum Heisenberg model with nearest- and next-nearest-neighbor interactions and constraints on the total magnetizations in the $x$, $y$, and $z$ directions set to be 1, 0, and 1, respectively.  All of the algorithms converge, but the HQC algorithms, shown as the average over five independent runs with shaded regions denoting one standard deviation, require more iterations to converge due to sampling noise inherent in them.}
\label{fig:2D-heis-nnn_4q}
\end{figure*}

In this section, we report the results
of simulating the LMPW algorithms on several models of interest
in quantum thermodynamics and the design of thermal states. Here,
we focus on the one- and two-dimensional quantum Heisenberg models
with a Hamiltonian having nearest- and next-nearest-neighbor interactions
and the conserved, non-commuting charges set to the total magnetizations
in the $x$, $y$, and $z$ directions (see Section~\ref{subsec:Quantum-Heisenberg-models}).
We also report the results of simulations of a stabilizer thermodynamic
system built from the perfect one-to-five qubit quantum error-correcting
code.

In Appendix~\ref{sec:Other-simulation-results}, we report the results
of other simulations, which include the one- and two-dimensional quantum
Heisenberg models with nearest-neighbor interactions and similar settings
as above for the charges. Therein, we also report simulation results
for other stabilizer thermodynamic systems, including the one-to-three
qubit repetition code and the two-to-four qubit quantum error-detecting
code.

For all of our numerical experiments, we set the error metric for the constrained energy minimization problems to be as follows:
\begin{equation}
\left| E^\star - \tilde{E}_m\right | + \left\| \tilde{\nabla}_m\right\|,
\label{eq:error-metric}
\end{equation}
in order to track the performance of the algorithms.
In~\eqref{eq:error-metric},  $E^\star$ is the true minimum energy output from a standard SDP solver (specifically, we employed CVXPY~\cite{diamond2016cvxpy} for these benchmarks) and $\tilde{E}_m$ and $\tilde{\nabla}_m$ are the estimates of the energy and the gradient, respectively, at the $m$th step of the algorithm. If $\left| E^\star - \tilde{E}_m\right |$ is small, then we are guaranteed that the algorithm has output a good approximation of the true energy at the $m$th step, and if $\left\| \tilde{\nabla}_m\right\|$ is small, then, according to~\eqref{eq:gradient-elements}, we are guaranteed that the algorithm has found a solution that satisfies the constraints in~\eqref{eq:constrained-energy-min-def} approximately. In all of our figures, we plot the logarithm of the error metric in~\eqref{eq:error-metric}, in order to make it easier to visualize the performance of the algorithms. In addition to using this error metric for plotting purposes, we also used the gradient norm $\left\| \tilde{\nabla}_m \right\|$ as the stopping criterion to determine algorithm convergence, terminating the optimization once it fell below a specified threshold.

\begin{figure*}
\includegraphics[width=\linewidth]{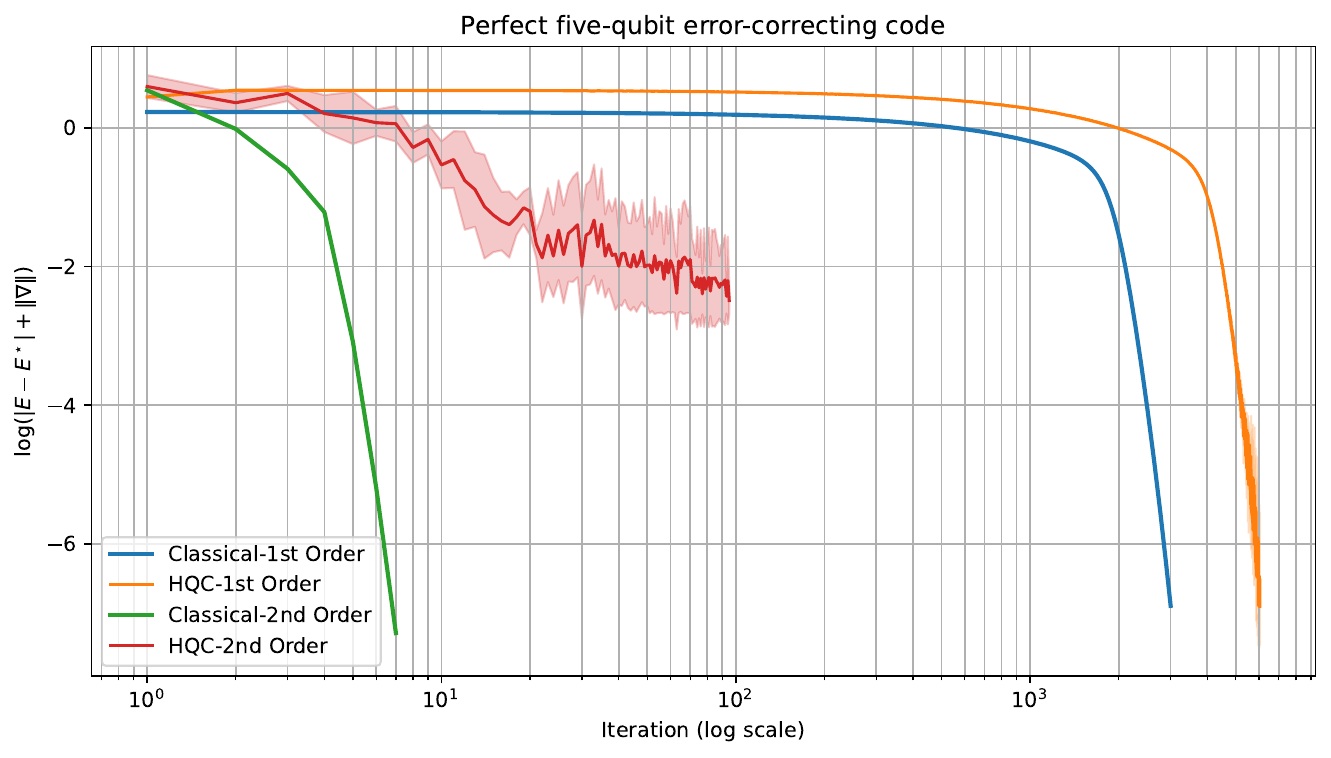}
\caption{The figure depicts the logarithm of the error metric in~\eqref{eq:error-metric} versus the number of iterations, for all of the LMPW algorithms (1st- and 2nd-order, and classical and HQC) for the task of constrained energy minimization for a stabilizer thermodynamic system formed from the perfect five-qubit error-correcting code. The constraints on the logical operators $\overline{X}$, $\overline{Y}$, and $\overline{Z}$ were set to $0.2$, 0, and $0.5$, respectively.  All of the algorithms converge, but the HQC algorithms, shown as the average over five independent runs with shaded regions denoting one standard deviation, require more iterations to converge due to sampling noise inherent in them. For this simulation, we did not warm-start the algorithm according to the recipe from Section~\ref{sec:warm-start}, but we instead started with $\mu_x = 1$, $\mu_y = 1$, and $\mu_z = 1$.
}
\label{fig:Avg-QEC-Perfect-Code-HQC}
\end{figure*}

Our first set of simulations involves the two-dimensional quantum Heisenberg
model with nearest- and next-nearest-neighbor interactions and
with the charges set to the total magnetizations
in the $x$, $y$, and $z$ directions. The goal is to solve
the constrained energy minimization problem from Section~\ref{subsec:Review-of-constrained},
with the expectation values on the total magnetizations in the $x$,
$y$, and $z$ directions set to be $1$, $0$, and $1$, respectively. Here we take the number
of qubits to be fixed at four, arranged on a $2\times2$ square lattice.
Figure~\ref{fig:2D-heis-nnn_4q} compares the performance of all the LMPW algorithms,
including the first- and second-order, classical and HQC algorithms, plotting the error metric in \eqref{eq:error-metric} versus the number of iterations. For the first- and second-order HQC algorithms, the average of five runs, along with one standard deviation shaded, is plotted.
The figure demonstrates that all algorithms converge.

The classical and HQC second-order algorithms employ backtracking, as discussed after Algorithm~\ref{alg:classical-LMPW-2nd-order}. The HQC algorithms employ sampling, resulting in noisy estimates of the gradient and, for the second-order algorithm, the Hessian. We used approximately $10^4$ samples per iteration to estimate expectation values of observables in both the first-order and second-order HQC algorithms, and $10^7$ samples per iteration were employed to estimate the Hessian for the second-order HQC algorithm. To reduce the computational cost of Hessian estimation, we parallelized the simulation across multiple CPU cores using process-level parallelism. Each process evaluates a subset of Hadamard-test circuits with a dedicated, single-threaded simulator to avoid resource contention. The circuits are pre-transpiled with placeholders for the time-evolution unitaries, which are dynamically inserted at runtime using cached matrix exponentials. This strategy significantly accelerates the Hessian estimation step while keeping memory usage and inter-process contention low. As shown in the figure, the second-order HQC exhibits a higher variance across independent runs, primarily due to the statistical and numerical challenges involved in estimating and inverting the Hessian. In particular, the inverse Hessian estimate is significantly more sensitive to noise and often suffers from numerical instabilities, which can amplify fluctuations in the update direction. To mitigate the computational cost of Hessian estimation, we adopted a looser convergence threshold for the gradient norm in the second-order HQC algorithm. This explains the comparatively higher error metric reached by the second-order HQC algorithm: unlike the other algorithms, whose faster runtime permitted more iterations and finer convergence, the second-order HQC was terminated earlier to maintain computational feasibility.

The figure clearly illustrates that incorporating second-derivative information, as done in the second-order algorithms, improves performance, i.e., requiring much fewer iterations to reach convergence compared to the first-order methods. However, as stated previously, the second-order algorithms require more computations per iteration than the first-order algorithms, and thus a trade-off is present.

Our next set of simulations reported here involve a stabilizer thermodynamic
system constructed from the perfect one-to-five qubit code. Specifically,
recall that a generating set for the code's stabilizer consists of
the following operators \cite[Section~8.2]{Gottesman1997}:
\begin{align}
S_{1} & \coloneqq X\otimes Z\otimes Z\otimes X\otimes I,\\
S_{2} & \coloneqq I\otimes X\otimes Z\otimes Z\otimes X,\\
S_{3} & \coloneqq X\otimes I\otimes X\otimes Z\otimes Z,\\
S_{4} & \coloneqq Z\otimes X\otimes I\otimes X\otimes Z,
\end{align}
and the code's logical operators are as follows:
\begin{align}
\overline{X} & \coloneqq X\otimes X\otimes X\otimes X\otimes X,\\
\overline{Y} & \coloneqq Y\otimes Y\otimes Y\otimes Y\otimes Y,\\
\overline{Z} & \coloneqq Z\otimes Z\otimes Z\otimes Z\otimes Z.
\end{align}
As such, the Hamiltonian for the stabilizer thermodynamic system is
as follows:
\begin{equation}
H\coloneqq-\left(S_{1}+S_{2}+S_{3}+S_{4}\right),
\end{equation}
and the conserved, non-commuting charges are set to $\overline{X}$,
$\overline{Y}$, and $\overline{Z}$. In this case, the thermal state
of the system at temperature $T$ is given by
\begin{equation}
\frac{e^{-\frac{1}{T}\left(H-\mu_{x}\overline{X}-\mu_{y}\overline{Y}-\mu_{z}\overline{Z}\right)}}{\Tr[e^{-\frac{1}{T}\left(H-\mu_{x}\overline{X}-\mu_{y}\overline{Y}-\mu_{z}\overline{Z}\right)}]},
\end{equation}
where $\left(\mu_{x},\mu_{y},\mu_{z}\right)\in\mathbb{R}^{3}$ is
the chemical potential vector. According to the discussion after~\eqref{eq:constrained-energy-min-stabilizer}, the ground energy of this stabilizer thermodynamic system is equal to $-(n-k)=-4$, because $n=5$ and $k=1$ in this case.

\begin{figure}
\includegraphics[width=\linewidth]{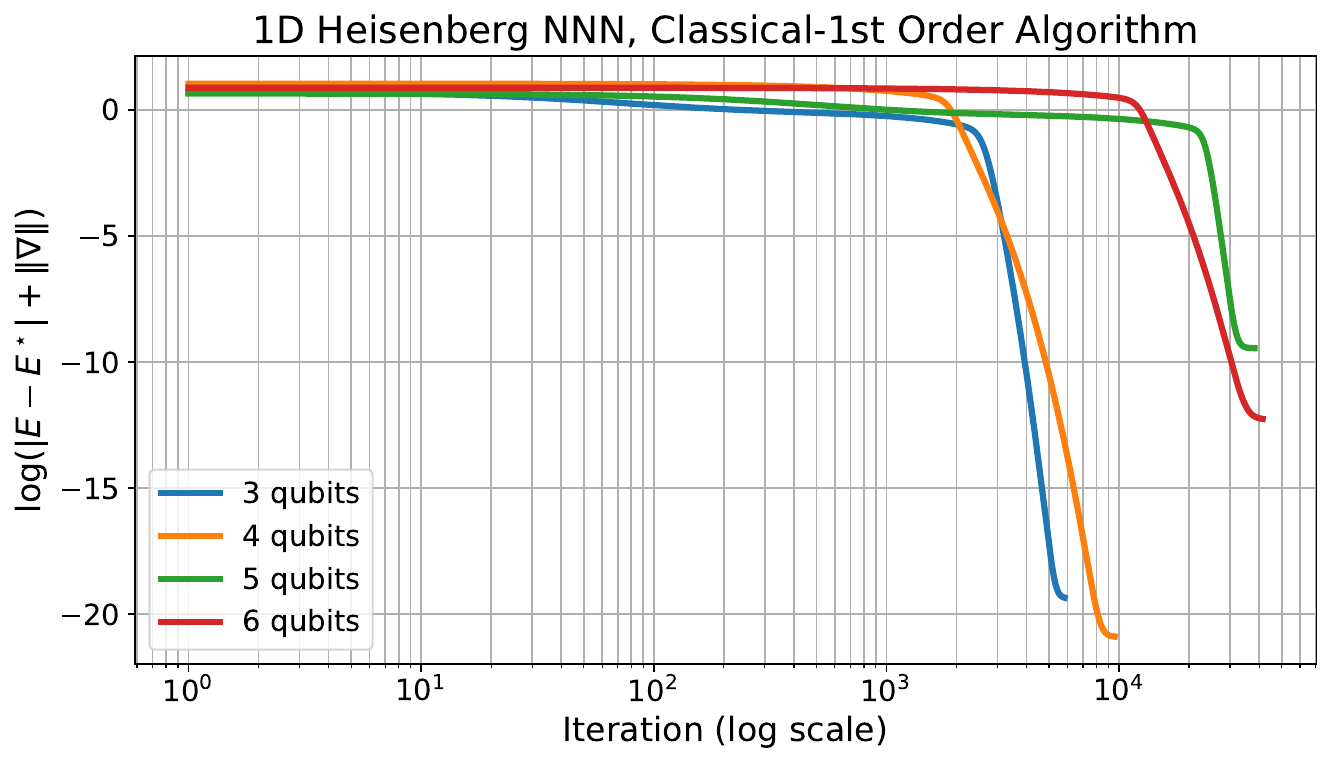}
\caption{The figure depicts the logarithm of the error metric in~\eqref{eq:error-metric} versus the number of iterations,  for the task of constrained energy minimization for the one-dimensional quantum Heisenberg model with nearest- and next-nearest-neighbor interactions and constraints on the total magnetizations in the $x$, $y$, and $z$ directions set to be 1, 0, and 1, respectively.  The LMPW classical algorithm converges in all cases (3, 4, 5, and 6 qubits), with larger systems taking more iterations to converge.}
\label{fig:1D-heis-nnn-1st-ord-classical-3-5-7-qubits-Log-Error}
\end{figure}

Figure~\ref{fig:Avg-QEC-Perfect-Code-HQC} plots the performance of the LMPW algorithms for the
above stabilizer thermodynamic system. Again the goal is to solve
the constrained energy minimization problem from Section~\ref{subsec:Review-of-constrained}.
However, as mentioned in Section~\ref{sec:Encoding-quantum-information},
the LMPW HQC algorithms in this context can alternatively be understood
as methods for encoding quantum information into an error-correcting
code. Indeed, by setting the expectation values of the logical operators to be
\begin{align}
    \Tr[\overline{X}\rho]& =r_{x}, \\ \Tr[\overline{Y}\rho]& =r_{y},\\ 
\Tr[\overline{Z}\rho]& =r_{z},
\end{align}
where $\left(r_{x},r_{y},r_{z}\right)\in\mathbb{R}^{3}$
is a vector in the Bloch ball, the LMPW HQC algorithms encode a
logical qubit with this Bloch sphere representation. In our simulations, we choose
\begin{equation}
r_x = 0.2, \quad  r_y = 0,\quad   r_z=0.5.    
\end{equation}
The plot in Figure~\ref{fig:Avg-QEC-Perfect-Code-HQC} is qualitatively similar to that from Figure~\ref{fig:2D-heis-nnn_4q},
comparing the error metric in \eqref{eq:error-metric} versus number of iterations for all of the algorithms. The details of the implementation are similar to those of the previous example, and so we do not repeat them here.
Again we find that the second-order algorithms require fewer iterations
than do the first-order algorithms, however with the same trade-off
mentioned previously.

\begin{figure}
\includegraphics[width=\linewidth]{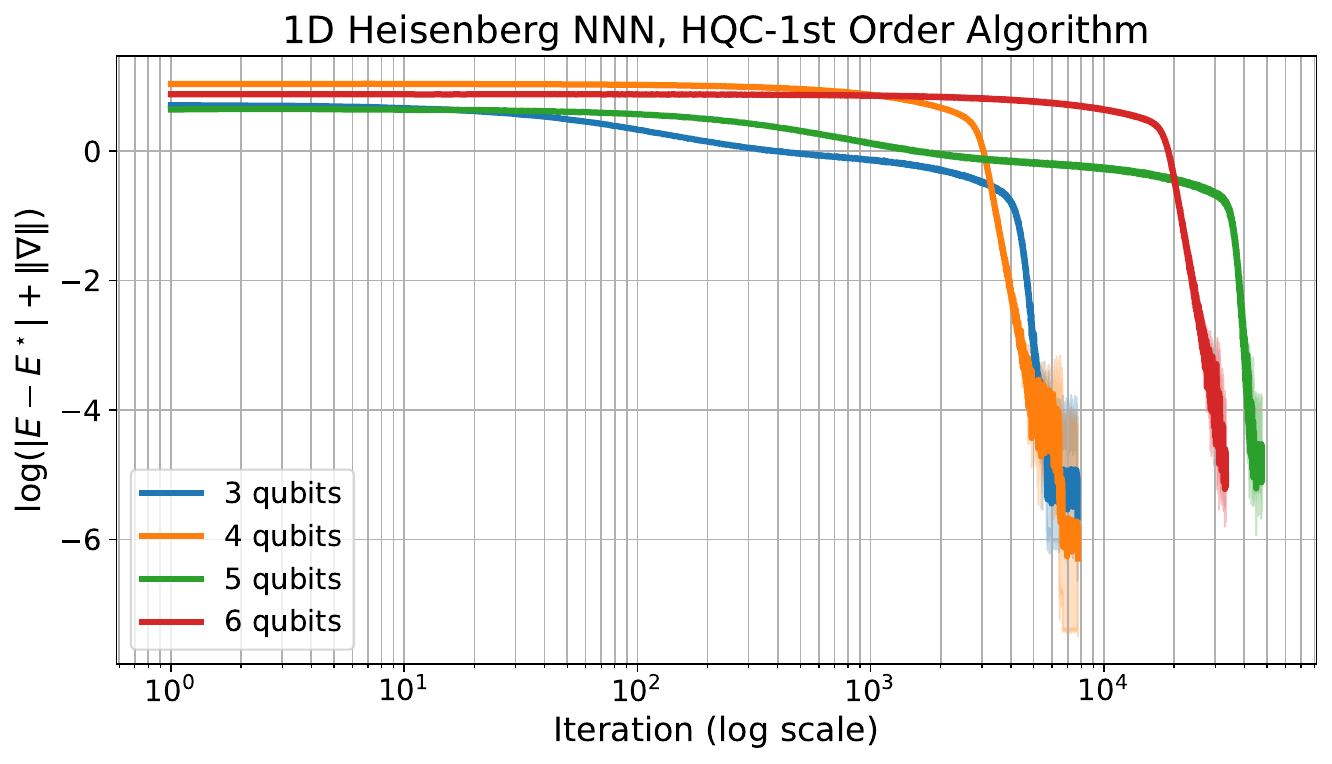}
\caption{The figure depicts the average of the logarithm of the error metric in~\eqref{eq:error-metric} over five independent runs, plotted against the number of iterations, with shaded regions indicating one standard deviation. The task is constrained energy minimization for the one-dimensional quantum Heisenberg model with nearest- and next-nearest-neighbor interactions, and constraints on the total magnetizations in the $x$, $y$, and $z$ directions set to be 1, 0, and 1, respectively.  The LMPW HQC algorithm converges in all cases (3, 4, 5, and 6 qubits), with larger systems taking more iterations to converge.
}
\label{fig:1D-heis-nnn-1st-ord-hqc-3-5-7-qubits-Log-Error}
\end{figure}

Our final set of simulations involves the one-dimensional quantum Heisenberg model with nearest- and next-nearest-neighbor interactions and with the charges set to the total magnetizations in the $x$, $y$, and $z$ directions. The goal is to solve the constrained energy minimization problem reviewed in Section~\ref{subsec:Review-of-constrained}, with the expectation values on the total magnetizations in the $x$, $y$, and $z$ directions set to be 1, 0, and 1, respectively. Figure~\ref{fig:1D-heis-nnn-1st-ord-classical-3-5-7-qubits-Log-Error} demonstrates the performance of the first-order LMPW classical algorithm as the number of qubits in the model increases (specifically, the figure plots results for three, four, five, and six qubits). The figure plots the error metric in \eqref{eq:error-metric} versus the number of iterations, demonstrating that every case converges to the true value as the number of iterations increases.

Figure~\ref{fig:1D-heis-nnn-1st-ord-hqc-3-5-7-qubits-Log-Error} demonstrates the performance of the first-order HQC algorithm for the same setup. The algorithm converges in all cases, but requires more iterations than the classical algorithm and exhibits noticeable noise in proximity to the optimal value. As expected, this behavior arises from the sampling noise inherent in the HQC algorithm. To mitigate the increased computational cost and instability associated with this noise, we employed a looser convergence threshold on the gradient norm, which also accounts for the slightly higher final error metric reached by the HQC algorithm compared to its classical counterpart. The number of shots used per iteration was set according to~\cite[Eq.~(F17)]{liu2025qthermoSDPs}, with $\varepsilon = 10^{-4}$ and $\delta = 10^{-2}$. In this case, the norm of the Pauli coefficients of the total magnetizations is $n$, where $n$ is the number of qubits.

\subsection{Comparison of Standard Gradient Ascent and Nesterov’s Accelerated Method}

\begin{figure}
\includegraphics[width=\linewidth]{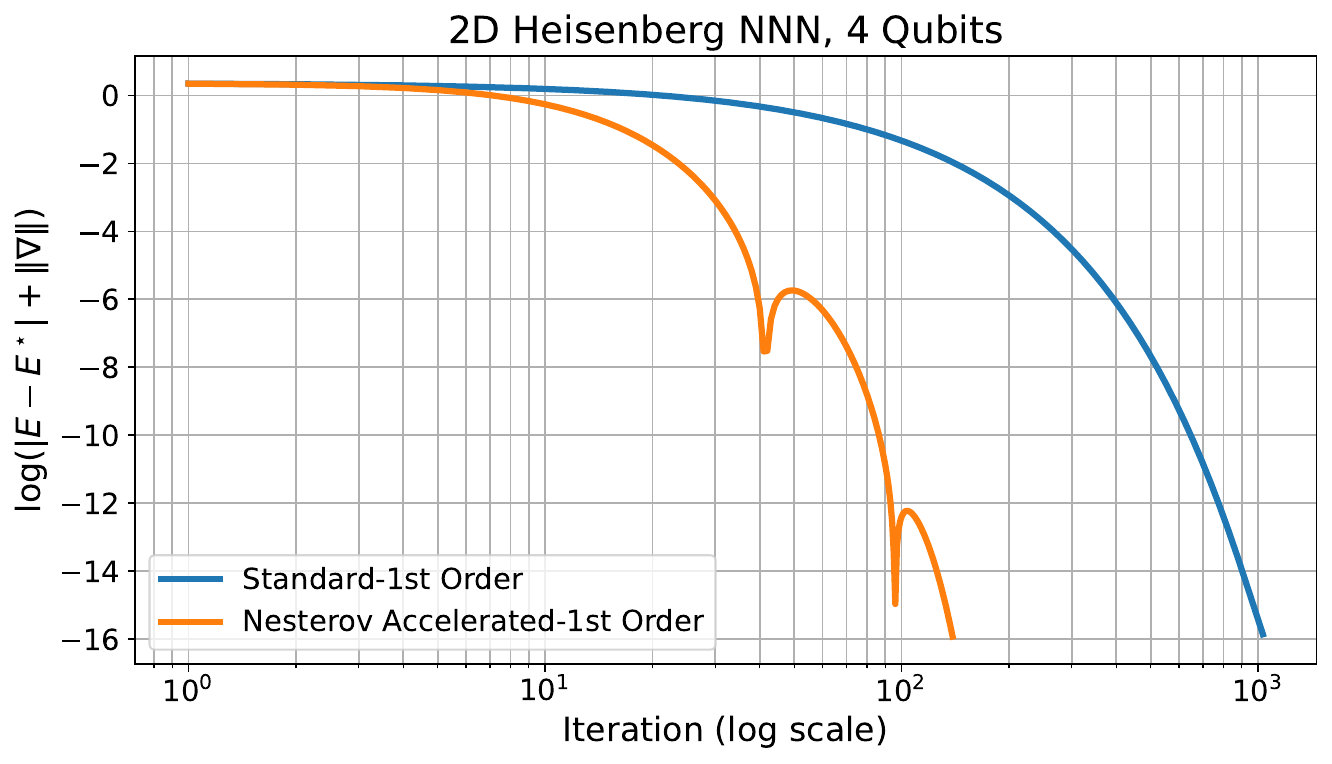}
\caption{The figure depicts the logarithm of the error metric in~\eqref{eq:error-metric} versus the number of iterations for the LMPW classical first-order algorithm, comparing the standard gradient ascent method with Nesterov’s accelerated version. The task is constrained energy minimization for the two-dimensional, four-qubit quantum Heisenberg model with nearest- and next-nearest-neighbor interactions and constraints on the total magnetizations in the $x$, $y$, and $z$ directions set to be 1, 0, and 1, respectively.
}
\label{fig:2d_heis_nnn_4q_Nesterov_Standard}
\end{figure}

\begin{figure}
\includegraphics[width=\linewidth]{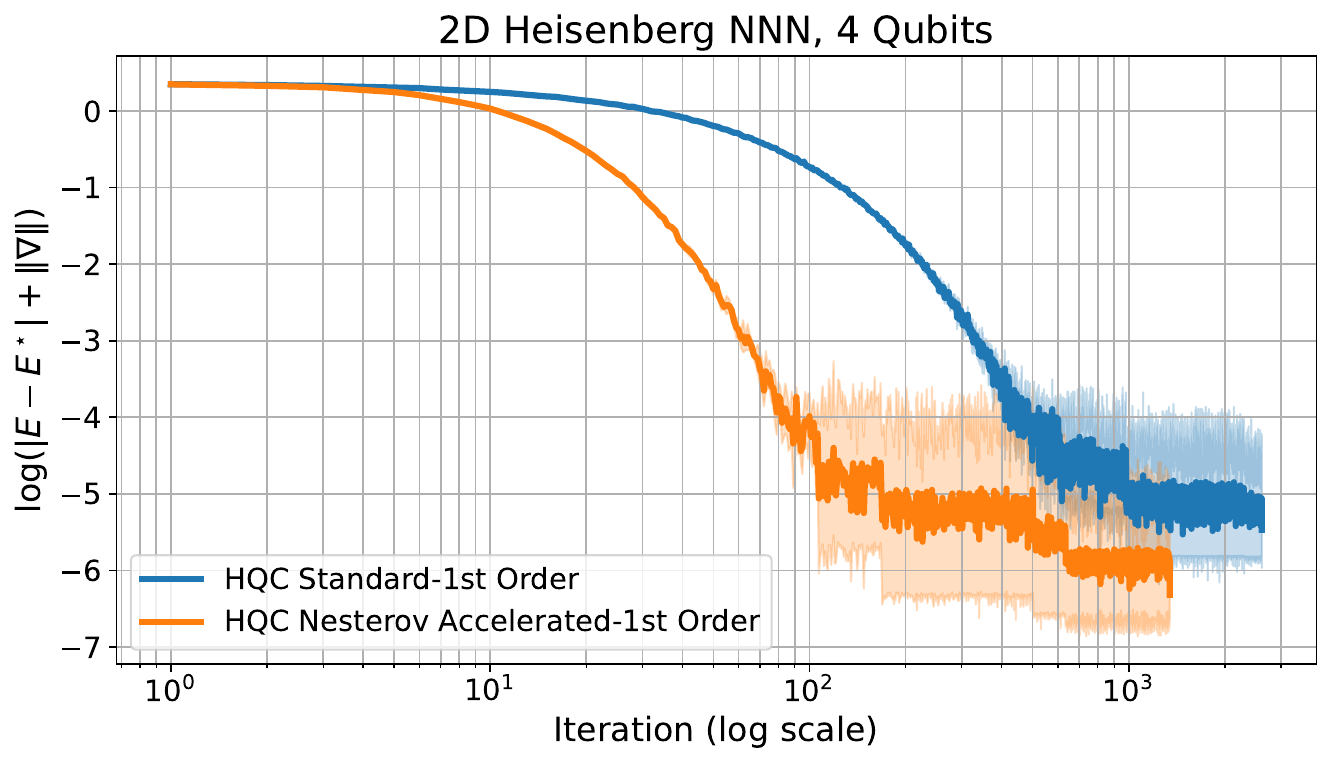}
\caption{The figure depicts the average of the logarithm of the error metric in~\eqref{eq:error-metric} over five independent runs, plotted against the number of iterations, with shaded regions indicating one standard deviation, for the LMPW HQC first-order algorithm, comparing the standard gradient ascent method with Nesterov’s accelerated version. The task is constrained energy minimization for the two-dimensional, four-qubit quantum Heisenberg model with nearest- and next-nearest-neighbor interactions and constraints on the total magnetizations in the $x$, $y$, and $z$ directions set to be 1, 0, and 1, respectively.
}
\label{fig:2d_heis_nnn_4q_Nesterov_Standard_HQC}
\end{figure}

Throughout this work, we employed the standard gradient ascent algorithm for both the classical and HQC implementations of the first-order optimization procedure. However, several well-known modifications of the standard gradient ascent method can lead to improved convergence performance, especially in concave settings. One such modification is Nesterov’s accelerated gradient method~\cite{Nesterov1983}, which introduces a momentum term to anticipate the direction of the next update, resulting in improved convergence speed.

While the standard gradient ascent update rule, given in~\eqref{eq:gradient-update-LMPW}, directly updates the parameters based on the gradient evaluated at the current iterate, Nesterov’s accelerated gradient ascent maintains an auxiliary sequence $v_m$, which incorporates a momentum term, and performs the update as follows:
\begin{align}
    v_{m+1} &= \tau v_m + \eta \nabla f(\mu_m + \tau v_m), \\
    \mu_{m+1} &= \mu_m + v_{m+1},
\end{align}
where $\tau \in [0, 1)$ is the momentum parameter and $\eta >0$ is the learning rate. The key idea is to compute the gradient not at the current iterate $\mu_m$, but at the projected position $\mu_m + \tau v_m$, effectively introducing a look-ahead mechanism that often leads to faster convergence.

Figure~\ref{fig:2d_heis_nnn_4q_Nesterov_Standard} illustrates the performance of the standard and Nesterov-accelerated gradient ascent algorithms for the classical first-order optimization of the two-dimensional four-qubit Heisenberg model with nearest- and next-nearest-neighbor interactions. While both methods converge to the same solution, the Nesterov-accelerated method does so in significantly fewer iterations.
Figure~\ref{fig:2d_heis_nnn_4q_Nesterov_Standard_HQC} presents the corresponding results for the HQC first-order optimization applied to the same task, again demonstrating that the Nesterov-accelerated method converges in significantly fewer iterations than does the standard gradient ascent method.

\section{Conclusion}

\label{sec:Conclusion}In this paper, we put forward several conceptual
contributions, including framing constrained energy minimization as
a method for designing ground states of Hamiltonians and constrained
free energy minimization as a method for designing thermal states
of Hamiltonians. We argued how the LMPW HQC algorithms can be used
for designing these Hamiltonians and how this application has implications
for designing molecules and materials with desired properties -- specifically by enabling the precise targeting of specific electronic sectors (defined by particle number, spin, or other charges), which is a prerequisite for calculating the properties of molecular isomers and magnetic materials. We
also introduced stabilizer thermodynamic systems as thermodynamic
systems for which the Hamiltonian is constructed from the stabilizer
operators of a stabilizer code and the conserved, non-commuting charges
are constructed from the code's logical operators. Then we demonstrated
how the LMPW HQC algorithms can be understood as alternative methods
for encoding quantum information into quantum error-correcting codes,
which could be a useful approach whenever low-temperature thermal state preparation methods are experimentally accessible.

The numerical contribution of our paper involved simulating the LMPW
algorithms for several key thermodynamic systems of interest, including
the one- and two-dimensional quantum Heisenberg models with nearest-neighbor
and next-nearest-neighbor interactions and several stabilizer thermodynamic
systems (the one-to-three qubit repetition code, the two-to-four qubit
error-detecting code, and the perfect one-to-five qubit code). To
test the performance of the LMPW algorithms on all of these examples,
we employed standard SDP solvers to determine the true values of the
constrained energy minimization problems. The LMPW classical algorithms
can be understood as establishing alternative benchmarks for the HQC
algorithms: the former algorithms act as noiseless simulators of the
HQC algorithms, which focus exclusively on converging when measurement
shot noise is not present, while our simulations of the latter incorporate
measurement shot noise.

Finally, we comment on the robustness of the LMPW algorithms in the low-temperature regime ($T \to 0$), which is required for high-precision ground state approximation. As discussed in Section~\ref{sec:Design-of-Hamiltonians}, the smoothness parameter of the dual objective function scales as $L \propto 1/T$. This implies that as the temperature decreases, the optimization landscape becomes stiffer, necessitating smaller step sizes ($\eta \propto T$) and leading to slower convergence rates. Additionally, for the second-order algorithms, the condition number of the Hessian increases, making the update steps more sensitive to sampling noise in the HQC setting. However, a key advantage of the LMPW framework is the guaranteed concavity of the dual landscape, which holds for all $T > 0$. Unlike non-convex variational ansatzes that may get trapped in local minima as the system freezes, the LMPW algorithms are guaranteed to converge to the global optimum. Thus, while the computational cost (in terms of iterations and samples) increases at very low temperatures, the algorithms remain robust and do not suffer from the breakdown of convergence often observed in other variational methods.

Furthermore, our framework opens new avenues for exploring the interplay between thermodynamic geometry and quantum dynamics. For instance, the strictly concave geometry of the dual optimization landscape established in our work could provide a constructive method to probe phase transitions associated with the elimination of dynamical symmetries in systems with non-commuting charges~\cite{Majidy_2024}. Investigating how this geometry changes as non-commuting terms are introduced may offer deeper insights into why non-commuting charges can distinctly hinder or promote thermalization depending on the physical context.

Going forward from here, there are several avenues to pursue for future
research. First, we think it would be worthwhile to employ matrix
product states for simulating the thermal states needed in the LMPW
classical algorithms, as used in prior work~\cite{Huang2021,Alhambra2021,Alhambra2023}.
Very likely such an approach would lead to an advanced classical algorithm
and serve several purposes: first, it could be a useful algorithm
on its own for understanding quantum thermodynamic systems; second,
it could be helpful for solving or approximating the design of ground
and thermal states of controllable Hamiltonians; third, it could be
useful as a benchmark to set against the LMPW HQC algorithms. Regarding
the latter point, any claimed quantum speedup for the LMPW HQC algorithms
would have to outperform any possible classical algorithm for this
task, and it is likely that approaches involving matrix product states would
lead to the fastest possible classical algorithms.

We also left it open to test the performance of the LMPW HQC algorithms
on currently available quantum computers. In doing so, it would be necessary to incorporate
the techniques of quantum error correction and mitigation in order
to limit the effects of noise, and one would also need to incorporate
various advanced methods for thermal state preparation~\cite{chen2023q_Gibbs_sampl,chen2023thermalstatepreparation,rajakumar2024gibbssampling,bergamaschi2024gibbs_sampling,chen2024sim_Lindblad,rouze2024efficientthermalization,bakshi2024hightemperaturegibbsstates,ding2024preparationlowtemperaturegibbs}
in order to have efficient approaches for this necessary step. As
discussed in~\cite{liu2025qthermoSDPs}, in order to have an efficient
HQC algorithm, it would be necessary to find a quantum thermodynamic
system (Hamiltonian and conserved charges) for which one can efficiently
prepare low-temperature thermal states.

\begin{acknowledgments}
We are grateful to Nicole Yunger Halpern
for a helpful discussion on quantum thermodynamics in the presence
of conserved, non-commuting charges. We also thank Zo\"e Holmes,
Peter McMahon, Karan Mehta, Erich Mueller, and Darren Pereira for helpful discussions.
MM acknowledges support of the NCCR MARVEL, a National Centre of Competence in Research, funded by the Swiss National Science Foundation (grant number 205602).
MC and JK acknowledge support from the Cornell Engineering Learning
Initiative program. IL, ML, and KW are grateful to the Cornell School of Electrical and Computer Engineering for hospitality during their summer internships.  NL acknowledges funding from the Science and Technology
Commission of Shanghai Municipality (STCSM) grant no.~24LZ1401200
(21JC1402900), NSFC grants no.~12471411 and no.~12341104, the Shanghai
Jiao Tong University 2030 Initiative, and the Fundamental Research
Funds for the Central Universities. MMW acknowledges support from
the National Science Foundation under grant no.~2329662 and from
the Cornell School of Electrical and Computer Engineering.
\end{acknowledgments}

\section*{Author Contributions}
\noindent
\textbf{Author Contributions}: The following describes the
different contributions of the authors of this work, using
roles defined by the CRediT (Contributor Roles Taxonomy) project~\cite{CRediT}:
\medskip 

\noindent\textbf{MM:} Data curation, Investigation, Resources, Software, Supervision, Validation,  Visualization, Writing - review \& editing.
\smallskip

\noindent\textbf{MC:} Data curation, Resources, Software, Validation,  Visualization, Writing - review \& editing.
\smallskip

\noindent\textbf{JK:} Data curation, Resources, Software, Validation, Visualization, Writing - review \& editing.
\smallskip

\noindent\textbf{NL:} Conceptualization, Methodology, Writing - review \& editing.
\smallskip

\noindent\textbf{IL:} Data curation, Resources, Software, Validation,  Visualization, Writing - review \& editing.
\smallskip

\noindent\textbf{ML:} Data curation, Resources, Software, Validation,  Visualization.
\smallskip

\noindent\textbf{SR:} Software, Supervision, Validation.
\smallskip

\noindent\textbf{KW:} Data curation, Resources, Software, Validation,  Visualization, Writing - review \& editing.
\smallskip

\noindent\textbf{MMW:} Conceptualization, Formal Analysis, Funding acquisition, Investigation, Methodology, Project Administration, Supervision, Validation, Writing - original draft, Writing - review \& editing.

\bibliography{Ref}

\appendix

\section{Closeness of low-temperature thermal states and ground states of
Hamiltonians}

\label{sec:Closeness-of-low-temperature}

How well does the thermal
state of a Hamiltonian approximate its ground state? Alternatively,
in the case of a degenerate Hamiltonian, how well does the thermal
state of a Hamiltonian approximate the maximally mixed state of its
ground space? It is well known that the approximation becomes better as $\beta\to\infty$,
where $\beta\coloneqq\frac{1}{T}$ is the inverse temperature.

In this appendix, we answer this question quantitatively in terms
of several standard distinguishability measures for quantum states:
the normalized trace distance (Proposition~\ref{prop:trace-dist-thermal-ground}),
the fidelity (Proposition~\ref{prop:fid-thermal-ground}), the quantum
relative entropy, and quantum R\'enyi relative entropies (Proposition
\ref{prop:rel-ent-thermal-ground}). Recall that these measures are
respectively defined for states $\rho$ and $\sigma$ and $\alpha\in(0,1)\cup(1,\infty)$
as follows:
\begin{align}
\frac{1}{2}\left\Vert \rho-\sigma\right\Vert _{1} & ,\\
F(\rho,\sigma) & \coloneqq\left\Vert \sqrt{\rho}\sqrt{\sigma}\right\Vert _{1}^{2},\\
D(\rho\|\sigma) & \coloneqq\Tr[\rho\left(\ln\rho-\ln\sigma\right)],\\
D_{\alpha}(\rho\|\sigma) & \coloneqq\frac{1}{\alpha-1}\ln\Tr[\rho^{\alpha}\sigma^{1-\alpha}],\label{eq:petz-renyi-def}
\end{align}
where $\left\Vert A\right\Vert _{1}\coloneqq\Tr\!\left[\sqrt{A^{\dag}A}\right]$.
The R\'enyi relative entropy in~\eqref{eq:petz-renyi-def} is the
Petz--R\'enyi relative entropy~\cite{Petz1986}, and we note that
there are other quantum R\'enyi relative entropies, including the
sandwiched R\'enyi relative entropy~\cite{MuellerLennert2013,Wilde2014}
and the geometric R\'enyi relative entropy~\cite{Matsumoto2013,Fang2021},
defined for $\alpha\in(0,1)\cup(1,\infty)$ as follows:
\begin{align}
\widetilde{D}_{\alpha}(\rho\|\sigma) & \coloneqq\frac{1}{\alpha-1}\ln\Tr\!\left[\left(\sigma^{\frac{1-\alpha}{2\alpha}}\rho\sigma^{\frac{1-\alpha}{2\alpha}}\right)^{\alpha}\right],\\
\widehat{D}_{\alpha}(\rho\|\sigma) & \coloneqq\frac{1}{\alpha-1}\ln\Tr\!\left[\sigma\left(\sigma^{-\frac{1}{2}}\rho\sigma^{-\frac{1}{2}}\right)^{\alpha}\right].
\end{align}
For our case of interest here, the states being considered commute,
i.e., $\left[\rho,\sigma\right]=0$, in which case the following equalities
hold for $\alpha\in(0,1)\cup(1,\infty)$:
\begin{equation}
D_{\alpha}(\rho\|\sigma)=\widetilde{D}_{\alpha}(\rho\|\sigma)=\widehat{D}_{\alpha}(\rho\|\sigma).\label{eq:renyi-eq-commuting}
\end{equation}

Let $H$ be a Hamiltonian for a $d$-dimensional quantum system with
the following spectral decomposition:
\begin{equation}
H=\sum_{i=1}^{M}\lambda_{i}\Pi_{i},\label{eq:ham-spectral-decomp}
\end{equation}
where we assume that the eigenvalues are ordered, so that $\lambda_{M}\geq\cdots\geq\lambda_{1}$,
and $\left\{ \Pi_{i}\right\} _{i=1}^{M}$ is the set of spectral projections,
satisfying $\Pi_{i}\Pi_{j}=\delta_{i,j}\Pi_{i}$ for all $i,j\in\left[M\right]$
and $\sum_{i=1}^{M}\Pi_{i}=I$. 
\begin{prop}
\label{prop:trace-dist-thermal-ground}For a $d$-dimensional Hamiltonian
$H$ as in~\eqref{eq:ham-spectral-decomp} and $\beta\geq0$, the
following hold:
\begin{align}
& \frac{1}{2}\left\Vert \frac{e^{-\beta H}}{\Tr[e^{-\beta H}]}-\frac{\Pi_{1}}{\Tr[\Pi_{1}]}\right\Vert _{1} \notag \\
& =\frac{1}{1+\frac{d_{G}}{\sum_{i=2}^{M}e^{-\beta\left(\lambda_{i}-\lambda_{1}\right)}\Tr[\Pi_{i}]}}.\\
 & \leq\frac{1}{1+e^{\beta\Delta}\frac{d_{G}}{d-d_{G}}},\label{eq:error-bound-thermal-state-ground-state}
\end{align}
where $\Delta\coloneqq\lambda_{2}-\lambda_{1}$ is the spectral gap
and $d_{G}\coloneqq\Tr[\Pi_{1}]$ is the dimension of the ground space
of $H$.
\end{prop}

\begin{proof}
Consider that
\begin{equation}
\frac{e^{-\beta H}}{\Tr[e^{-\beta H}]}=\frac{\sum_{i=1}^{M}e^{-\beta\lambda_{i}}\Pi_{i}}{\sum_{i=1}^{M}e^{-\beta\lambda_{i}}\Tr[\Pi_{i}]}.
\end{equation}
Now consider that
\begin{align}
 & \left\Vert \frac{\sum_{i=1}^{M}e^{-\beta\lambda_{i}}\Pi_{i}}{\sum_{i=1}^{M}e^{-\beta\lambda_{i}}\Tr[\Pi_{i}]}-\frac{\Pi_{1}}{\Tr[\Pi_{1}]}\right\Vert _{1}\nonumber \\
 & =\left\Vert \frac{\sum_{i=2}^{M}e^{-\beta\lambda_{i}}\Pi_{i}}{\sum_{i=1}^{M}e^{-\beta\lambda_{i}}\Tr[\Pi_{i}]}+\frac{e^{-\beta\lambda_{1}}\Pi_{1}}{\sum_{i=1}^{M}e^{-\beta\lambda_{i}}\Tr[\Pi_{i}]}-\frac{\Pi_{1}}{\Tr[\Pi_{1}]}\right\Vert _{1}\\
 & =\left\Vert \begin{array}{cc}
\frac{\sum_{i=2}^{M}e^{-\beta\lambda_{i}}\Pi_{i}}{\sum_{i=1}^{M}e^{-\beta\lambda_{i}}\Tr[\Pi_{i}]}+\\
\left(\frac{e^{-\beta\lambda_{1}}}{\sum_{i=1}^{M}e^{-\beta\lambda_{i}}\Tr[\Pi_{i}]}-\frac{1}{\Tr[\Pi_{1}]}\right)\Pi_{1}
\end{array}\right\Vert _{1}\\
 & =\left\Vert \frac{\sum_{i=2}^{M}e^{-\beta\lambda_{i}}\Pi_{i}}{\sum_{i=1}^{M}e^{-\beta\lambda_{i}}\Tr[\Pi_{i}]}\right\Vert _{1}\nonumber \\
 & \qquad+\left|\frac{e^{-\beta\lambda_{1}}}{\sum_{i=1}^{M}e^{-\beta\lambda_{i}}\Tr[\Pi_{i}]}-\frac{1}{\Tr[\Pi_{1}]}\right|\left\Vert \Pi_{1}\right\Vert _{1}\\
 & =\frac{\sum_{i=2}^{M}e^{-\beta\lambda_{i}}\left\Vert \Pi_{i}\right\Vert _{1}}{\sum_{i=1}^{M}e^{-\beta\lambda_{i}}\Tr[\Pi_{i}]}\nonumber \\
 & \qquad+\left|\frac{e^{-\beta\lambda_{1}}}{\sum_{i=1}^{M}e^{-\beta\lambda_{i}}\Tr[\Pi_{i}]}-\frac{1}{\Tr[\Pi_{1}]}\right|\Tr[\Pi_{1}]\\
 & =\frac{\sum_{i=2}^{M}e^{-\beta\lambda_{i}}\Tr[\Pi_{i}]}{\sum_{i=1}^{M}e^{-\beta\lambda_{i}}\Tr[\Pi_{i}]}+\left|\frac{e^{-\beta\lambda_{1}}\Tr[\Pi_{1}]}{\sum_{i=1}^{M}e^{-\beta\lambda_{i}}\Tr[\Pi_{i}]}-1\right|\\
 & =\frac{\sum_{i=2}^{M}e^{-\beta\lambda_{i}}\Tr[\Pi_{i}]}{\sum_{i=2}^{M}e^{-\beta\lambda_{i}}\Tr[\Pi_{i}]+e^{-\beta\lambda_{1}}\Tr[\Pi_{1}]}\nonumber \\
 & \qquad+1-\frac{e^{-\beta\lambda_{1}}\Tr[\Pi_{1}]}{\sum_{i=1}^{M}e^{-\beta\lambda_{i}}\Tr[\Pi_{i}]}\\
 & =2\frac{\sum_{i=2}^{M}e^{-\beta\lambda_{i}}\Tr[\Pi_{i}]}{\sum_{i=2}^{M}e^{-\beta\lambda_{i}}\Tr[\Pi_{i}]+e^{-\beta\lambda_{1}}\Tr[\Pi_{1}]}\\
 & =2\frac{\sum_{i=2}^{M}e^{-\beta\left(\lambda_{i}-\lambda_{1}\right)}\Tr[\Pi_{i}]}{\sum_{i=2}^{M}e^{-\beta\left(\lambda_{i}-\lambda_{1}\right)}\Tr[\Pi_{i}]+\Tr[\Pi_{1}]}\\
 & =2\frac{1}{1+\frac{\Tr[\Pi_{1}]}{\sum_{i=2}^{M}e^{-\beta\left(\lambda_{i}-\lambda_{1}\right)}\Tr[\Pi_{i}]}}.
\end{align}

Now recall that $\Delta=\lambda_{2}-\lambda_{1}$. Then $e^{-\beta\left(\lambda_{i}-\lambda_{1}\right)}\leq e^{-\beta\Delta}$
for all $i\in\left\{ 2,\ldots,M\right\} $, which implies that
\begin{align}
 & \frac{1}{1+\frac{\Tr[\Pi_{1}]}{\sum_{i=2}^{M}e^{-\beta\left(\lambda_{i}-\lambda_{1}\right)}\Tr[\Pi_{i}]}}\nonumber \\
 & \leq\frac{1}{1+\frac{\Tr[\Pi_{1}]}{\sum_{i=2}^{M}e^{-\beta\Delta}\Tr[\Pi_{i}]}}\\
 & =\frac{1}{1+e^{\beta\Delta}\frac{\Tr[\Pi_{1}]}{\sum_{i=2}^{M}\Tr[\Pi_{i}]}}\\
 & =\frac{1}{1+e^{\beta\Delta}\frac{\Tr[\Pi_{1}]}{\Tr\left[\sum_{i=2}^{M}\Pi_{i}\right]}}\\
 & =\frac{1}{1+e^{\beta\Delta}\frac{\Tr[\Pi_{1}]}{\Tr\left[I-\Pi_{1}\right]}}\\
 & =\frac{1}{1+e^{\beta\Delta}\frac{\Tr[\Pi_{1}]}{d-\Tr\left[\Pi_{1}\right]}},
\end{align}
thus concluding the proof.
\end{proof}
\begin{cor}
For the normalized trace distance to be $\leq\varepsilon\in[0,1]$
in~\eqref{eq:error-bound-thermal-state-ground-state}, i.e., to have
\begin{equation}
\frac{1}{2}\left\Vert \frac{e^{-\beta H}}{\Tr[e^{-\beta H}]}-\frac{\Pi_{1}}{\Tr[\Pi_{1}]}\right\Vert _{1}\leq\varepsilon,
\end{equation}
we require the inverse temperature $\beta$ to be set as follows:
\begin{equation}
\beta=\frac{1}{\Delta}\ln\!\left[\left(\frac{1-\varepsilon}{\varepsilon}\right)\left(\frac{d-d_{G}}{d_{G}}\right)\right],
\end{equation}
\end{cor}

\begin{proof}
Consider that we need to solve for $\varepsilon$ in~\eqref{eq:1st-eq-solve-eps}
and find that
\begin{align}
 \frac{1}{1+e^{\beta\Delta}\frac{\Tr[\Pi_{1}]}{d-\Tr\left[\Pi_{1}\right]}} & = \varepsilon\label{eq:1st-eq-solve-eps}\\
\iff\quad1+e^{\beta\Delta}\frac{\Tr[\Pi_{1}]}{d-\Tr\left[\Pi_{1}\right]} & =\frac{1}{\varepsilon}\\
\iff\quad e^{\beta\Delta}\frac{\Tr[\Pi_{1}]}{d-\Tr\left[\Pi_{1}\right]} & =\frac{1}{\varepsilon}-1\\
\iff\quad \left(\frac{1-\varepsilon}{\varepsilon}\right)\left(\frac{d-\Tr\left[\Pi_{1}\right]}{\Tr[\Pi_{1}]}\right)  & = e^{\beta\Delta}
\end{align}
\begin{equation}
\iff\quad\frac{1}{\Delta}\ln\!\left[\left(\frac{1-\varepsilon}{\varepsilon}\right)\left(\frac{d-\Tr\left[\Pi_{1}\right]}{\Tr[\Pi_{1}]}\right)\right]=\beta,
\end{equation}
concluding the proof.
\end{proof}
\begin{prop}
\label{prop:fid-thermal-ground}For a $d$-dimensional Hamiltonian
$H$ as in~\eqref{eq:ham-spectral-decomp} and $\beta\geq0$, the
following hold:
\begin{align}
& F\!\left(\frac{e^{-\beta H}}{\Tr[e^{-\beta H}]},\frac{\Pi_{1}}{\Tr[\Pi_{1}]}\right) \notag \\
& =\frac{1}{1+\frac{\sum_{i=2}^{M}e^{-\beta(\lambda_{i}-\lambda_{1})}\Tr[\Pi_{i}]}{d_{G}}}\\
 & \geq\frac{1}{1+e^{-\beta\Delta}\left(\frac{d-d_{G}}{d_{G}}\right)},\label{eq:fidelity-to-ground-state-bnd}
\end{align}
where we have used the same notation as in Proposition~\ref{prop:trace-dist-thermal-ground}.
\end{prop}

\begin{proof}
Consider that
\begin{align}
 & F\!\left(\frac{e^{-\beta H}}{\Tr[e^{-\beta H}]},\frac{\Pi_{1}}{\Tr[\Pi_{1}]}\right)\nonumber \\
 & =\left(\Tr\!\left[\sqrt{\sqrt{\frac{\Pi_{1}}{\Tr[\Pi_{1}]}}\frac{e^{-\beta H}}{\Tr[e^{-\beta H}]}\sqrt{\frac{\Pi_{1}}{\Tr[\Pi_{1}]}}}\right]\right)^{2}\\
 & =\frac{\left(\Tr\!\left[\sqrt{\Pi_{1}e^{-\beta H}\Pi_{1}}\right]\right)^{2}}{\Tr[\Pi_{1}]\Tr[e^{-\beta H}]}\\
 & =\frac{\left(\Tr\!\left[\sqrt{\Pi_{1}\left(\sum_{i=1}^{M}e^{-\beta\lambda_{i}}\Pi_{i}\right)\Pi_{1}}\right]\right)^{2}}{\Tr[\Pi_{1}]\sum_{i=1}^{M}e^{-\beta\lambda_{i}}\Tr[\Pi_{i}]}\\
 & =\frac{\left(\Tr\!\left[\sqrt{e^{-\beta\lambda_{1}}\Pi_{1}}\right]\right)^{2}}{\Tr[\Pi_{1}]\sum_{i=1}^{M}e^{-\beta\lambda_{i}}\Tr[\Pi_{i}]}\\
 & =\frac{e^{-\beta\lambda_{1}}\Tr[\Pi_{1}]}{\sum_{i=1}^{M}e^{-\beta\lambda_{i}}\Tr[\Pi_{i}]}\\
 & =\frac{e^{-\beta\lambda_{1}}\Tr[\Pi_{1}]}{e^{-\beta\lambda_{1}}\Tr[\Pi_{1}]+\sum_{i=2}^{M}e^{-\beta\lambda_{i}}\Tr[\Pi_{i}]}\\
 & =\frac{1}{1+\frac{1}{d_{G}}\sum_{i=2}^{M}e^{-\beta(\lambda_{i}-\lambda_{1})}\Tr[\Pi_{i}]}\\
 & \geq\frac{\Tr[\Pi_{1}]}{1+\frac{1}{d_{G}}\sum_{i=2}^{M}e^{-\beta\Delta}\Tr[\Pi_{i}]}\\
 & =\frac{1}{1+\frac{e^{-\beta\Delta}}{d_{G}}\sum_{i=2}^{M}\Tr[\Pi_{i}]}\\
 & =\frac{1}{1+e^{-\beta\Delta}\left(\frac{d-d_{G}}{d_{G}}\right)}.
\end{align}
This completes the proof.
\end{proof}
\begin{cor}
For the fidelity to be $\geq1-\varepsilon$ in~\eqref{eq:fidelity-to-ground-state-bnd}
for $\varepsilon\in[0,1]$, i.e., to have
\begin{equation}
F\!\left(\frac{e^{-\beta H}}{\Tr[e^{-\beta H}]},\frac{\Pi_{1}}{\Tr[\Pi_{1}]}\right)\geq1-\varepsilon,
\end{equation}
we require the inverse temperature $\beta$ to be set as follows:
\begin{equation}
\beta=\frac{1}{\Delta}\ln\!\left[\left(\frac{1-\varepsilon}{\varepsilon}\right)\left(\frac{d-d_{G}}{d_{G}}\right)\right].
\end{equation}
\end{cor}

\begin{proof}
Considering~\eqref{eq:fidelity-to-ground-state-bnd}, we should solve
for $\beta$ therein and find that
\begin{align}
  \frac{1}{1+e^{-\beta\Delta}\left(\frac{d-d_{G}}{d_{G}}\right)} & = 1-\varepsilon\\
\iff\quad1+e^{-\beta\Delta}\left(\frac{d-d_{G}}{d_{G}}\right) & =\frac{1}{1-\varepsilon}\\
\iff\quad e^{-\beta\Delta}\left(\frac{d-d_{G}}{d_{G}}\right) & =\frac{1}{1-\varepsilon}-1\\
\iff\quad e^{-\beta\Delta} & =\left(\frac{\varepsilon}{1-\varepsilon}\right)\left(\frac{d_{G}}{d-d_{G}}\right)
\end{align}
\begin{equation}
\iff\quad\beta=\frac{1}{\Delta}\ln\!\left[\left(\frac{1-\varepsilon}{\varepsilon}\right)\left(\frac{d-d_{G}}{d_{G}}\right)\right],
\end{equation}
thus concluding the proof.
\end{proof}
\begin{prop}
\label{prop:rel-ent-thermal-ground}For a $d$-dimensional Hamiltonian
$H$ as in~\eqref{eq:ham-spectral-decomp} and $\beta\geq0$, the
following hold:
\begin{align}
& D\!\left(\frac{\Pi_{1}}{\Tr[\Pi_{1}]}\middle\|\frac{e^{-\beta H}}{\Tr[e^{-\beta H}]}\right) \notag \\
& =\ln\!\left(1+\frac{\sum_{i=2}^{M}e^{-\beta\left(\lambda_{i}-\lambda_{1}\right)}\Tr[\Pi_{i}]}{d_{G}}\right)\label{eq:rel-ent-to-gnd-state-eq}\\
 & \leq\ln\!\left(1+e^{-\beta\Delta}\left(\frac{d-d_{G}}{d_{G}}\right)\right),\label{eq:rel-ent-to-ground-state-bnd}
\end{align}
where we have used the same notation as in Proposition~\ref{prop:trace-dist-thermal-ground}.
Furthermore, the following equalities hold for all $\alpha\in(0,1)\cup(1,\infty)$:
\begin{align}
& D\!\left(\frac{\Pi_{1}}{\Tr[\Pi_{1}]}\middle\|\frac{e^{-\beta H}}{\Tr[e^{-\beta H}]}\right) \notag \\
& =D_{\alpha}\!\left(\frac{\Pi_{1}}{\Tr[\Pi_{1}]}\middle\|\frac{e^{-\beta H}}{\Tr[e^{-\beta H}]}\right)\label{eq:Petz-Renyi-equality}\\
 & =\widetilde{D}_{\alpha}\!\left(\frac{\Pi_{1}}{\Tr[\Pi_{1}]}\middle\|\frac{e^{-\beta H}}{\Tr[e^{-\beta H}]}\right)\\
 & =\widehat{D}_{\alpha}\!\left(\frac{\Pi_{1}}{\Tr[\Pi_{1}]}\middle\|\frac{e^{-\beta H}}{\Tr[e^{-\beta H}]}\right).\label{eq:geometric-Renyi-equality}
\end{align}
\end{prop}

\begin{proof}
Consider that
\begin{align}
 & D\!\left(\frac{\Pi_{1}}{\Tr[\Pi_{1}]}\middle\|\frac{e^{-\beta H}}{\Tr[e^{-\beta H}]}\right)\nonumber \\
 & =\Tr\!\left[\frac{\Pi_{1}}{\Tr[\Pi_{1}]}\ln\!\left(\frac{\Pi_{1}}{\Tr[\Pi_{1}]}\right)\right]\nonumber \\
 & \qquad-\Tr\!\left[\frac{\Pi_{1}}{\Tr[\Pi_{1}]}\ln\!\left(\frac{e^{-\beta H}}{\Tr[e^{-\beta H}]}\right)\right]\\
 & =-\ln d_{G}-\Tr\!\left[\frac{\Pi_{1}}{\Tr[\Pi_{1}]}\ln\!\left(e^{-\beta H}\right)\right]+\ln\Tr[e^{-\beta H}]\\
 & =-\ln d_{G}+\frac{\beta}{\Tr[\Pi_{1}]}\Tr\!\left[\Pi_{1}H\right]+\ln\Tr[e^{-\beta H}]\\
 & =-\ln d_{G}+\frac{\lambda_{1}\beta}{\Tr[\Pi_{1}]}\Tr\!\left[\Pi_{1}\right]+\ln\Tr[e^{-\beta H}]\\
 & =\lambda_{1}\beta-\ln d_{G}+\ln\Tr[e^{-\beta H}]\label{eq:rel-ent-mid-expression}\\
 & =\lambda_{1}\beta-\ln d_{G}+\ln\sum_{i=1}^{M}e^{-\beta\lambda_{i}}\Tr[\Pi_{i}]\\
 & =\ln\!\left(\frac{1}{d_{G}}\sum_{i=1}^{M}e^{-\beta\left(\lambda_{i}-\lambda_{1}\right)}\Tr[\Pi_{i}]\right)\\
 & =\ln\!\left(\frac{\Tr[\Pi_{1}]+\sum_{i=2}^{M}e^{-\beta\left(\lambda_{i}-\lambda_{1}\right)}\Tr[\Pi_{i}]}{d_{G}}\right)\\
 & =\ln\!\left(1+\frac{1}{d_{G}}\sum_{i=2}^{M}e^{-\beta\left(\lambda_{i}-\lambda_{1}\right)}\Tr[\Pi_{i}]\right)\\
 & \leq\ln\!\left(1+\frac{1}{d_{G}}\sum_{i=2}^{M}e^{-\beta\Delta}\Tr[\Pi_{i}]\right)\\
 & =\ln\!\left(1+e^{-\beta\Delta}\left(\frac{d-d_{G}}{d_{G}}\right)\right).
\end{align}
This establishes~\eqref{eq:rel-ent-to-gnd-state-eq} and~\eqref{eq:rel-ent-to-ground-state-bnd}.

We now prove~\eqref{eq:Petz-Renyi-equality}--\eqref{eq:geometric-Renyi-equality}.
Since the states of interest commute, the equalities in~\eqref{eq:renyi-eq-commuting}
follow, and it suffices to evaluate the Petz--R\'enyi relative entropy.
To this end, consider that
\begin{align}
 & D_{\alpha}\!\left(\frac{\Pi_{1}}{\Tr[\Pi_{1}]}\middle\|\frac{e^{-\beta H}}{\Tr[e^{-\beta H}]}\right)\nonumber \\
 & =\frac{1}{\alpha-1}\ln\Tr\!\left[\left(\frac{\Pi_{1}}{\Tr[\Pi_{1}]}\right)^{\alpha}\left(\frac{e^{-\beta H}}{\Tr[e^{-\beta H}]}\right)^{1-\alpha}\right]\\
 & =\frac{1}{\alpha-1}\ln\frac{1}{d_{G}^{\alpha}\left(\Tr[e^{-\beta H}]\right)^{1-\alpha}}\nonumber \\
 & \qquad+\frac{1}{\alpha-1}\ln\Tr\!\left[\Pi_{1}^{\alpha}e^{-\beta\left(1-\alpha\right)H}\right]\\
 & =-\frac{\alpha}{\alpha-1}\ln d_{G}+\ln\Tr[e^{-\beta H}]\nonumber \\
 & \qquad+\frac{1}{\alpha-1}\ln\Tr\!\left[e^{-\beta\left(1-\alpha\right)\lambda_{1}}\Pi_{1}\right]\\
 & =-\frac{\alpha}{\alpha-1}\ln d_{G}+\ln\Tr[e^{-\beta H}]\nonumber \\
 & \qquad+\beta\lambda_{1}+\frac{1}{\alpha-1}\ln d_{G}\\
 & =-\ln d_{G}+\ln\Tr[e^{-\beta H}]+\beta\lambda_{1}.
\end{align}
This last expression coincides with that in~\eqref{eq:rel-ent-mid-expression},
thus concluding the proof.
\end{proof}
\begin{cor}
Fix $\varepsilon>0$. For the relative entropy to be $\leq\varepsilon$
in~\eqref{eq:rel-ent-to-ground-state-bnd}, i.e., to have
\begin{equation}
D\!\left(\frac{\Pi_{1}}{\Tr[\Pi_{1}]}\middle\|\frac{e^{-\beta H}}{\Tr[e^{-\beta H}]}\right)\leq\varepsilon,
\end{equation}
we require the inverse temperature $\beta$ to be set as follows:
\begin{equation}
\beta=\frac{1}{\Delta}\ln\!\left[\left(\frac{1}{e^{\varepsilon}-1}\right)\left(\frac{d-d_{G}}{d_{G}}\right)\right].
\end{equation}
Due to~\eqref{eq:Petz-Renyi-equality}--\eqref{eq:geometric-Renyi-equality},
the same statement applies to the Petz, sandwiched, and geometric
R\'enyi relative entropies for all $\alpha\in(0,1)\cup(1,\infty)$.
\end{cor}

\begin{proof}
We have to solve for the inverse temperature $\beta$ in the following
equation:
\begin{align}
 \ln\!\left(1+e^{-\beta\Delta}\left(\frac{d-d_{G}}{d_{G}}\right)\right) & = \varepsilon\\
\iff\qquad e^{\varepsilon}-1 & =e^{-\beta\Delta}\left(\frac{d-d_{G}}{d_{G}}\right)\\
\iff\qquad\frac{d_{G}}{d-d_{G}}\left(e^{\varepsilon}-1\right) & =e^{-\beta\Delta}
\end{align}
\begin{equation}
\iff\qquad\beta=\frac{1}{\Delta}\ln\!\left[\left(\frac{1}{e^{\varepsilon}-1}\right)\left(\frac{d-d_{G}}{d_{G}}\right)\right],
\end{equation}
thus concluding the proof.
\end{proof}
The following corollary establishes equalities relating several distinguishability
measures considered in this appendix, whenever the two states being
compared are the temperature-$T$ thermal state of a Hamiltonian and
the maximally mixed state on its ground space.
\begin{cor}
For a $d$-dimensional Hamiltonian $H$ as in~\eqref{eq:ham-spectral-decomp}
and $\beta\geq0$, the following equalities hold:
\begin{multline}
\frac{1}{2}\left\Vert \frac{e^{-\beta H}}{\Tr[e^{-\beta H}]}-\frac{\Pi_{1}}{\Tr[\Pi_{1}]}\right\Vert _{1}=\\
1-F\!\left(\frac{e^{-\beta H}}{\Tr[e^{-\beta H}]},\frac{\Pi_{1}}{\Tr[\Pi_{1}]}\right),
\end{multline}
\begin{multline}
D\!\left(\frac{\Pi_{1}}{\Tr[\Pi_{1}]}\middle\|\frac{e^{-\beta H}}{\Tr[e^{-\beta H}]}\right)=\\
-\ln\!\left(F\!\left(\frac{e^{-\beta H}}{\Tr[e^{-\beta H}]},\frac{\Pi_{1}}{\Tr[\Pi_{1}]}\right)\right)
\end{multline}
\end{cor}

\begin{proof}
Set
\begin{align}
a & \equiv\frac{d_{G}}{\sum_{i=2}^{M}e^{-\beta\left(\lambda_{i}-\lambda_{1}\right)}\Tr[\Pi_{i}]},\\
T & \equiv\frac{1}{2}\left\Vert \frac{e^{-\beta H}}{\Tr[e^{-\beta H}]}-\frac{\Pi_{1}}{\Tr[\Pi_{1}]}\right\Vert _{1},\\
F & \equiv F\!\left(\frac{e^{-\beta H}}{\Tr[e^{-\beta H}]},\frac{\Pi_{1}}{\Tr[\Pi_{1}]}\right),\\
D & \equiv D\!\left(\frac{\Pi_{1}}{\Tr[\Pi_{1}]}\middle\|\frac{e^{-\beta H}}{\Tr[e^{-\beta H}]}\right).
\end{align}
By Propositions~\ref{prop:trace-dist-thermal-ground},~\ref{prop:fid-thermal-ground},
and~\ref{prop:rel-ent-thermal-ground}, we have that
\begin{align}
T & =\frac{1}{1+a},\\
F & =\frac{1}{1+\frac{1}{a}},\\
D & =\ln\!\left(1+\frac{1}{a}\right).
\end{align}
Clearly, $D=-\ln F$. Additionally,
\begin{align}
1-F & =1-\frac{1}{1+\frac{1}{a}}=\frac{1+\frac{1}{a}}{1+\frac{1}{a}}-\frac{1}{1+\frac{1}{a}}\\
 & =\frac{\frac{1}{a}}{1+\frac{1}{a}}=\frac{1}{1+a}\\
 & =T,
\end{align}
thus concluding the proof.
\end{proof}

\section{Simplifying the implementation of the channel $\Phi_\mu$ for extensive, conserved charges}

\label{app:simplifying-phi-mu}

In this appendix, we provide a proof of \eqref{eq:simplify-phi-conserved-extensive-charges}.
Starting from~\eqref{eq:simplify-phi-conserved-2}, Eq.~\eqref{eq:simplify-phi-conserved-extensive-charges} follows because
\begin{align}
    & \Phi_\mu(Q_i)  \notag \\
    & = \int_{-\infty}^{\infty}dt\ p(t)\,e^{i\mu\cdot Qt/T}Q_i e^{-i\mu\cdot Qt/T} \\
    & = \int_{-\infty}^{\infty}dt\ p(t)\,e^{i\mu\cdot Qt/T}\left(\sum_{j=1}^{n}C_{i}^{(j)}\right) e^{-i\mu\cdot Qt/T} \label{eq:extensive-simplify-1}\\
    & = \sum_{j=1}^{n}\int_{-\infty}^{\infty}dt\ p(t)\,e^{i\mu\cdot Qt/T}C_{i}^{(j)} e^{-i\mu\cdot Qt/T}\\
    & = \sum_{j=1}^{n}\int_{-\infty}^{\infty}dt\ p(t)\,e^{\frac{it}{T} \sum_{i'\in [c]} \mu_{i'}Q_{i'} }C_{i}^{(j)} e^{-\frac{it}{T} \sum_{i''\in [c]} \mu_{i''} Q_{i''} }\\ 
    & = \sum_{j=1}^{n}\int_{-\infty}^{\infty}dt\ p(t)\,e^{\frac{it}{T} \sum_{i'\in [c]} \mu_{i'} \sum_{j'=1}^{n} C_{i'}^{(j')}}C_{i}^{(j)} \times \notag \\
    & \qquad\qquad\qquad\qquad e^{-\frac{it}{T} \sum_{i''\in [c]} \mu_{i''} \sum_{j''=1}^{n} C_{i''}^{(j'')} }\label{eq:extensive-simplify-2}\\ 
    & = \sum_{j=1}^{n}\int_{-\infty}^{\infty}dt\ p(t)\,\prod_{j'=1}^{n} e^{\frac{it}{T} \sum_{i'\in [c]} \mu_{i'} C_{i'}^{(j')}}C_{i}^{(j)} \times \notag \\
    & \qquad\qquad\qquad\qquad \prod_{j''=1}^{n}e^{-\frac{it}{T} \sum_{i''\in [c]} \mu_{i''}  C_{i''}^{(j'')} }\label{eq:extensive-simplify-3}\\ 
        & = \sum_{j=1}^{n}\int_{-\infty}^{\infty}dt\ p(t)\, e^{\frac{it}{T} \sum_{i'\in [c]} \mu_{i'} C_{i'}^{(j)}} C_{i}^{(j)} \times \notag \\
    & \qquad\qquad\qquad\qquad e^{-\frac{it}{T} \sum_{i''\in [c]} \mu_{i''}  C_{i''}^{(j)} }\label{eq:extensive-simplify-4}\\ 
        & = \sum_{j=1}^{n}\int_{-\infty}^{\infty}dt\ p(t)\, e^{\frac{it}{T}  \mu\cdot  C^{(j)}} C_{i}^{(j)}  e^{-\frac{it}{T}  \mu \cdot  C^{(j)} }.
\end{align}
Eqs.~\eqref{eq:extensive-simplify-1} and~\eqref{eq:extensive-simplify-2} follow by substituting \eqref{eq:def-extensive-charges}. Eqs.~\eqref{eq:extensive-simplify-3} and~\eqref{eq:extensive-simplify-4} follow by observing that $\left[C_{i'}^{(j')},C_{i''}^{(j'')}\right]=0$ for all $j'\neq j''$ and applying \eqref{eq:exp-sum-commute}.

\section{Derivation of Equation~\eqref{eq:encoded-pauli-expansion}}

\label{sec:Derivation-of-Equation}In this appendix, we include a
derivation of Eq.~\eqref{eq:encoded-pauli-expansion} for completeness.

First, recall that a Clifford unitary is one that preserves the $n$-fold
Pauli group $\mathcal{P}_{n}$ when acting on its elements by conjugation
\cite{Gottesman1998} (see also \cite[Section~3.5]{Gottesman2009});
i.e., the set $C_{n}$ of all Clifford unitaries acting on $n$ qubits
is defined as follows: 
\begin{equation}
C_{n}\coloneqq\left\{ U\in\mathcal{U}(2^{n}):UPU^{\dag}\in\mathcal{P}_{n}\,\forall P\in\mathcal{P}_{n}\right\} ,
\end{equation}
where $\mathcal{U}(2^{n})$ denotes the set of all $n$-qubit unitaries.

Consider a stabilizer code that encodes $k\in\mathbb{N}$ logical
qubits into $n\in\mathbb{N}$ physical qubits, where $k\leq n$. Such
a code is encoded by a Clifford encoding unitary $U$. The unencoded
state before the encoded unitary acts is as follows:
\begin{equation}
\rho\otimes|0\rangle\!\langle0|^{\otimes n-k},
\end{equation}
where $\rho$ is defined as in~\eqref{eq:pauli-expansion-k-qubits},
and the state after the encoding is as follows:
\begin{equation}
U\left(\rho\otimes|0\rangle\!\langle0|^{\otimes n-k}\right)U^{\dag}.
\end{equation}
Recalling that $|0\rangle\!\langle0|=\frac{I+Z}{2}$, we can write
the state before encoding takes place as follows:
\begin{multline}
\rho\otimes\left(\frac{I+Z}{2}\right)^{\otimes n-k}\\
=\left(I^{\otimes k}\otimes\left(\frac{I+Z}{2}\right)^{\otimes n-k}\right)\left(\rho\otimes I^{\otimes n-k}\right)\times\\
\left(I^{\otimes k}\otimes\left(\frac{I+Z}{2}\right)^{\otimes n-k}\right),
\end{multline}
so that the state after encoding takes place is
\begin{multline}
U\left[\rho\otimes\left(\frac{I+Z}{2}\right)^{\otimes n-k}\right]U^{\dag}\\
=\left[U\left(I^{\otimes k}\otimes\left(\frac{I+Z}{2}\right)^{\otimes n-k}\right)U^{\dag}\right]\times \\
\left[U\left(\rho\otimes I^{\otimes n-k}\right)U^{\dag}\right]\times\\
\left[U\left(I^{\otimes k}\otimes\left(\frac{I+Z}{2}\right)^{\otimes n-k}\right)U^{\dag}\right].\label{eq:clifford-encoded-state-expand}
\end{multline}
Now recall that the action of the Clifford unitary is to transform
the $n-k$ operators $Z_{k+1},\ldots,Z_{n}$ to $n-k$ stabilizer
generators $\overline{Z}_{k+1},\ldots,\overline{Z}_{n}$ as follows:
\begin{equation}
UZ_{i}U^{\dag}=\overline{Z}_{i}\qquad\forall i\in\left\{ k+1,\ldots,n\right\} ,
\end{equation}
and it transforms the $3k$ logical operators $X_{1}$, $Y_{1}$,
$Z_{1}$, \ldots, $X_{k}$, $Y_{k}$, $Z_{k}$ for the unencoded
state to the $3k$ logical operators $\overline{X}_{1}$, $\overline{Y}_{1}$,
$\overline{Z}_{1}$, \ldots, $\overline{X}_{k}$, $\overline{Y}_{k}$,
$\overline{Z}_{k}$ for the encoded state, as follows:
\begin{align}
UX_{i}U^{\dag} & =\overline{X}_{i}\qquad\forall i\in\left[k\right],\\
UY_{i}U^{\dag} & =\overline{Y}_{i}\qquad\forall i\in\left[k\right],\\
UZ_{i}U^{\dag} & =\overline{Z}_{i}\qquad\forall i\in\left[k\right].
\end{align}
As such, it follows that
\begin{align}
 & U\left(I^{\otimes k}\otimes\left(\frac{I+Z}{2}\right)^{\otimes n-k}\right)U^{\dag}\nonumber \\
 & =U\left(\prod_{i=k+1}^{n}\frac{I+Z_{i}}{2}\right)U^{\dag}\label{eq:clifford-expand-stabilizer-ops}\\
 & =\prod_{i=k+1}^{n}U\left(\frac{I+Z_{i}}{2}\right)U^{\dag}\\
 & =\prod_{i=k+1}^{n}\left(\frac{I+\overline{Z}_{i}}{2}\right)\\
 & =\Pi_{\mathcal{C}},
\end{align}
and
\begin{multline}
U\left(\rho\otimes I^{\otimes n-k}\right)U^{\dag}=\\
\frac{1}{2^{k}}\sum_{i_{1},\ldots,i_{k}\in\{0,1,2,3\}}r_{i_{1},\ldots,i_{k}}\overline{\sigma}_{i_{1},1}\cdots\overline{\sigma}_{i_{k},k},\label{eq:clifford-expand-log-ops}
\end{multline}
where we have applied the definitions in~\eqref{eq:pauli-expansion-k-qubits}
and~\eqref{eq:encoded-pauli-expansion}. Thus, combining~\eqref{eq:clifford-encoded-state-expand}
and~\eqref{eq:clifford-expand-stabilizer-ops}--\eqref{eq:clifford-expand-log-ops},
we conclude~\eqref{eq:encoded-pauli-expansion}.

\section{Derivation of Equation~\eqref{eq:complementary-slack-cond}}

\label{sec:Derivation-of-Equation-1}Let us redo the development in
\eqref{eq:proof-lagrange-mults}--\eqref{eq:duality-proof-last}
in a slightly different way, in order to arrive at the complementary
slackness condition in~\eqref{eq:complementary-slack-cond}. Let us
set $Q_{0}=I$ and $q_{0}=1$, in order to encode the constraint that
$\rho$ is a density matrix, i.e., $\Tr[\rho]=1$. Consider that
\begin{align}
 & \min_{\rho\in\mathcal{D}_{d^{n}}}\left\{ \Tr[H\rho]:\Tr[Q_{i}\rho]=q_{i}\ \forall i\in\left[c\right]\right\} \nonumber \\
 & =\min_{\rho\geq0}\left\{ \Tr[H\rho]:\Tr[Q_{i}\rho]=q_{i}\ \forall i\in\left\{ 0,1,\ldots,c\right\} \right\} \label{eq:add-trace-density-constraint}\\
 & =\min_{\rho\geq0}\left\{ \Tr[H\rho]+\sup_{\substack{\mu\in\mathbb{R}^{c},\\
\mu_{0}\in\mathbb{R}
}
}\sum_{i\in\left\{ 0,1,\ldots,c\right\} }\mu_{i}\left(q_{i}-\Tr[Q_{i}\rho]\right)\right\} \label{eq:proof-lagrange-mults-1}\\
 & =\min_{\rho\geq0}\sup_{\substack{\mu\in\mathbb{R}^{c},\\
\mu_{0}\in\mathbb{R}
}
}\left\{ \Tr[H\rho]+\sum_{i\in\left\{ 0,1,\ldots,c\right\} }\mu_{i}\left(q_{i}-\Tr[Q_{i}\rho]\right)\right\} \label{eq:sup-to-outside-1}\\
 & =\min_{\rho\geq0}\sup_{\substack{\mu\in\mathbb{R}^{c},\\
\mu_{0}\in\mathbb{R}
}
}\left\{ \mu\cdot q+\mu_{0}+\Tr[\left(H-\mu\cdot Q-\mu_{0}I\right)\rho]\right\} \label{eq:algebraic-simplifications-1}\\
 & =\sup_{\substack{\mu\in\mathbb{R}^{c},\\
\mu_{0}\in\mathbb{R}
}
}\min_{\rho\geq0}\left\{ \mu\cdot q+\mu_{0}+\Tr[\left(H-\mu\cdot Q-\mu_{0}I\right)\rho]\right\} \label{eq:minimax-app}\\
 & =\sup_{\substack{\mu\in\mathbb{R}^{c},\\
\mu_{0}\in\mathbb{R}
}
}\left\{ \mu\cdot q+\mu_{0}:H-\mu\cdot Q\geq\mu_{0}I\right\} .\label{eq:duality-proof-last-1}
\end{align}
The equalities follow for reasons similar to those given for~\eqref{eq:proof-lagrange-mults}--\eqref{eq:duality-proof-last},
with the exception that~\eqref{eq:minimax-app} follows from strong
duality of semidefinite programming, under the assumption that there
exists at least one state $\rho$ satisfying the constraints in~\eqref{eq:constraints-for-ground-space}.
Observe that an optimal $\mu_{0}$ in~\eqref{eq:duality-proof-last-1}
is equal to the minimum eigenvalue of $H-\mu\cdot Q$. Now suppose
that $\rho^{*}$ is optimal for~\eqref{eq:add-trace-density-constraint}
and $\mu^{*}$ is optimal for~\eqref{eq:duality-proof-last-1}. Then
it follows that
\begin{align}
\Tr[H\rho^{*}] & =\mu^{*}\cdot q+\mu_{0}\\
 & =\mu^{*}\cdot q+\lambda_{\min}(H-\mu^{*}\cdot Q).
\end{align}
Consider that $\Tr[Q_{i}\rho^{*}]=q_{i}$ for all $i\in\left[c\right]$.
Then
\begin{equation}
\mu^{*}\cdot q=\Tr[\left(\mu^{*}\cdot Q\right)\rho^{*}].
\end{equation}
Plugging back in above, we conclude the desired equality:
\begin{align}
\Tr[H\rho^{*}] & =\Tr[\left(\mu^{*}\cdot Q\right)\rho]+\lambda_{\min}(H-\mu^{*}\cdot Q)\\
\iff\qquad & \Tr[\left(H-\mu^{*}\cdot Q\right)\rho^{*}]=\lambda_{\min}(H-\mu^{*}\cdot Q),
\end{align}
so that $\rho^{*}$ is in the ground space of $H-\mu^{*}\cdot Q$.

\section{Proof of Theorem~\ref{thm:mixture-to-exp-qubit}}

\label{app:mixture-to-exp-proof}
\begin{proof}
Consider that, in order for the equality in~\eqref{eq:desired-eq-mix-exp}
to hold, the eigenvectors of the matrix on the left-hand side and
those for the matrix on the right-hand side should match. The eigenvectors
for the matrix on the left-hand side are the same as those for the
following matrix:
\begin{equation}
r_{x}X+r_{y}Y+r_{z}Z,\label{eq:r-matrix}
\end{equation}
and the eigenvectors for the matrix on the right-hand side are the
same as those for the following matrix:
\begin{equation}
\mu_{x}X+\mu_{y}Y+\mu_{z}Z.
\end{equation}
Thus, for the eigenvectors to match, we simply set $\mu=r$. This
is also intuitive, because the vector $r$ identifies a direction
in the Bloch sphere, and so does $\mu$, and these directions should
be aligned. So now we are trying to solve the following equation for
$\beta$, which is the only parameter left to solve for:
\begin{multline}
\frac{1}{2}\left(I+r_{x}X+r_{y}Y+r_{z}Z\right)=\\
\frac{1}{Z_{\beta}(r)}\exp\!\left(-\beta\left(r_{x}X+r_{y}Y+r_{z}Z\right)\right).\label{eq:equa-to-solve}
\end{multline}
It is well known that the eigenvalues of the matrix in~\eqref{eq:r-matrix}
are given by
\begin{equation}
\pm\left\Vert r\right\Vert .\label{eq:pm-norm-bloch-vec}
\end{equation}
To see this, recall that the eigenvalues of a $2\times2$ matrix $A$
are given by the following formula:
\begin{equation}
\lambda_{\pm}=\frac{1}{2}\left(\Tr[A]\pm\sqrt{\left(\Tr[A]\right)^{2}-4\det(A)}\right).
\end{equation}
Since the trace of the matrix in~\eqref{eq:r-matrix} is equal to
zero, this implies that the eigenvalues are given by the following
formula:
\begin{equation}
\pm\sqrt{-\det(r_{x}X+r_{y}Y+r_{z}Z)},
\end{equation}
which one finds upon direct calculation to be equal to~\eqref{eq:pm-norm-bloch-vec}.
This implies that the eigenvalues of the matrix in~\eqref{eq:bloch-rep}
are given by
\begin{equation}
\frac{1\pm\left\Vert r\right\Vert }{2},\label{eq:eigenvals-eye-plus-pauli}
\end{equation}
and the eigenvalues of $\exp\!\left(-\beta\left(r_{x}X+r_{y}Y+r_{z}Z\right)\right)$
are given by
\begin{equation}
e^{\mp\beta\left\Vert r\right\Vert }.
\end{equation}
Then
\begin{equation}
Z_{\beta}(r)=e^{-\beta\left\Vert r\right\Vert }+e^{\beta\left\Vert r\right\Vert },
\end{equation}
and we conclude that the eigenvalues of the matrix on the right-hand
side of~\eqref{eq:equa-to-solve} are given by
\begin{equation}
\frac{e^{\mp\beta\left\Vert r\right\Vert }}{e^{-\beta\left\Vert r\right\Vert }+e^{\beta\left\Vert r\right\Vert }}.\label{eq:thermal-form-eigenvals}
\end{equation}
Now setting the first eigenvalue in~\eqref{eq:eigenvals-eye-plus-pauli}
to the first eigenvalue in~\eqref{eq:thermal-form-eigenvals}, we
conclude that
\begin{equation}
\frac{1+\left\Vert r\right\Vert }{2}=\frac{e^{-\beta\left\Vert r\right\Vert }}{e^{-\beta\left\Vert r\right\Vert }+e^{\beta\left\Vert r\right\Vert }}.\label{eq:almost-eq-to-solve}
\end{equation}
Since $\beta$ is a free parameter, let us set
\begin{equation}
\gamma=\beta\left\Vert r\right\Vert \label{eq:gamma-beta-vars}
\end{equation}
and rewrite~\eqref{eq:almost-eq-to-solve} as follows:
\begin{equation}
\frac{1+\left\Vert r\right\Vert }{2}=\frac{e^{-\gamma}}{e^{-\gamma}+e^{\gamma}}.
\end{equation}
Now let us solve this equation for $\gamma$:
\begin{align}
\frac{1+\left\Vert r\right\Vert }{2} & =\frac{e^{-\gamma}}{e^{-\gamma}+e^{\gamma}}\\
\iff\qquad & 1+\left\Vert r\right\Vert =\frac{2e^{-\gamma}}{e^{-\gamma}+e^{\gamma}}\\
\iff\qquad & \left\Vert r\right\Vert =\frac{2e^{-\gamma}}{e^{-\gamma}+e^{\gamma}}-\frac{e^{-\gamma}+e^{\gamma}}{e^{-\gamma}+e^{\gamma}}\\
\iff\qquad & \left\Vert r\right\Vert =\frac{e^{-\gamma}-e^{\gamma}}{e^{-\gamma}+e^{\gamma}}\\
\iff\qquad & \left\Vert r\right\Vert =-\tanh(\gamma).
\end{align}
This last line implies that
\begin{equation}
\gamma=\arctanh\!\left(-\left\Vert r\right\Vert \right).
\end{equation}
Finally, substituting into~\eqref{eq:gamma-beta-vars}, we find that
\begin{equation}
\beta=\frac{\arctanh\!\left(-\left\Vert r\right\Vert \right)}{\left\Vert r\right\Vert }.
\end{equation}
This concludes our first proof.

We also provide a second, alternative proof based on Taylor expansions.
Consider that the following equality should hold:
\begin{equation}
\frac{1}{2}\left(I+\overrightarrow{r}\cdot\overrightarrow{\sigma}\right)=\frac{\exp\!\left(-\beta\overrightarrow{\mu}\cdot\overrightarrow{\sigma}\right)}{\Tr\!\left[\exp\!\left(-\beta\overrightarrow{\mu}\cdot\overrightarrow{\sigma}\right)\right]},\label{eq:thermal-qubit-eq-to-solve}
\end{equation}
where
\begin{align}
\overrightarrow{r} & \coloneqq\left(r_{x},r_{y},r_{z}\right),\\
\overrightarrow{\sigma} & \coloneqq\left(\sigma_{X},\sigma_{Y},\sigma_{Z}\right),\\
\overrightarrow{\mu} & \coloneqq\left(\mu_{x},\mu_{y},\mu_{z}\right).
\end{align}
Without loss of generality, let us suppose that $\overrightarrow{\mu}$
is normalized, so that $\left\Vert \overrightarrow{\mu}\right\Vert =1$.
Observe that $\left(\overrightarrow{\mu}\cdot\overrightarrow{\sigma}\right)^{2}=I$,
which implies for all $n\in\mathbb{N}$ that
\begin{equation}
\left(\overrightarrow{\mu}\cdot\overrightarrow{\sigma}\right)^{n}=\begin{cases}
I & \text{if \ensuremath{n} is even}\\
\overrightarrow{\mu}\cdot\overrightarrow{\sigma} & \text{else}
\end{cases}.
\end{equation}
Then we find that
\begin{align}
 & \exp\!\left(-\beta\overrightarrow{\mu}\cdot\overrightarrow{\sigma}\right)\nonumber \\
 & =\sum_{n=0}^{\infty}\frac{\left(-\beta\overrightarrow{\mu}\cdot\overrightarrow{\sigma}\right)^{n}}{n!}\\
 & =\sum_{k=0}^{\infty}\frac{\left(-\beta\right)^{2k}\left(\overrightarrow{\mu}\cdot\overrightarrow{\sigma}\right)^{2k}}{\left(2k\right)!}+\sum_{k=0}^{\infty}\frac{\left(-\beta\right)^{2k+1}\left(\overrightarrow{\mu}\cdot\overrightarrow{\sigma}\right)^{2k+1}}{\left(2k+1\right)!}\\
 & =\sum_{k=0}^{\infty}\frac{\left(-\beta\right)^{2k}}{\left(2k\right)!}I+\sum_{k=0}^{\infty}\frac{\left(-\beta\right)^{2k+1}}{\left(2k+1\right)!}\overrightarrow{\mu}\cdot\overrightarrow{\sigma}\\
 & =\cosh(-\beta)I+\sinh(-\beta)\overrightarrow{\mu}\cdot\overrightarrow{\sigma}\\
 & =\cosh(\beta)I-\sinh(\beta)\overrightarrow{\mu}\cdot\overrightarrow{\sigma}.
\end{align}
Then
\begin{align}
& \Tr\!\left[\exp\!\left(-\beta\overrightarrow{\mu}\cdot\overrightarrow{\sigma}\right)\right] \notag \\
& =\Tr[\cosh(\beta)I-\sinh(\beta)\overrightarrow{\mu}\cdot\overrightarrow{\sigma}]\\
 & =\cosh(\beta)\Tr[I]-\sinh(\beta)\Tr[\overrightarrow{\mu}\cdot\overrightarrow{\sigma}]\\
 & =2\cosh(\beta),
\end{align}
which implies that
\begin{equation}
\frac{\exp\!\left(-\beta\overrightarrow{\mu}\cdot\overrightarrow{\sigma}\right)}{\Tr\!\left[\exp\!\left(-\beta\overrightarrow{\mu}\cdot\overrightarrow{\sigma}\right)\right]}=\frac{1}{2}\left(I-\tanh(\beta)\overrightarrow{\mu}\cdot\overrightarrow{\sigma}\right)
\end{equation}
As such, for \eqref{eq:thermal-qubit-eq-to-solve} to be satisfied,
the following should hold:
\begin{align}
\frac{1}{2}\left(I+\overrightarrow{r}\cdot\overrightarrow{\sigma}\right) & =\frac{1}{2}\left(I-\tanh(\beta)\overrightarrow{\mu}\cdot\overrightarrow{\sigma}\right),
\end{align}
which implies that $\overrightarrow{\mu}=-\frac{\overrightarrow{r}}{\left\Vert \overrightarrow{r}\right\Vert }$
and $\left\Vert \overrightarrow{r}\right\Vert =\tanh(\beta)$, the
latter of which is equivalent to $\beta=\arctanh\!\left(\left\Vert \overrightarrow{r}\right\Vert \right)$.
\end{proof}

\begin{rem}
If we would like $\mu$ to be normalized and the temperature to be non-negative, then we could alternatively set
$\mu=- \frac{r}{\left\Vert r\right\Vert }$ and $\beta=\arctanh\!\left(\left\Vert r\right\Vert \right)$, as done in our second proof above,
and arrive at the same conclusion. In this way, $\mu$ exclusively
represents the direction of the Bloch vector, and the inverse temperature
$\beta$ is related only to its norm. This leads to the following
corollary, which gives a function that maps $(\mu,\beta)$ to $r$.
\end{rem}

\begin{cor}
Let $\mu\in\mathbb{R}^{3}$ be the unit vector (i.e., $\left\Vert \mu\right\Vert =1$)
of exponential coordinates and $\beta\geq0$ the inverse temperature
such that a qubit state $\rho$ can be written as in~\eqref{eq:exponential-coord-qubit}.
Then, by choosing
\begin{equation}
r=\tanh(-\beta)\mu,
\end{equation}
the following equality holds:
\begin{multline}
\frac{1}{2}\left(I+r_{x}X+r_{y}Y+r_{z}Z\right)=\\
\frac{1}{Z_{\beta}(\mu)}\exp\!\left(-\beta\left(\mu_{x}X+\mu_{y}Y+\mu_{z}Z\right)\right).
\end{multline}
\end{cor}

\section{Proof of Theorem~\ref{thm:solution-stab-thermo}}

\label{app:proof-solution-stab-thermo}

\begin{proof}
We know from~\eqref{eq:eq:proof-lagrange-mults-thermal}--\eqref{eq:last-step-legendre-duality} and~\eqref{eq:optimality-condition-thermal} that the optimal state for \eqref{eq:SDP-QECC}
has the following form:
\begin{equation}
\omega^{*}\equiv\frac{\exp\!\left(-\frac{1}{T}\left(H-T\sum_{i_{1},\ldots,i_{k}}\lambda_{i_{1},\ldots,i_{k}}\overline{\sigma}_{i_{1}}\cdots\overline{\sigma}_{i_{k}}\right)\right)}{\Tr\!\left[\exp\!\left(-\frac{1}{T}\left(H-T\sum_{i_{1},\ldots,i_{k}}\lambda_{i_{1},\ldots,i_{k}}\overline{\sigma}_{i_{1}}\cdots\overline{\sigma}_{i_{k}}\right)\right)\right]},
\end{equation}
where $\left(\lambda_{i_{1},\ldots,i_{k}}\right)_{i_{1},\ldots,i_{k}\in\left\{ 0,1,2,3\right\} }$
is a tuple of real coefficients. What remains is to determine the
coefficients in $\left(\lambda_{i_{1},\ldots,i_{k}}\right)_{i_{1},\ldots,i_{k}\in\left\{ 0,1,2,3\right\} }$.
To this end, recall that every stabilizer code consists of a Clifford
encoding unitary $U$ such that
\begin{align}
U^{\dag}HU & =-\sum_{i=1}^{n-k}Z_{i},\\
U^{\dag}\overline{\sigma}_{i_{1}}\cdots\overline{\sigma}_{i_{k}}U & =I^{\otimes n-k}\otimes\sigma_{i_{1}}\otimes\cdots\otimes\sigma_{i_{k}}.
\end{align}
Then, after setting
\begin{equation}
Z\equiv\Tr\!\left[\exp\!\left(-\frac{1}{T}\left(H-T\sum_{i_{1},\ldots,i_{k}}\lambda_{i_{1},\ldots,i_{k}}\overline{\sigma}_{i_{1}}\cdots\overline{\sigma}_{i_{k}}\right)\right)\right],
\end{equation}
consider that
\begin{align}
 & \frac{\exp\!\left(-\frac{1}{T}\left(H-T\sum_{i^{k}}\lambda_{i^{k}}\overline{\sigma}_{i_{1}}\cdots\overline{\sigma}_{i_{k}}\right)\right)}{Z}\nonumber \\
 & =\frac{U\exp\!\left(\frac{1}{T}\left(\begin{array}{c}
\sum_{i=1}^{n-k}Z_{i}+\\
T\sum_{i^{k}}\lambda_{i^{k}}I^{\otimes n-k}\otimes\sigma_{i_{1}}\otimes\cdots\otimes\sigma_{i_{k}}
\end{array}\right)\right)U^{\dag}}{Z}\label{eq:encoded-to-unencoded-thermal-state-1}\\
 & =U\left(\frac{\bigotimes_{i=1}^{n-k}e^{\frac{1}{T}Z}\otimes e^{\sum_{i^{k}}\lambda_{i^{k}}\sigma_{i_{1}}\otimes\cdots\otimes\sigma_{i_{k}}}}{Z}\right)U^{\dag},\label{eq:encoded-to-unencoded-thermal-state}
\end{align}
where we have used the abbreviation $i^{k}\equiv i_{1},\ldots,i_{k}$.
Also, we conclude from \eqref{eq:encoded-to-unencoded-thermal-state-1}--\eqref{eq:encoded-to-unencoded-thermal-state}
that
\begin{align}
Z & =\Tr\!\left[\bigotimes_{i=1}^{n-k}e^{\frac{1}{T}Z}\otimes e^{\sum_{i^{k}}\lambda_{i^{k}}\sigma_{i_{1}}\otimes\cdots\otimes\sigma_{i_{k}}}\right]\\
 & =\left(\Tr\!\left[e^{\frac{1}{T}Z}\right]\right)^{n-k}\Tr\!\left[e^{\sum_{i^{k}}\lambda_{i^{k}}\sigma_{i_{1}}\otimes\cdots\otimes\sigma_{i_{k}}}\right].\label{eq:partition-func-stab-split}
\end{align}
Now observe that
\begin{align}
 & \Tr\!\left[\left(\overline{\sigma}_{i_{1}}\cdots\overline{\sigma}_{i_{k}}\right)\omega^{*}\right]\nonumber \\
 & =\Tr\!\left[U^{\dag}\overline{\sigma}_{i_{1}}\cdots\overline{\sigma}_{i_{k}}UU^{\dag}\omega^{*}U\right]\label{eq:logical-constraints-rewrite}\\
 & =\Tr\!\left[\left(I^{\otimes n-k}\otimes\sigma_{i_{1}}\otimes\cdots\otimes\sigma_{i_{k}}\right)U^{\dag}\omega^{*}U\right]\\
 & =\Tr\!\left[\left(\sigma_{i_{1}}\otimes\cdots\otimes\sigma_{i_{k}}\right)\Tr_{n-k}\!\left[U^{\dag}\omega^{*}U\right]\right].\label{eq:logical-constraints-rewrite-last}
\end{align}
The constraints in \eqref{eq:SDP-QECC} and the development in \eqref{eq:logical-constraints-rewrite}--\eqref{eq:logical-constraints-rewrite-last}
imply that
\begin{equation}
\Tr_{n-k}\!\left[U^{\dag}\omega^{*}U\right]=\rho,\label{eq:constraint-on-marginal-encoded-state}
\end{equation}
where $\rho$ is defined by \eqref{eq:k-qubit-constraints-for-state}
and \eqref{eq:unencoded-state-thermal-form}. However, \eqref{eq:encoded-to-unencoded-thermal-state}
and \eqref{eq:partition-func-stab-split} imply that
\begin{align}
 & \Tr_{n-k}\!\left[U^{\dag}\omega^{*}U\right]\nonumber \\
 & =\Tr_{n-k}\!\left[\frac{e^{\frac{1}{T}Z_{1}}\otimes\cdots\otimes e^{\frac{1}{T}Z_{n-k}}\otimes e^{\sum_{i^{k}}\lambda_{i^{k}}\sigma_{i_{1}}\otimes\cdots\otimes\sigma_{i_{k}}}}{\left(\Tr[e^{\frac{1}{T}Z}]\right)^{n-k}\Tr\!\left[e^{\sum_{i^{k}}\lambda_{i^{k}}\sigma_{i_{1}}\otimes\cdots\otimes\sigma_{i_{k}}}\right]}\right]\label{eq:marg-state-omega-1}\\
 & =\frac{e^{\sum_{i_{1},\ldots,i_{k}}\lambda_{i_{1},\ldots,i_{k}}\sigma_{i_{1}}\otimes\cdots\otimes\sigma_{i_{k}}}}{\Tr\!\left[e^{\sum_{i_{1},\ldots,i_{k}}\lambda_{i_{1},\ldots,i_{k}}\sigma_{i_{1}}\otimes\cdots\otimes\sigma_{i_{k}}}\right]}.\label{eq:marg-state-omega-2}
\end{align}
Given the form in \eqref{eq:unencoded-state-thermal-form} and the
equalities in \eqref{eq:marg-state-omega-1}--\eqref{eq:marg-state-omega-2},
the only way that the equality in \eqref{eq:constraint-on-marginal-encoded-state}
can hold is if $\lambda_{i_{1},\ldots,i_{k}}=\mu_{i_{1},\ldots,i_{k}}$
for all $i_{1},\ldots,i_{k}\in\left\{ 0,1,2,3\right\} $, thus concluding
the proof.
\end{proof}

\begin{rem}
\label{rem:stab-thermo-encoding}
The development in \eqref{eq:encoded-to-unencoded-thermal-state-1}--\eqref{eq:partition-func-stab-split} indicates
that there is a simple method of encoding the state of a stabilizer
thermodynamic system. Indeed, the state has the following form:
\begin{equation}
U\left(\bigotimes_{i=1}^{n-k}\frac{e^{\frac{1}{T}Z}}{\Tr[e^{\frac{1}{T}Z}]}\otimes\frac{e^{\sum_{i^{k}}\lambda_{i^{k}}\sigma_{i_{1}}\otimes\cdots\otimes\sigma_{i_{k}}}}{\Tr\!\left[e^{\sum_{i^{k}}\lambda_{i^{k}}\sigma_{i_{1}}\otimes\cdots\otimes\sigma_{i_{k}}}\right]}\right)U^{\dag},
\end{equation}
where we have again used the abbreviation $i^{k}\equiv i_{1},\ldots,i_{k}$.
This implies that it can be encoded by preparing the $n-k$ ancilla
qubits each in temperature-$T$ thermal states of the Hamiltonian $-Z$ (which are
approximations of the ideal state $|0\rangle\!\langle0|$) and the $k$
logical qubits in the state
\begin{equation}
\frac{e^{\sum_{i^{k}}\lambda_{i^{k}}\sigma_{i_{1}}\otimes\cdots\otimes\sigma_{i_{k}}}}{\Tr\!\left[e^{\sum_{i^{k}}\lambda_{i^{k}}\sigma_{i_{1}}\otimes\cdots\otimes\sigma_{i_{k}}}\right]}.
\end{equation}
After doing so, one then applies the Clifford encoding unitary $U$
for the stabilizer code.
\end{rem}

\section{Other simulation results}

\label{sec:Other-simulation-results}

The technical details of the algorithmic implementations and simulation setup used for the results presented in this appendix are similar to those described in Section~\ref{sec:Simulation-results} of the main text. We therefore refer the reader to that section for explanations of implementation-specific aspects of the classical and HQC algorithms.

\subsection{One-dimensional quantum Heisenberg model with nearest-neighbor interactions}

\begin{figure*}
\includegraphics[width=\linewidth]{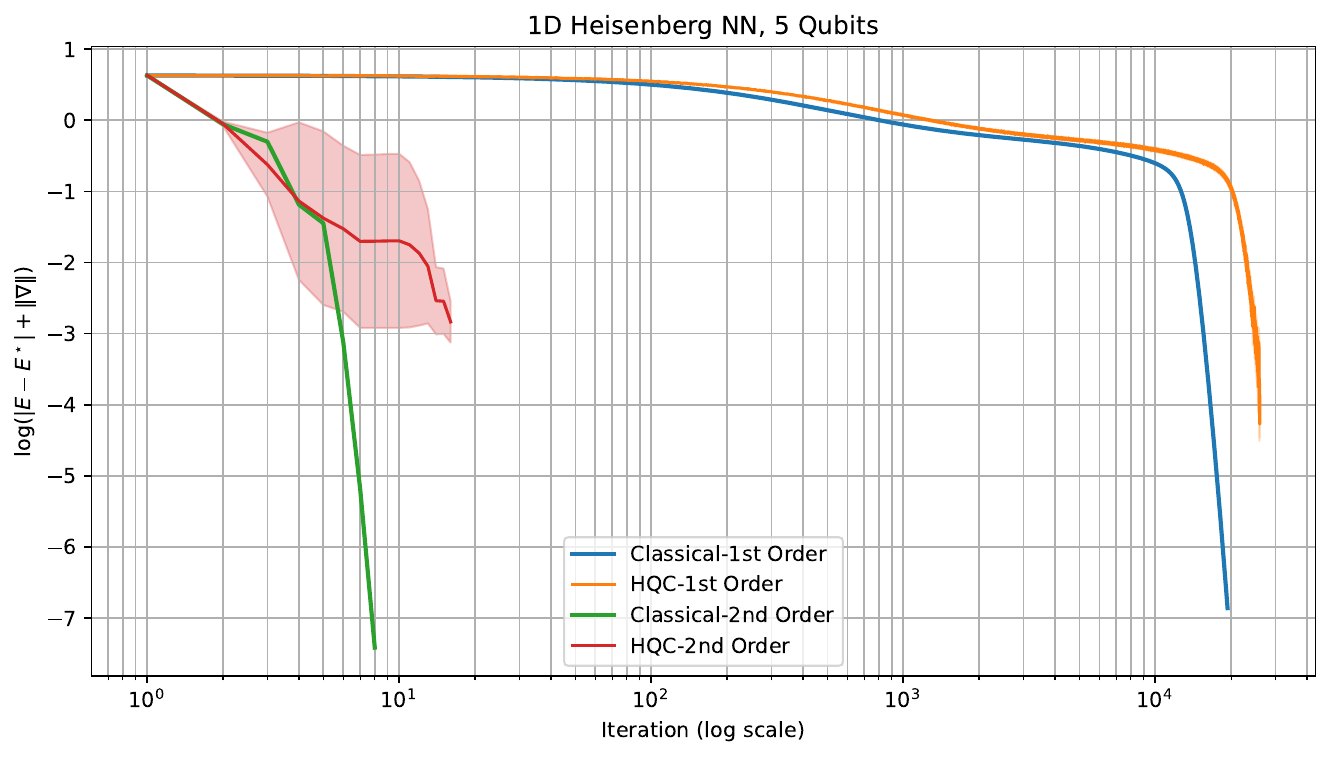}
\caption{The figure depicts the logarithm of the error metric in~\eqref{eq:error-metric} versus the number of iterations, for the first-order classical and HQC algorithms and the second-order classical algorithm for the task of constrained energy minimization for the five-qubit, one-dimensional quantum Heisenberg model with nearest-neighbor interactions and constraints on the total magnetizations in the $x$, $y$, and $z$ directions set to be 1, 0, and 1, respectively.  All of the algorithms converge, but the HQC algorithms, shown as the average over five independent runs with shaded regions denoting one standard deviation,  require more iterations to converge due to sampling noise inherent in them.
}
\label{fig:1D-heis-nn-Log-Error}
\end{figure*}

Figure~\ref{fig:1D-heis-nn-Log-Error} depicts the performance of the first-order classical and HQC algorithms and the second-order classical algorithm for the five-qubit, one-dimensional quantum Heisenberg model with nearest-neighbor interactions and constraints on the total magnetizations in the $x$, $y$, and $z$ directions set to be 1, 0, and 1, respectively.

\subsection{Two-dimensional quantum Heisenberg model with nearest-neighbor interactions}

\begin{figure*}
\includegraphics[width=\linewidth]{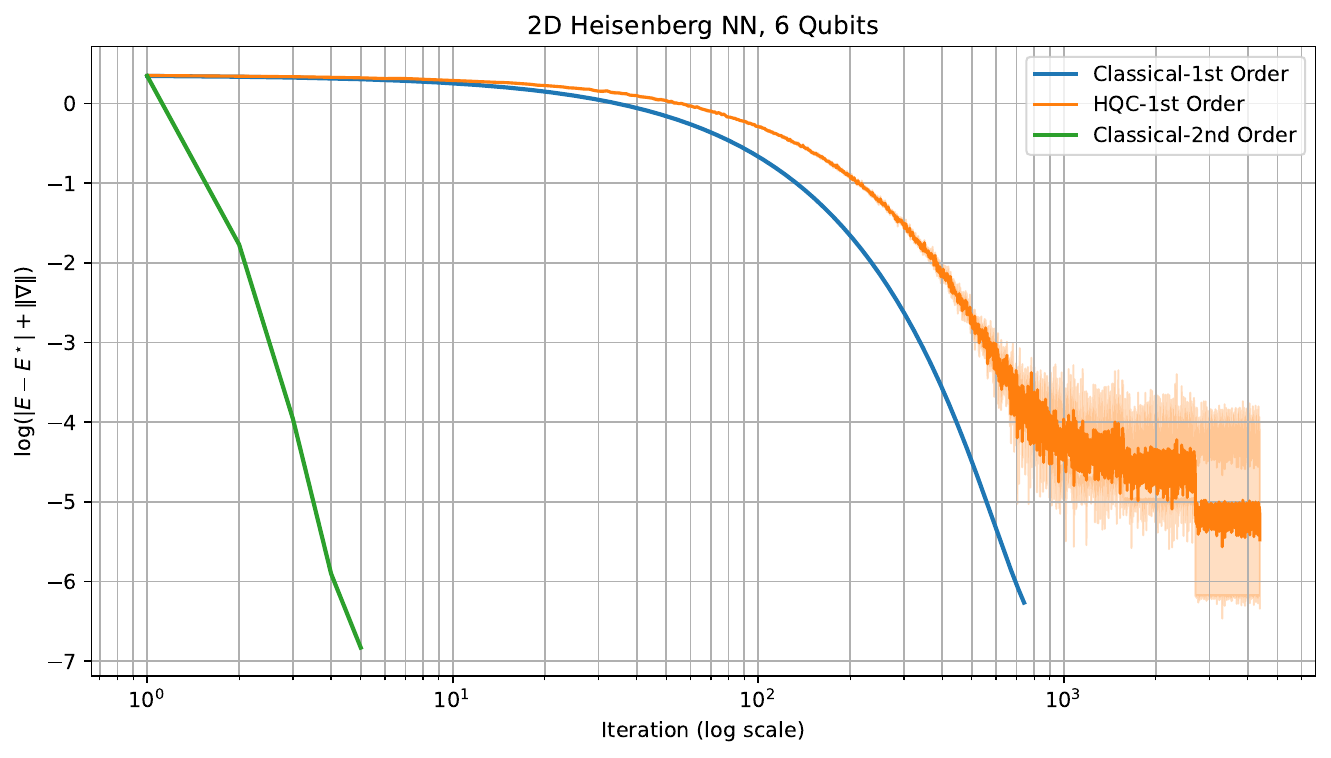}
\caption{The figure depicts the logarithm of the error metric in~\eqref{eq:error-metric} versus the logarithm of the number of iterations, for the first-order classical and HQC LMPW algorithms for the task of constrained energy minimization for the six-qubit, two-dimensional quantum Heisenberg model with nearest-neighbor interactions and constraints on the total magnetizations in the $x$, $y$, and $z$ directions set to be 1, 0, and 1, respectively.  Both algorithms converge, but the HQC algorithm, shown as the average over five independent runs with the shaded region denoting one standard deviation, requires more iterations to converge due to sampling noise inherent in it.
}
\label{fig:2D-heis-nn-Log-Error}
\end{figure*}

Figure~\ref{fig:2D-heis-nn-Log-Error} depicts the performance of the first-order classical and HQC algorithms and the second-order classical algorithm for the six-qubit, two-dimensional quantum Heisenberg model with nearest-neighbor interactions and constraints on the total magnetizations in the $x$, $y$, and $z$ directions set to be 1, 0, and 1, respectively. Given the substantial computational cost associated with Hessian estimation in the HQC setting, we did not include the second-order HQC algorithm in this simulation.

\subsection{One-to-three-qubit repetition code}

The one-to-three qubit repetition code was the earliest discovered
quantum error-correcting code~\cite{Peres1985}. We presented it as
a stabilizer thermodynamic system in Example~\ref{exa:repetition}.

Figure~\ref{fig:rep-code-warm-start-compare-classical-1st-order} plots the performance of the LMPW classical first-order algorithm with and without a warm start, showing that warm starting leads to immediate convergence in just one iteration. Figure~\ref{fig:rep-code-warm-start-compare-HQC-1st-order} presents a similar plot for the LMPW HQC first-order algorithm, which likewise benefits from a warm start by converging in a single iteration.

\begin{figure}
\includegraphics[width=\linewidth]{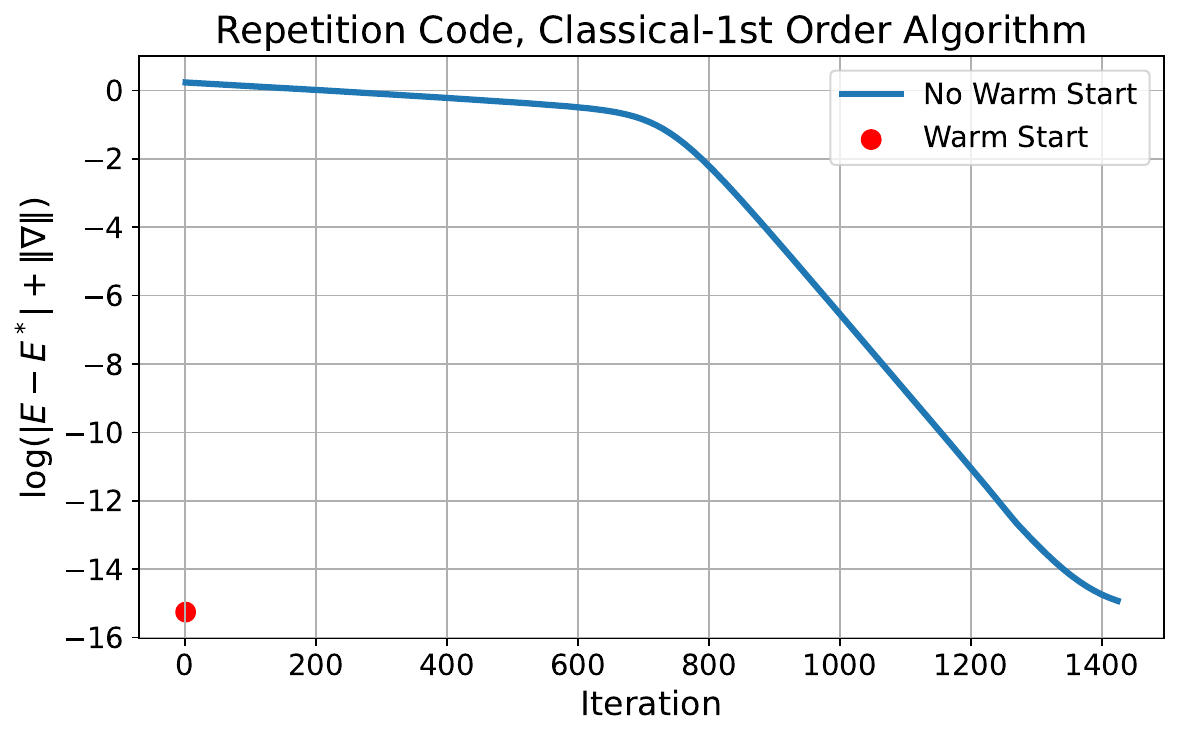}
\caption{The figure depicts the logarithm of the error metric in~\eqref{eq:error-metric} versus the number of iterations, for the LMPW classical first-order algorithm for the task of constrained energy minimization for the one-to-three-qubit repetition code. In this case, the plot compares the performance with and without a warm start for the algorithm, as discussed in Section~\ref{sec:warm-start}. The warm-started algorithm converges in a single iteration.
}
\label{fig:rep-code-warm-start-compare-classical-1st-order}
\end{figure}

\begin{figure}
\includegraphics[width=\linewidth]{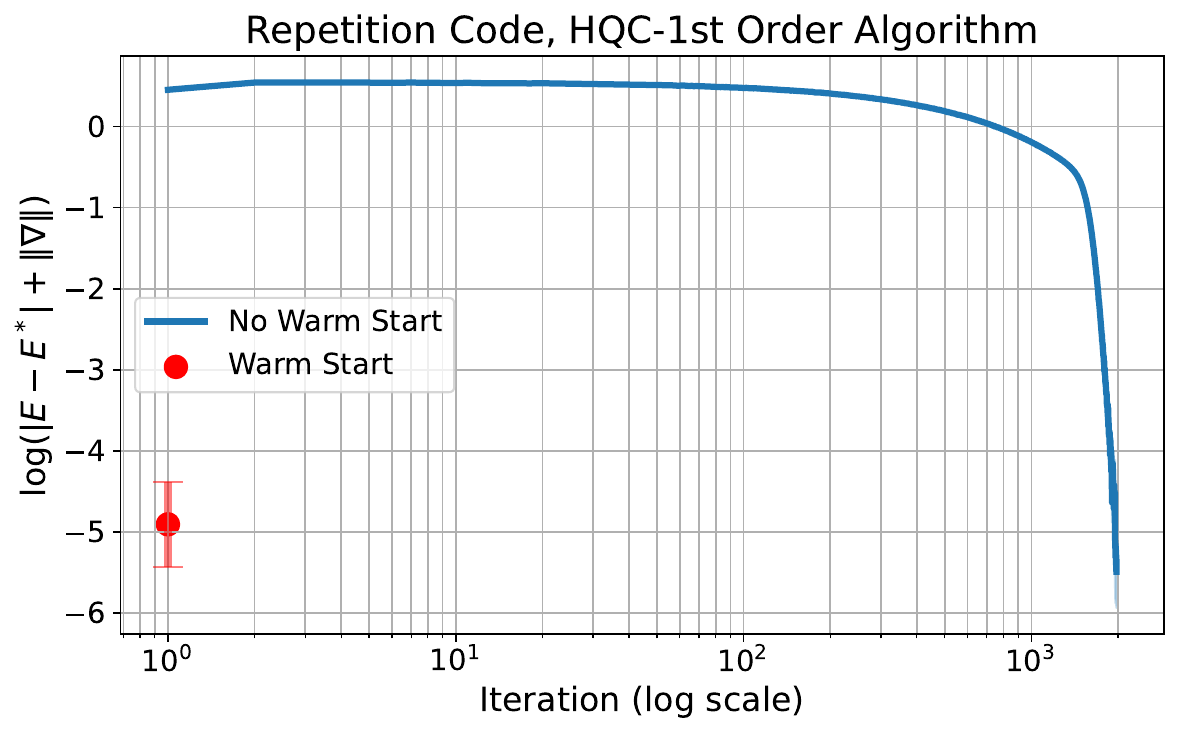}
\caption{The figure depicts the average of the logarithm of the error metric in~\eqref{eq:error-metric} over five independent runs, plotted against the number of iterations, with shaded regions indicating one standard deviation, for the LMPW HQC first-order algorithm for the task of constrained energy minimization for the one-to-three-qubit repetition code. In this case, the plot compares the performance with and without a warm start for the algorithm, as discussed in Section~\ref{sec:warm-start}. The warm-started algorithm converges in a single iteration, and the error bars reflect the standard deviation across the five runs.
}
\label{fig:rep-code-warm-start-compare-HQC-1st-order}
\end{figure}

\subsection{Two-to-four-qubit quantum error-detecting code}

The two-to-four qubit quantum error-detecting code encodes two logical
qubits into four physical qubits and can detect an arbitrary single-qubit
error~\cite{Vaidman1996,Grassl1997}. Its stabilizer generators are
as follows \cite[Section~8.1]{Gottesman1997}:
\begin{align}
S_{1} & \coloneqq X\otimes X\otimes X\otimes X,\\
S_{2} & \coloneqq Z\otimes Z\otimes Z\otimes Z,
\end{align}
and logical operators for it are as follows \cite[Eqs.~(8.1)--(8.2)]{Gottesman1997}:
\begin{align}
\overline{X}_{1} & \coloneqq X\otimes X\otimes I\otimes I,\\
\overline{Z}_{1} & \coloneqq I\otimes Z\otimes I\otimes Z,\\
\overline{X}_{2} & \coloneqq X\otimes I\otimes X\otimes I,\\
\overline{Z}_{2} & \coloneqq I\otimes I\otimes Z\otimes Z.
\end{align}
From the above, we deduce that
\begin{align}
\overline{Y}_{1} & =i\overline{X}_{1}\overline{Z}_{1}=X\otimes Y\otimes I\otimes Z,\\
\overline{Y}_{2} & =i\overline{X}_{2}\overline{Z}_{2}=X\otimes I\otimes Y\otimes Z,
\end{align}
Here, we test the performance of the LMPW algorithms for encoding
a Bell state $|\Phi^{+}\rangle\!\langle\Phi^{+}|$, where $|\Phi^{+}\rangle\coloneqq\frac{1}{\sqrt{2}}\left(|00\rangle+|11\rangle\right)$,
into the four-qubit error-detecting code. The representation of this
state in the Pauli basis is as follows:
\begin{equation}
|\Phi^{+}\rangle\!\langle\Phi^{+}|=\frac{1}{4}(I\otimes I+X\otimes X-Y\otimes Y+Z\otimes Z).\label{eq:bell-state-pauli-exp}
\end{equation}
Thus, we take the non-commuting charges to be all possible logical
operators formed from $\overline{X}_{1}$, $\overline{Z}_{1}$, $\overline{X}_{2}$,
and $\overline{Z}_{2}$, and the constraints to be as follows:
\begin{align}
\Tr[\overline{X}_{2}\rho] & =0,\qquad\Tr[\overline{Y}_{2}\rho]=0,\\
\Tr[\overline{Z}_{2}\rho] & =0,\qquad\Tr[\overline{X}_{1}\rho]=0,\\
\Tr[\overline{X}_{1}\overline{X}_{2}\rho] & =1,\qquad\Tr[\overline{X}_{1}\overline{Y}_{2}\rho]=0,\\
\Tr[\overline{X}_{1}\overline{Z}_{2}\rho] & =0,\qquad\Tr[\overline{Y}_{1}\rho]=0,\\
\Tr[\overline{Y}_{1}\overline{X}_{2}\rho] & =0,\qquad\Tr[\overline{Y}_{1}\overline{Y}_{2}\rho]=-1,\\
\Tr[\overline{Y}_{1}\overline{Z}_{2}\rho] & =0,\qquad\Tr[\overline{Z}_{1}\rho]=0,\\
\Tr[\overline{Z}_{1}\overline{X}_{2}\rho] & =0,\qquad\Tr[\overline{Z}_{1}\overline{Y}_{2}\rho]=0,\\
\Tr[\overline{Z}_{1}\overline{Z}_{2}\rho] & =1.
\end{align}
These constraints are chosen in accordance with the representation
in~\eqref{eq:bell-state-pauli-exp}, in order to ensure that the encoded
two-qubit state is indeed a Bell state.
Figure~\ref{fig:ErrorDetectingCode_2to4q} plots the performance of the LMPW algorithms for this constrained energy minimization problem.

\begin{figure*}
\includegraphics[width=\linewidth]{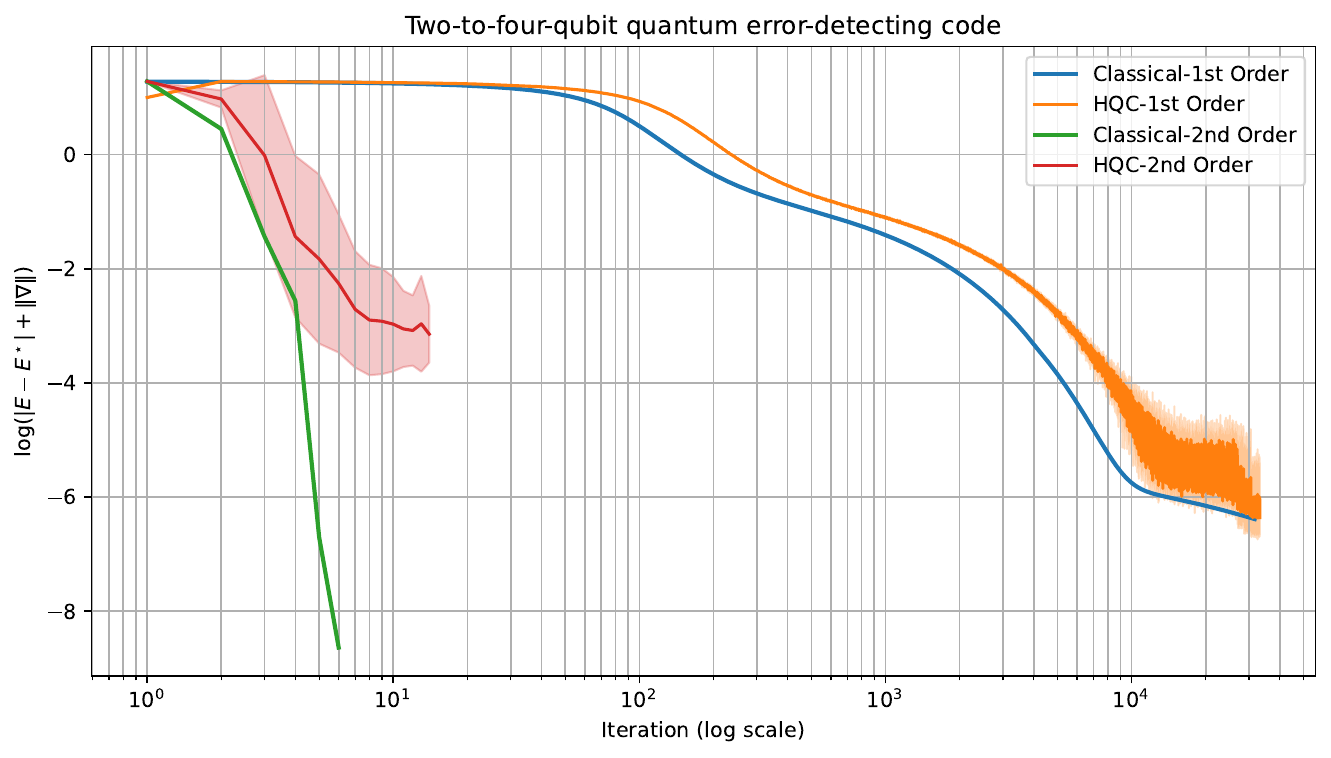}
\caption{The figure depicts the logarithm of the error metric in~\eqref{eq:error-metric} versus the number of iterations, for the first-order classical and HQC algorithms and the second-order classical algorithm for the task of constrained energy minimization for the two-to-four-qubit quantum error-detecting code.  All of the algorithms converge, but the HQC algorithms, shown as the average over five independent runs with the shaded region denoting one standard deviation, require more iterations to converge due to sampling noise inherent in them.
}
\label{fig:ErrorDetectingCode_2to4q}
\end{figure*}

\end{document}